\documentclass[10pt,a4paper]{article}
\usepackage{fullpage,graphicx,amssymb,amsmath,xspace,amsfonts,tabularx,rotating,url,float,color}
\usepackage[vlined,linesnumbered,ruled,algosection]{algorithm2e}
\usepackage{cyrillic}
\newcommand{\cyrrm}{\fontencoding{OT2}\selectfont\textcyrup}

\newcommand{\Reals}{\mathbb{R}}

\newcommand{\eps}{\varepsilon}
\def\andfrac#1/#2{%
   \leavevmode\kern.1em
   \raise.5ex\hbox{\the\scriptfont0 #1}\kern-.1em
   /\kern-.15em\lower.25ex\hbox{\the\scriptfont0 #2}}

\def\volume{\ensuremath{\mathrm{vol}}\xspace}
\def\relvolume{\ensuremath{\mathrm{fraction}}\xspace}
\def\succ#1{\ensuremath{\hat{#1}}\xspace}
\def\pred#1{\ensuremath{\check{#1}}\xspace}
\def\base{\ensuremath{b}}
\def\depth{\ensuremath{\ell}}
\def\TRGC{\ensuremath{3\textsc{rgc}}}
\def\BRGC{\ensuremath{2\textsc{rgc}}}
\def\kRGC{\ensuremath{\base\textsc{rgc}}}
\def\reps#1#2{\ensuremath{(#1)_{#2}}}
\def\takeout#1{\ensuremath{\begin{smallmatrix}#1\\\times\end{smallmatrix}}}
\def\reduction#1#2{\ensuremath{\begin{smallmatrix}#1=#2\\\times\end{smallmatrix}}}
\def\mirror#1{\ensuremath{\begin{smallmatrix}#1\\\ominus\end{smallmatrix}}}
\def\putin#1#2{\ensuremath{\begin{smallmatrix}#1=#2\\\wedge\end{smallmatrix}}}
\def\lift#1{\ensuremath{\mathrm{lift}_{#1}}\xspace}
\def\diaglift{\ensuremath{\mathrm{double}}\xspace}
\def\inv#1{\ensuremath{\overline{#1}}\xspace}
\def\identity{\ensuremath{id}\xspace}
\def\halffix{first-half\xspace}
\def\Halffix{First-half\xspace}
\def\lefthalf#1{\ensuremath{#1^{\dashv}}}
\def\righthalf#1{\ensuremath{#1^{\vdash}}}

\makeatletter
\newcommand\floatc@mybox[2]{\vbox{\hbadness10000
\moveleft3.4pt\vbox{\advance\hsize by6.8pt
\hrule \hbox to\hsize{\vrule\kern3pt
\vbox{\kern3pt\vbox{\advance\hsize by-6.8pt{\@fs@cfont #1} #2}\kern3pt}\kern3pt\vrule}}}}%
\newcommand\fs@mybox{\def\@fs@cfont{\bfseries}\let\@fs@capt\floatc@mybox
\def\@fs@pre{\setbox\@currbox\vbox{\hbadness10000
\moveleft3.4pt\vbox{\advance\hsize by6.8pt
\hrule \hbox to\hsize{\vrule\kern3pt
\vbox{\kern4.5pt\box\@currbox\kern4.5pt}\kern3pt\vrule}\hrule}}}%
\def\@fs@mid{}%
\def\@fs@post{}%
\let\@fs@iftopcapt\iftrue}
\makeatother
\floatstyle{mybox}
\newfloat{inset}{tb}{ins}
\floatname{inset}{Inset}

\newtheorem{lemma}{Lemma}
\newtheorem{corollary}{Corollary}
\newtheorem{theorem}{Theorem}
\newtheorem{definition}{Definition}
\newenvironment{proof}{Proof:}{\qed}
\def\squareforqed{\hbox{\rlap{$\sqcap$}$\sqcup$}}
\def\qed{\ifmmode\squareforqed\else{\unskip\nobreak\hfil
\penalty50\hskip1em\null\nobreak\hfil\squareforqed
\parfillskip=0pt\finalhyphendemerits=0\endgraf}\fi}

\SetFuncSty{textsc}
\SetDataSty{textit}
\SetCommentSty{textrm}
\SetKwFunction{Extract}{ExtractFirstDigit}
\SetKwData{Forward}{forward}
\SetKwData{Direction}{direction}
\SetKwData{Permutation}{axis}
\SetKwData{AltPermutation}{altAxis}
\SetKwData{Rotation}{rotation}
\SetKwData{NewRotation}{newRotation}
\SetKwData{Reverse}{reverse}
\SetKwData{Reflected}{reflected}
\SetKwData{NumberOfTwos}{numberOfTwos}
\SetKwData{pFirstDigit}{pDigit}
\SetKwData{qFirstDigit}{qDigit}
\SetKwData{iMinusOne}{iMinusOne}
\SetKwData{Rank}{rank}
\SetKw{True}{true}
\SetKw{False}{false}
\DontPrintSemicolon
\SetKwData{Empty}{empty}
\newcommand\Not{\ensuremath{\neg}}

\SetKw{XOR}{XOR}

\newcommand{\XSays}[3]{{\color{#2}
      {$\rule[-0.12cm]{0.2in}{0.5cm}$\fbox{\tt
            #1:} }%
      \itshape #3
      \marginpar{\color{#2}\tt #1}%
      \def\comment{#3}\def\empty{}\ifx\comment\empty\else
      {$\rule[0.1cm]{0.3in}{0.1cm}$\fbox{\tt
            end}$\rule[0.1cm]{0.3in}{0.1cm}$} \fi
   }%
}

\begin{document}
\title{Harmonious Hilbert curves\\and other extradimensional space-filling curves\footnote{This manuscript has been put on arXiv.org for early dissemination of the results. I have not found the time (yet) to prepare this manuscript (or parts of it) for submission to be published in a journal or a conference. I welcome any corrections and other comments which could help me to improve this manuscript and/or to disseminate the results further: please mail me at cs.herman@haverkort.net.}}
\author{%
Herman~Haverkort\thanks{Dept.\ of Computer Science, Eindhoven University of Technology, the Netherlands, cs.herman@haverkort.net}
}
\maketitle

\begin{abstract}
This paper introduces a new way of generalizing Hilbert's two-dimensional space-filling curve to arbitrary dimensions. The new curves, called \emph{harmonious Hilbert curves}, have the unique property that for any $d' < d$, the $d$-dimensional curve is compatible with the $d'$-dimensional curve with respect to the order in which the curves visit the points of any $d'$-dimensional axis-parallel space that contains the origin. Similar generalizations to arbitrary dimensions are described for several variants of Peano's curve (the original Peano curve, the coil curve, the half-coil curve, and the Meurthe curve). The $d$-dimensional harmonious Hilbert curves and the Meurthe curves have neutral orientation: as compared to the curve as a whole, arbitrary pieces of the curve have each of $d!$ possible rotations with equal probability. Thus one could say these curves are `statistically invariant' under rotation---unlike the Peano curves, the coil curves, the half-coil curves, and the familiar generalization of Hilbert curves by Butz and Moore.

In addition, prompted by an application in the construction of R-trees, this paper shows how to construct a $2d$-dimensional generalized Hilbert or Peano curve that traverses the points of a certain $d$-dimensional diagonally placed subspace in the order of a given $d$-dimensional generalized Hilbert or Peano curve.

Pseudocode is provided for comparison operators based on the curves presented in this paper.
\end{abstract}

\section{Introduction}
\label{sec:intro}

\paragraph{Space-filling curves}
A space-filling curve in $d$ dimensions is a continuous, surjective mapping from $\Reals$ to $\Reals^d$.
In the late 19th century Peano~\cite{Peano} described such mappings for $d = 2$ and $d = 3$. Since then, quite a number of space-filling curves have appeared in the literature, and space-filling curves have been applied in diverse areas such as spatial databases, load balancing in parallel computing, improving cache utilization in computations on large matrices, finite element methods, image compression, and combinatorial optimization~\cite{Bader}.

Space-filling curves are usually recursive constructions that map the unit interval $[0,1]$ to a subset of $\Reals^d$ that has measure larger than zero. In the case of Peano's two-dimensional curve, the subset of $\Reals^d$ that is `filled' is the unit square. Peano's curve is based on subdividing this unit square into a grid of $3 \times 3$ square cells, and simultaneously subdividing the unit interval into nine subintervals. Each subinterval is then matched to a cell; thus Peano's curve traverses the cells one by one in a particular order. The mapping from unit interval to unit square is refined by applying the procedure recursively to each subinterval-cell pair, so that within each cell, the curve makes a similar traversal (see Figure~\ref{fig:peano2d}(a)). The result is a fully-specified mapping from the unit interval to the unit square. A mapping from $\Reals$ to $\Reals^2$ could then be constructed by inverting the recursion, recursively considering the unit interval and the unit square as a subinterval and a cell of a larger interval and a larger square. A higher-dimensional version would be based on subdividing a $d$-dimensional hypercube into a grid of $3^d$ hypercubic cells.

Peano's curve has been the curve of choice for certain applications. However, in many other applications preference is given to curves based on subdividing squares (or hypercubes) into only $2^d$ squares (or hypercubes). The cell in which a given point $p$ lies can then be determined by inspecting the binary representations of the coordinates of $p$ bit by bit, and no divisions by three need to be computed.
\begin{figure}
\centering
\hbox to \hsize{\hfill
(a)~\includegraphics[width=0.4\hsize]{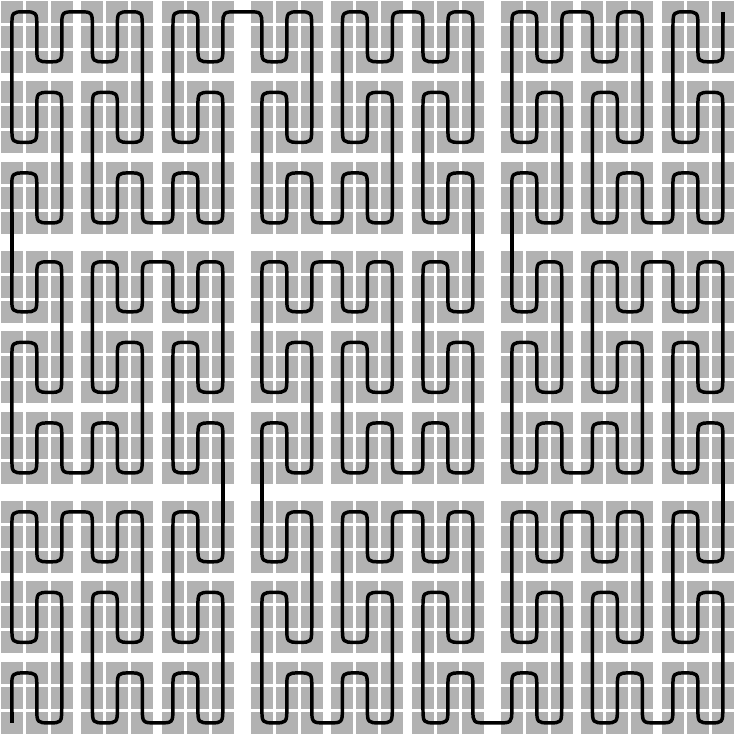}\hfill
(b)~\includegraphics[width=0.4\hsize]{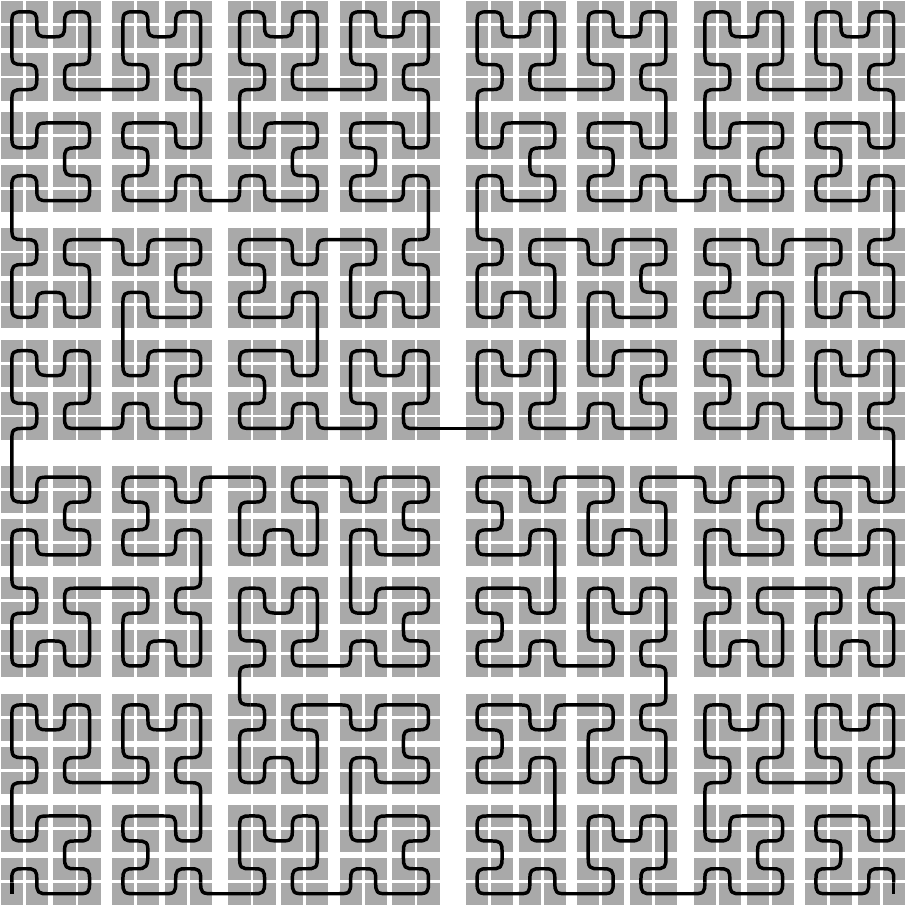}\hfill
}
\caption{(a) A sketch of Peano's space-filling curve.\quad (b) A sketch of Hilbert's space-filling curve.}
\label{fig:peano2d}\label{fig:hilbert2d}
\end{figure}
In response to Peano's publication, Hilbert~\cite{Hilbert} described a two-dimensional space-filling curve based on subdividing a square into four squares (Figure~\ref{fig:hilbert2d}(b)). However, he did not describe how to do something similar in higher dimensions. It is now well-established that it can be done; in fact there are many ways to define a three-dimensional curve based on subdividing a cube into eight octants~\cite{Alber,Haverkort3D}.

\paragraph{The goal: extradimensional space-filling curves}
In this paper we consider different ways of generalizing Hilbert curves, Peano curves and similar curves to higher dimensions. Let $P$ be a set of points on a $d'$-dimensional facet $U'$ of a $d$-dimensional unit hypercube $U$, such that the origin is a vertex of $U$ and $U'$. In this paper, we study generalizations of space-filling curves that have the following property: regardless of our choice of $P$ and $U'$, the order in which the points of $P$ are visited by the $d$-dimensional curve that fills $U$, is the same as the order in which these points are visited by the $d'$-dimensional curve that fills $U'$.

\begin{figure}
\centering
\hbox to \hsize{\hfill
\includegraphics[page=1,scale=1.15]{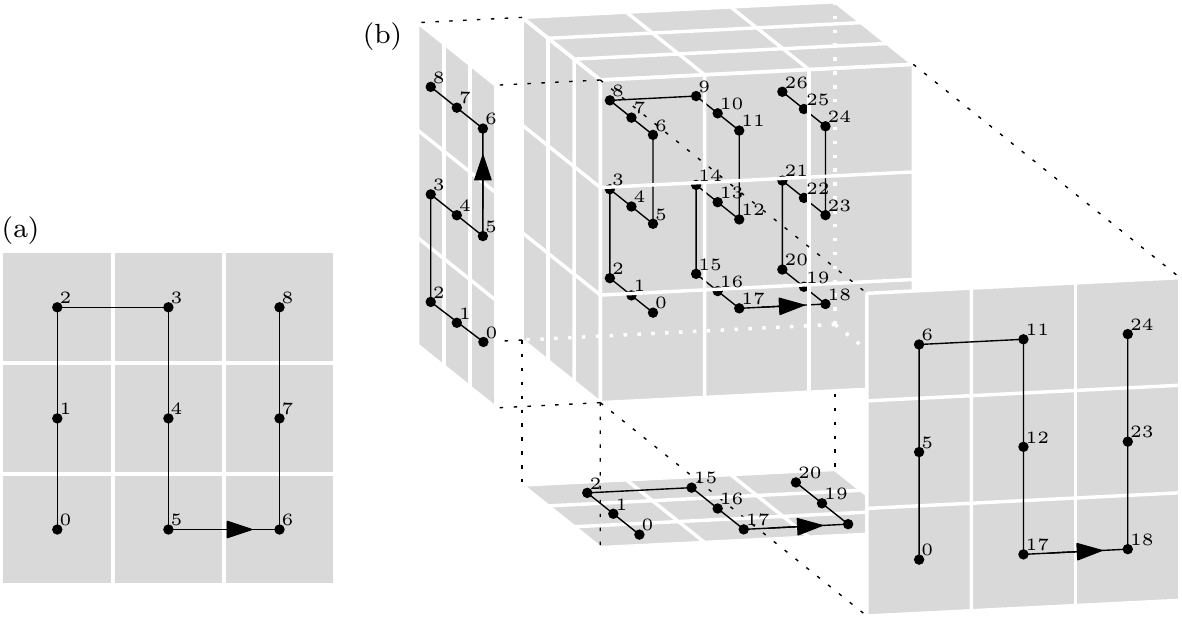}\hfill
}
\caption{The Peano curve in two and three dimensions. (a) Ordering of nine squares in 2D.\quad\break (b)~Ordering of 27 cubes in 3D, with the induced orderings on the front, left, and bottom side.}
\label{fig:peano-level1-visible-orders}
\end{figure}

\begin{figure}
\centering
\hbox to \hsize{\hfill
\includegraphics[page=2,scale=1.15]{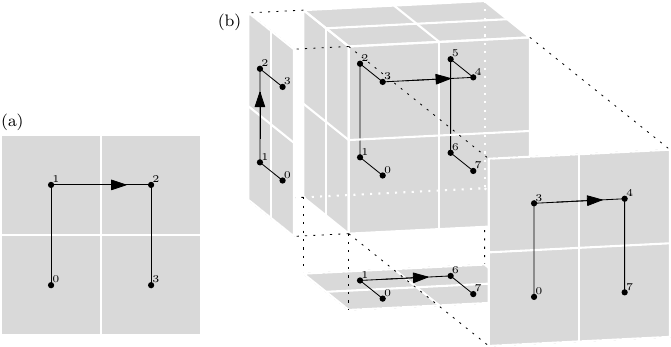}\hfill
}
\caption{The Hilbert curve and the three-dimensional version of Butz-Moore. (a) Ordering of four squares in 2D.\quad (b) Ordering of eight cubes in 3D, with the induced orderings on the front, left, and bottom side.}
\label{fig:butzmoore-level1-visible-orders}
\end{figure}

\begin{figure}
\centering
\hbox to \hsize{\hfill
\includegraphics[page=3,scale=1.15]{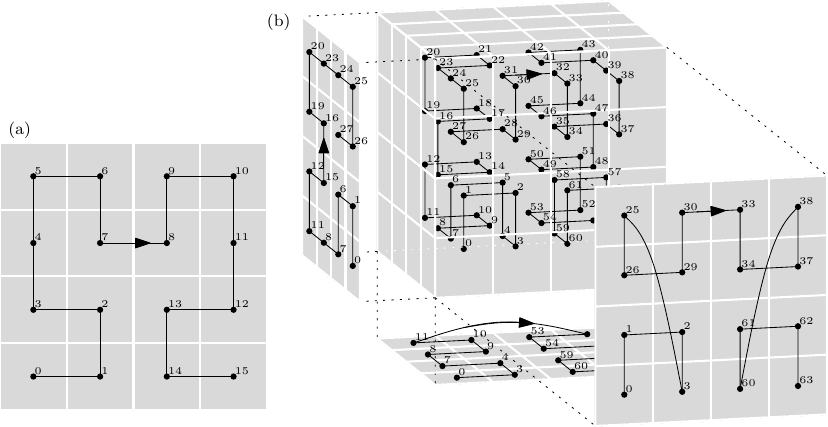}\hfill
}
\caption{The Hilbert curve and the three-dimensional version of Butz-Moore. (a) Ordering of 16 squares in 2D.\quad (b) Ordering of 64 cubes in 3D, with the induced orderings on the front, left, and bottom side.}
\label{fig:butzmoore-level2-visible-orders}
\end{figure}

To get some intuition for what this means, consider Figure~\ref{fig:peano-level1-visible-orders}. Figure~\ref{fig:peano-level1-visible-orders}(a) shows the order in which the two-dimensional Peano curve traverses the nine subsquares of the unit square, and Figure~\ref{fig:peano-level1-visible-orders}(b) shows the order in which a three-dimensional Peano curve traverses the 27 subcubes of the unit cube. If we consider only the subcubes that appear on the front, left, or bottom face of the unit cube, then we see that on each face, these subcubes are visited in the order of the two-dimensional Peano curve. Figure~\ref{fig:butzmoore-level1-visible-orders}(a) shows the order in which the two-dimensional Hilbert curve traverses the four subsquares of the unit square, and Figure~\ref{fig:butzmoore-level1-visible-orders}(b) shows the order in which the three-dimensional generalization by Moore~\cite{Moore} (citing Butz~\cite{Butz}) traverses the eight subcubes of the unit cube. Again, we see that on each face, the subcubes are visited in the order of the two-dimensional curve. However, from these observations one cannot conclude that all points of the front, left or bottom face are visited in the order of the two-dimensional curve. Figure~\ref{fig:butzmoore-level2-visible-orders} shows the next level of refinement of the two-dimensional Hilbert curve and the three-dimensional Butz-Moore curve. The order in which the subcubes of the front, left or bottom face are visited no longer corresponds to the two-dimensional curve. Our goal in this paper is to identify generalizations of the Hilbert curve and several variants of the Peano curve, in which the correspondence is maintained for any level of refinement of the order in which squares and cubes are visited.

To make this statement more precise, we should first define the order in which points appear along the curve more precisely. Let $f$ be a space-filling curve, or more precisely, a continuous surjective function from $[0,1]$ to a unit hypercube $U$. Let a \emph{valid ordering} $\inv{f}$ of $f$ be any bijection from the unit hypercube $U$ to $[0,1]$, such that for any point $p \in P$ we have $f(\inv{f}(p)) = p$.

Let $\mu: \{0,...,d'-1\} \rightarrow \{0,...,d-1\}$, $d > d'$ be a strictly increasing injection, that is, for any $i, j \in \{0,...,d'-1\}$ with $i < j$ we have $\mu(i) < \mu(j)$. This function $\mu$ will be used as a function that selects $d'$ dimensions out of $d$, in the following way: given a $d'$-dimensional point $p = (p_0,...,p_{d'-1})$, we define $\lift\mu(p)$ as the $d$-dimensional point $(q_0,...,q_{d-1})$ such that $q_j = p_i$ when $j = \mu(i)$, and $q_j = 0$ when there is no $i$ such that $j = \mu(i)$. In other words, $\lift\mu(p)$ is obtained from $p$ by inserting zero coordinates at the positions not selected by $\mu$. For a set of points $P$, let $\lift\mu(P)$ be the set $\bigcup_{p \in P} \lift\mu(p)$.

\begin{definition}\label{def:extradimensional}
A $d$-dimensional space-filling curve $f$ is \emph{extradimensional to} a $d'$-dimensional space-filling curve $f'$ with co-domain $U'$, if the following holds for any valid ordering $\inv{f'}$ of $f'$ and any choice of $\mu$: there is a valid ordering $\inv{f}$ of $f$ such that for any pair of points $a,b\in U'$ we have $\inv{f'}(a) < \inv{f'}(b)$ if and only if $\inv{f}(\lift\mu(a)) < \inv{f}(\lift\mu(b))$.
\end{definition}

\begin{definition}\label{def:consistent}
An (infinite) set $F$ of space-filling curves is \emph{interdimensionally consistent}, if $F$ contains a unique $d$-dimensional space-filling curve $f_d$ for any integer $d \geq 1$, and $f_j \in F$ is extradimensional to $f_i \in F$ whenever $j > i$.
\end{definition}

Inverting the terminology, we could say that $f_{d'}$ is \emph{intradimensional} to $f_{d}$ if and only if $f_{d}$ is extradimensional to $f_{d'}$, and each curve of an interdimensionally consistent set could be regarded as an intradimensional curve of an infinitely-dimensional\footnote{I have not explored where the idea of infinitely-dimensional curves would lead us. I believe they might constitute a new class of ``monster curves'' (to use Gardner's words~\cite{Gardner}).} `curve'~$f_\infty$.

In the definition of extradimensionality, a $d$-dimensional point is constructed from a $d'$-dimensional point by inserting zero coordinates at the positions not selected by $\mu$. We will also consider \emph{\halffix-extradimensionality}, where a $2d$-dimensional point is constructed by appending $d$ zeros to the coordinates of a $d$-dimensional point, and \emph{diagonal-extradimensionality}, where a $2d$-dimensional point is constructed by concatenating two copies of the coordinate sequence of a $d$-dimensional point. In our previous work we found that \halffix- and diagonal-extradimensional space-filling curves can be useful in improving the performance and robustness of certain data structures based on space-filling curves~\cite{Haverkort4D}. 

\paragraph{The results}
This work provides descriptions of interdimensionally consistent generalizations of Hilbert's curve and several variants of Peano's curve (coil, half-coil, and Meurthe)
that caught my attention in previous work~\cite{Arrwwid,Haverkort2D} because of their excellent locality-preserving properties. Furthermore, this paper provides a general technique to derive \halffix- and diagonal-extradimensional curves from a given generalized Hilbert or Peano curve.
This enables the extension of our previous work on R-trees for rectangles to higher dimensions~\cite{Haverkort4D}.

It turns out that the new generalized Hilbert curves and the generalized Meurthe curves share another interesting property: they are, in some sense, statistically invariant under rotation. For an intuitive understanding of what this means, one could look at Figure~\ref{fig:peano2d}. At the highest level, the Peano curve shows more vertical steps than horizontal steps, and in recursion, it stays that way. Thus the curve has a clear orientation, which can be recognized immediately even if one sees only a small piece of the curve. This could be a disadvantage if the recognizable orientation somehow interacts negatively with the data in an application of the curve. On the other hand, at the smallest depicted level of detail of the Hilbert curve, all possible orientations of the curve occur approximately equally often. Therefore, from the looks of a small piece of the Hilbert curve one could not tell if the curve as a whole has been rotated or not: we say the curve has \emph{neutral orientation}. The harmonious Hilbert curves and the Meurthe curves maintain this property also in higher dimensions---unlike the Peano curves, coil curves, half-coil curves and the Butz-Moore generalization of the Hilbert curve.

\paragraph{How to read this paper}
Below, in Section~\ref{sec:definitions}, we fill first see how exactly we can define and describe space-filling curves like Hilbert's and Peano's. We will distinguish \emph{2-regular} curves, which are based on subdividing $d$-dimensional hypercubes into $2^d$ subcubes, and \emph{3-regular} curves, which are based on subdividing $d$-dimensional hypercubes into $3^d$ subcubes. In Section~\ref{sec:toolbox} we see some basic results that will come in useful in proving the claimed properties of the curves presented in the sections that follow. In Section~\ref{sec:ternary} we will present interdimensionally consistent generalizations of several 3-regular curves, and prove that among those, the generalized Meurthe curves have neutral orientation. We will also give pseudocode for comparison operators that decide which of any two given points comes first on the curve. The pseudocode may not be a recipe for highly optimized code in any particular programming language, but it may give some insight in how one may produce such code and, in addition, it reveals properties of the curves that would be difficult to see otherwise. In Section~\ref{sec:hilbert} we learn about the Butz-Moore generalization of Hilbert curves and about our harmonious Hilbert curves in detail. Again, pseudocode is included. Section~\ref{sec:composition} defines the concepts of \halffix- and diagonal-extradimensionality in detail, and explains how, given any $\base$-regular curve~$f$ (for $\base \in \{2,3\}$), one can construct a $\base$-regular curve $g$ that is \halffix- and/or diagonal-extradimensional to~$f$. The paper concludes with Section~\ref{sec:discussion} with a comparison between the results from the different sections of this paper and with open questions for possible further research.

In principle, Sections~\ref{sec:ternary}, \ref{sec:hilbert}~and~\ref{sec:composition} are independent from each other. The reader who wants to learn quickly how to get interdimensionally consistent, \halffix-extradimensional or diagonal-extradimensional curves, may read Section~\ref{sec:definitions} and then skip to any of the Sections \ref{sec:ternary}, \ref{sec:hilbert} or~\ref{sec:composition} right away. However, the reader who wants to understand the proofs in these sections, is advised to read Sections \ref{sec:definitions}, \ref{sec:toolbox}, \ref{sec:ternary} and~\ref{sec:hilbert} in order, because the proofs in these sections build up in complexity.

\section{Describing space-filling curves}
\label{sec:definitions}

\subsection{Graphical descriptions and their interpretation}
\label{sec:graphical}
In general, we can define an order (\emph{scanning order}) of points in $d$-dimensional space as follows. We give a set of rules, each of which specifies (i) how to subdivide a region in $d$-dimensional space into subregions; (ii) what is the order of these subregions; and (iii) for each subregion, which rule is to be applied to establish the order within that subregion. We also specify a unit region, and we indicate what rule is used to subdivide and order it.

\begin{figure}
\centering
\hbox to \hsize{\hfill
(a)\includegraphics[scale=0.7]{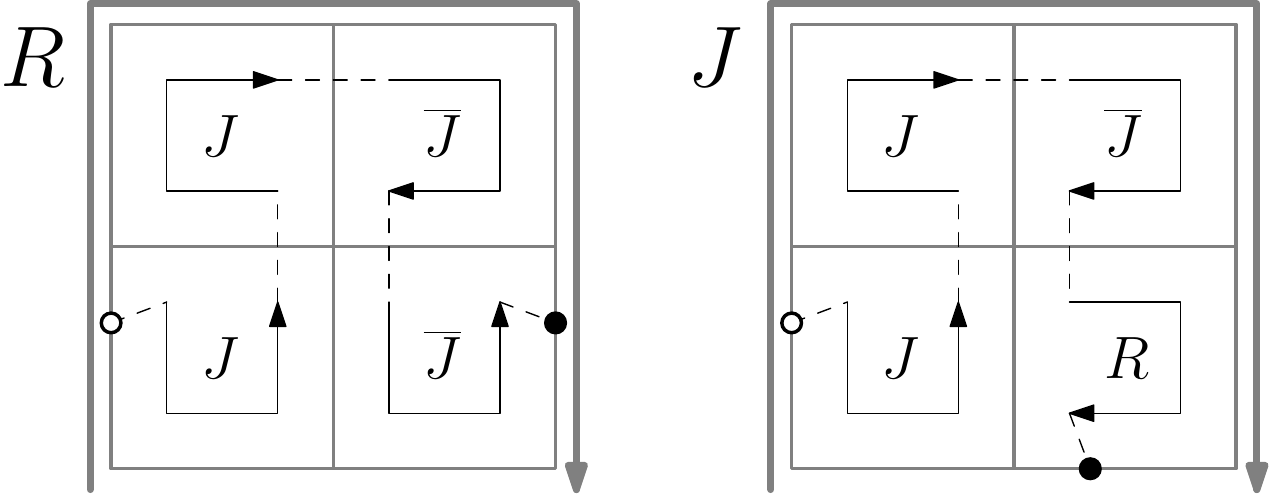}\hfill
(b)\includegraphics[scale=0.7]{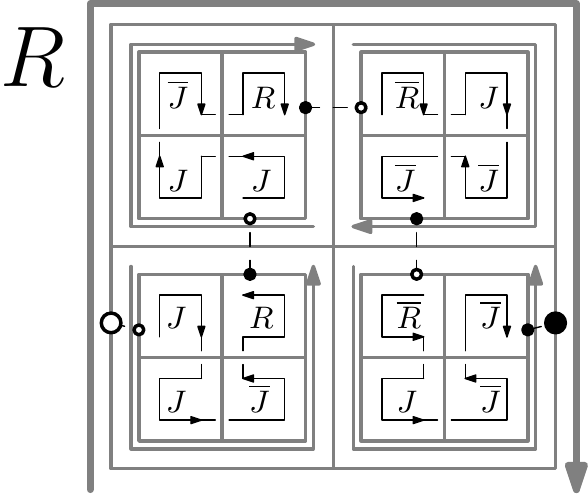}\hfill
}

\vskip\baselineskip
\hbox to \hsize{\hfill
(c)\quad\includegraphics[height=0.35\hsize]{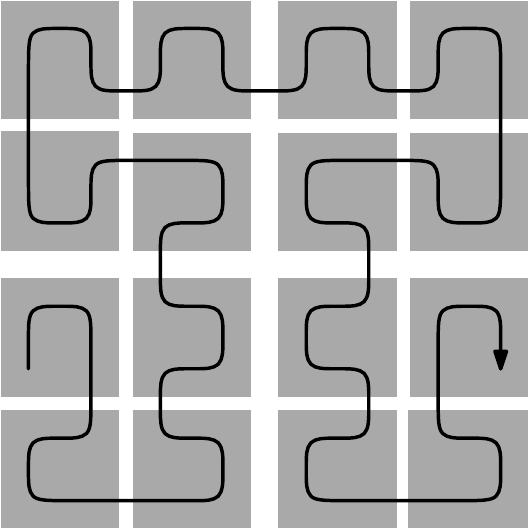}\hfill
(d)\quad\includegraphics[height=0.35\hsize]{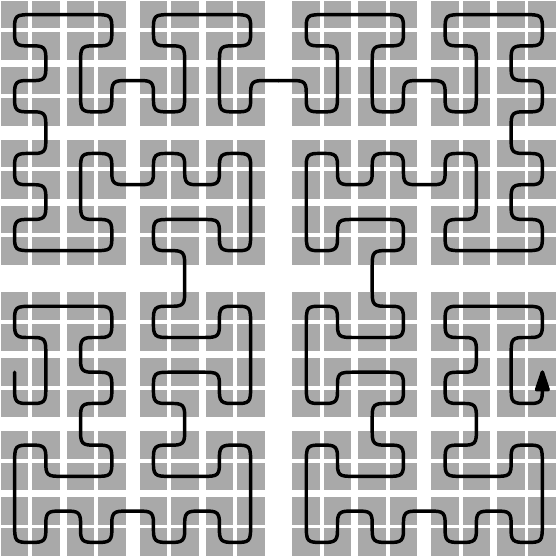}\hfill
}
\caption{%
(a) Definition of a scanning order corresponding to a variant of the $\beta\Omega$-curve.\quad
(b) Applying the rules recursively within \textsf{R}.\quad
(c) The resulting order of the $4 \times 4$-grid.\quad
(d) Order of the $16 \times 16$-grid. On each level, the last cell of each quadrant shares a vertex (and even an edge) with the first cell of the next quadrant.}
\label{fig:betaomega}
\end{figure}

Figure~\ref{fig:betaomega}(a) shows an example (the scanning order depicted corresponds to a section of the $\beta\Omega$-curve by Wierum~\cite{Wierum}). Each rule is identified by a letter, and pictured by showing a unit square, its subdivision into quadrants, the scanning order of the quadrants (by a directed curve along the outer vertices of the quadrants), and the rules applied to the quadrants (by letters). Within each quadrant, there is a scaled copy of the \emph{base pattern}: the curve along the outer vertices of the quadrants. Variations of the rules that consist of simply rotating or mirroring the order within quadrants are indicated by rotating or mirroring the pattern. Variations of rules that consist of reversing the order of cells within a quadrant are indicated by reversing the direction of the pattern and by an overscore above the letter identifying the rule\footnote{Reversals are among the standard transformations applied in the definitions of space-filling curves in the literature---see, for example, the Gosper flowsnake~\cite{Gardner}, arguably one of the prototypical space-filling curves along with the curves by Peano, Hilbert, Lebesgue~\cite{Lebesgue} and Sierpi\'nski~\cite{Bader}. In fact, reversals are so standard, that they are often not described explicitly, and it is left to the reader to infer the reversals from how the pieces of the curve are put together. I prefer to be explicit about reversals. To indicate a reversal, it is necessary to reverse the direction of the pattern \emph{and} to include the overscore. Figure~\ref{fig:betaomega}(a) illustrates this. Rule $J$ is asymmetric because different rules are applied to the first and the last quadrant. It is therefore important to distinguish between reflected, unreversed applications of $J$, as in the lower left quadrant of rule~$R$, and unreflected, reversed applications of $J$, as in the lower right quadrant. Since the base pattern of $J$ is symmetric, it is only the overscore that indicates which of these alternatives is meant.}.
Figure~\ref{fig:betaomega}(b) shows how applying the rules recursively, starting with rule \textsf{R}, results in an order of the subquadrants within each quadrant. Combined with the order of the quadrants, this gives an order of all squares in a $4 \times 4$ grid, which is sketched by the curve shown in Figure~\ref{fig:betaomega}(c). From the orientations of the base patterns in Figure~\ref{fig:betaomega}(b) and (c), one already gets a hint of the order of all squares in a $8 \times 8$ grid. By expanding the recursion further, one may order the cells of an arbitrarily fine grid.

The polygonal curve that connects the centres of these cells in order thus forms an arbitrarily fine approximation of the space-filling curve, which is more precisely defined as follows. At any level of recursion, let the grid consist of the cells $C_1,...,C_n$, indexed from $1$ to $n$ according to the scanning order, let $\volume[1,i] = \sum_{j=1}^{i} \volume(C_i)$ be the total volume of cells $C_1$ up to and including $C_i$, and define $\relvolume[1,i] = \volume[1,i]/\volume[1,n]$. Then the curve $f$ maps the interval $[\relvolume[1,i-1], \relvolume[1,i]]$ to $C_i$, for all $i \in \{1,...,n\}$. The image $f(q)$ of any point or interval $q$ in $[0,1]$ can now be obtained by taking the limit of the image of the intervals containing $q$ as the level of recursion goes to infinity. Thus the definition of a scanning order is also a definition of a space-filling curve.

In the definitions of space-filling curves, the letters indicating the rules may be omitted if there are no reversals and no two rules have the same base pattern modulo rotation and reflection---see, for example, Figure~\ref{fig:original-curves}.
Figures \ref{fig:2DWunderlich}, \ref{fig:3DMeurthe}, \ref{fig:3DHarmonious} and~\ref{fig:monotone} and Inset~\ref{box:variety} show further examples of definitions of two- and three-dimensional space-filling curves.

\paragraph{Entrance and exit gates}
When the scanning order is refined to an infinitely fine grid, the first and the last cell visited shrink to points. These points are the images of 0 and 1 of the space-filling curve defined by the scanning order, and they typically lie on the boundary of the unit region\footnote{It is not automatic that these points lie on the boundary, but space-filling curves are usually designed such that $f(0)$ and $f(1)$ do lie on the boundary---otherwise the curve cannot be continuous from one subregion into another subregion if each subregion is filled with a scaled-down copy of the curve.}. We will call $f(0)$ the \emph{entrance gate} and we will call $f(1)$ the \emph{exit gate} of the curve; in Figure~\ref{fig:betaomega} their locations are indicated by the white dots and the black dots, respectively. Note that these dots are only shown for clarity: the locations of the gates are a consequence of the definition of the scanning order, rather than a part of the definition.

\paragraph{Using space-filling curves}
For many practical applications of space-filling curves it is not necessary to actually draw a curve. It is often enough to be able to decide for any two given points $p$ and $q$, which of these appears first on the curve---that is, which point comes first in the scanning order. This can be decided by expanding the recursion only to the smallest depth at which $p$ and $q$ lie in different quadrants. A technical problem with this is that, down from some depth of recursion, a given point $p$ may always lie on the boundary of two or more quadrants. This may create ambiguity about which of two given points comes first. This ambiguity can be resolved by using a consistent tie-breaking rule. For example, in the pseudocode in this paper, we will always assign each point $p$ to the quadrant that lies to its upper right---that is, for a space-filling curve $f$ we define $\inv{f}(p)$ as the common limit of the values $t$ such that $f(t) = p + (\eps,...,\eps)$ as $\eps > 0$ approaches zero.

\subsection{Classes of space-filling curves considered in this paper}
\label{sec:classes}
There are many different possible generalizations of Peano's and Hilbert's curves to three and more dimensions. To be able describe the classes of curves we will be studying in this paper precisely, I will now introduce some terminology.

All space-filling curves studied in this paper are defined by a rule system as described above.

We will call a curve $\base$-\emph{regular} if the unit region of each rule is a $d$-dimensional unit hypercube, and each rule subdivides this hypercube into $\base^d$ smaller hypercubes of equal size. We adopt the convention that the unit hypercube is an axis-parallel hypercube that spans the interval $[0,1]$ in each dimension; thus, each of the smaller hypercubes has width~$1/\base$.

We will call a curve a \emph{mono-curve} if the defining rule system contains only a single rule.

We will call a curve \emph{order-preserving} if it is a mono-curve and the transformations of the single definining rule within the subregions are restricted to rotation and reflection---excluding reversal of the order in which the subregions are traversed.

We will call a scanning order and the space-filling curve defined by it \emph{vertex-continuous} if it has the following property: whenever the scanning order visits a set of hypercubes $C_1,...,C_n$, in that order, and we refine the scanning order in each of these hypercubes recursively to depth $\depth$ so that we obtain an order on $n \cdot \base^{\depth d}$ cells, for any $\depth \geq 1$, then the last cell within $C_{i-1}$ shares at least one point with the first cell within $C_i$, for all $1 < i \leq n$ (see, for example, Figures \ref{fig:peano2d} and~\ref{fig:betaomega}(d)). Note that if a scanning order is not vertex-continuous, then it does not directly define a \emph{continuous} mapping from the unit interval to a higher-dimensional space, and thus, the scanning order does not actually define a space-filling \emph{curve}.\footnote{A work-around would be to insert infinitely many connecting line segments to bridge the gaps. In effect, this is how Lebesgue makes a space-filling curve out of the non-vertex-continuous scanning order that is nowadays known as Z-order~\cite{Lebesgue}. However, such tricks come at the expense of decreased performance in applications~\cite{Haverkort2D}, and therefore they are not considered in this paper.} One of our main tasks in the following sections of this paper, will always be to prove that whatever is presented as a space-filling curve, is indeed vertex-continuous.

Inset~\ref{box:variety} shows some of the variety of vertex-continuous 2-regular mono-curves that exist.

\begin{inset}
\caption{What does it take to be a Hilbert curve?}
\footnotesize
Hilbert curves are often associated with 2-regular, order-preserving, symmetric mono-curves, that visit the subregions of the hypercube in the order of the binary reflected Gray code (see Section~\ref{sec:graycodes}), and have the entrance and exit gate located at the two endpoints of an edge of the hypercube. Indeed, Hilbert's original two-dimensional curve has all of these properties, but one may argue that they do not \emph{define} the curve. There is only one vertex-continuous 2-regular mono-curve, and this happens to be the order-preserving, symmetric, reflected-Gray-code-based, vertex-gated Hilbert curve.

In three-dimensional space, there is not a single vertex-continuous 2-regular mono-curve, but there are 10\,694\,807 of them---not counting rotated, reflected and/or reversed copies of the same curve~\cite{Haverkort3D}. We will call all of these curves \emph{mono-Hilbert curves}. Only 48 of them have all of the non-defining properties mentioned above.
The 48 order-preserving, symmetric, reflected-Gray-code-based, vertex-gated curves include those depicted in Figures (a), (b) and~(c) below. Figure~(a) shows the curve implemented by Moore~\cite{Moore}, based on the work of Butz. Curve~(b) is the only three-dimensional mono-Hilbert curve that is extradimensional to the two-dimensional Hilbert curve; it is the three-dimensional \emph{harmonious Hilbert curve} as defined in Section~\ref{sec:hilbert}. Curve~(c) has a very regular structure: it is not only composed of eight similar octants and of two symmetric halves, but also of four quarters (of two octants each) which are reflected, rotated and/or reversed copies of each other.

The remaining 10\,694\,759 mono-Hilbert curves include curves that follow an alternative Gray code (e.g.\ Figure~(d)) and curves that are not order-preserving (e.g.\ Figure~(e)). There is even one curve that has its entrance and exit gates in the interiors of two faces of the cube (Figure~(f)): in fact, this is the three-dimensional mono-Hilbert curve that has the best locality-preserving properties if measured according to the Euclidean dilation measure~\cite{Haverkort3D}. The curves in Figures (e) and~(f) are the only two mono-Hilbert curves such that on any level recursion, the distance between the centroids of the predecessor and the successor of a subregion is only $\sqrt{2}$ times the width of the subregion---the centroids of three consecutive subregions are never collinear.
There are also curves with `diagonal connections', where successive subregions only share an edge or a vertex: for example, see Curve~(g) (with a similar regular structure as Curve~(c)) and Curve~(h) (rotationally symmetric in the left-right axis, it may be more regular than it seems).

\hbox to \hsize{\hfill
(a)\includegraphics[width=0.21\hsize]{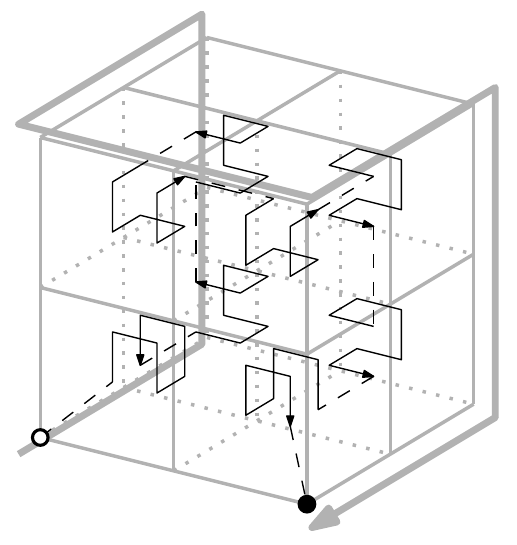}\hfill
(b)\includegraphics[width=0.21\hsize]{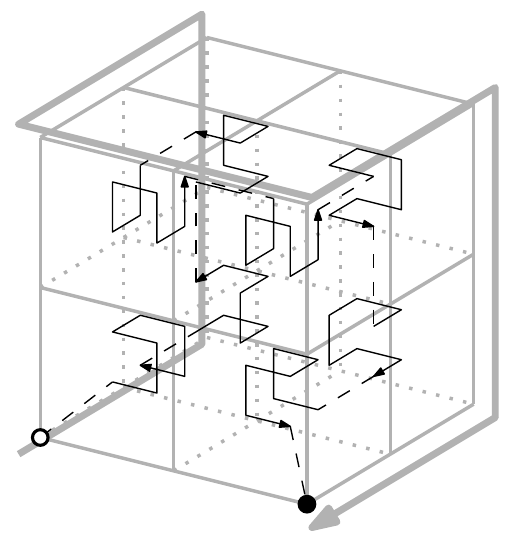}\hfill
(c)\includegraphics[width=0.21\hsize]{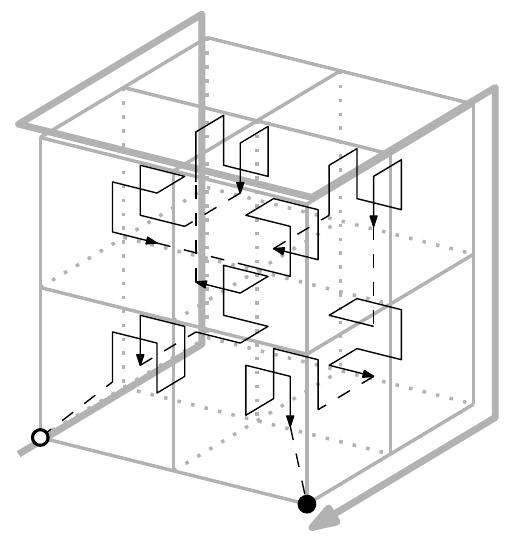}\hfill
(d)\includegraphics[width=0.21\hsize]{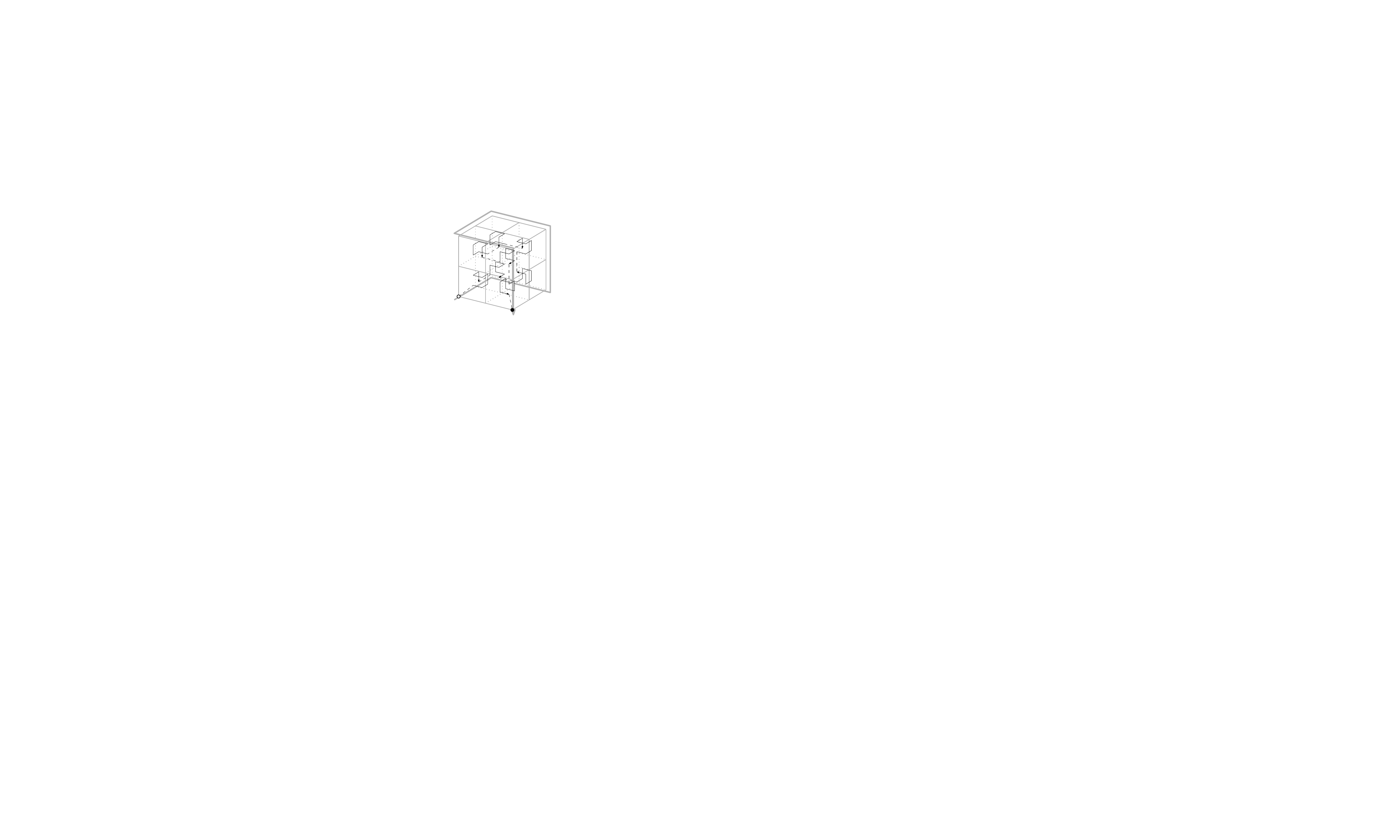}\hfill
}\hbox to \hsize{\hfill
(e)\includegraphics[width=0.21\hsize]{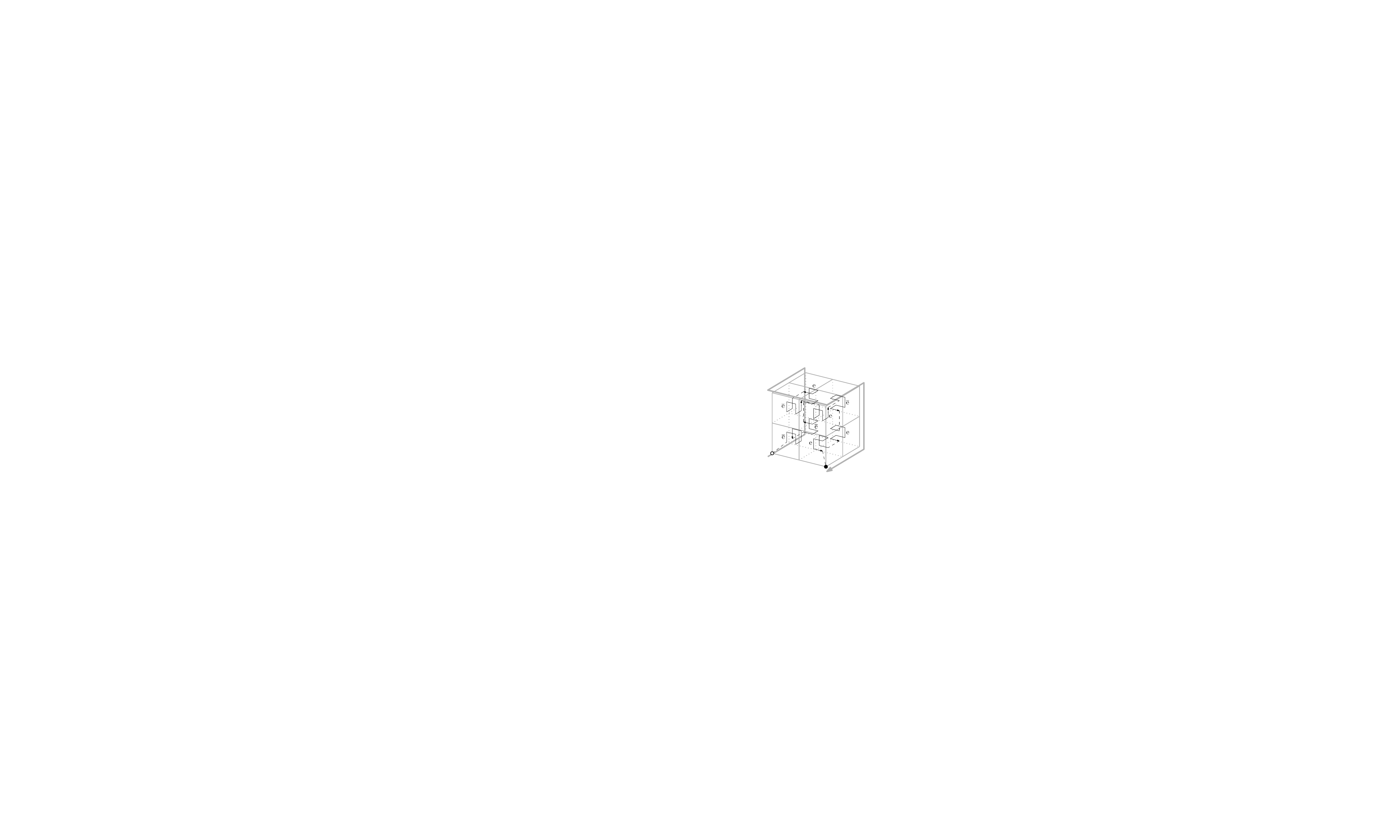}\hfill
(f)\includegraphics[width=0.21\hsize]{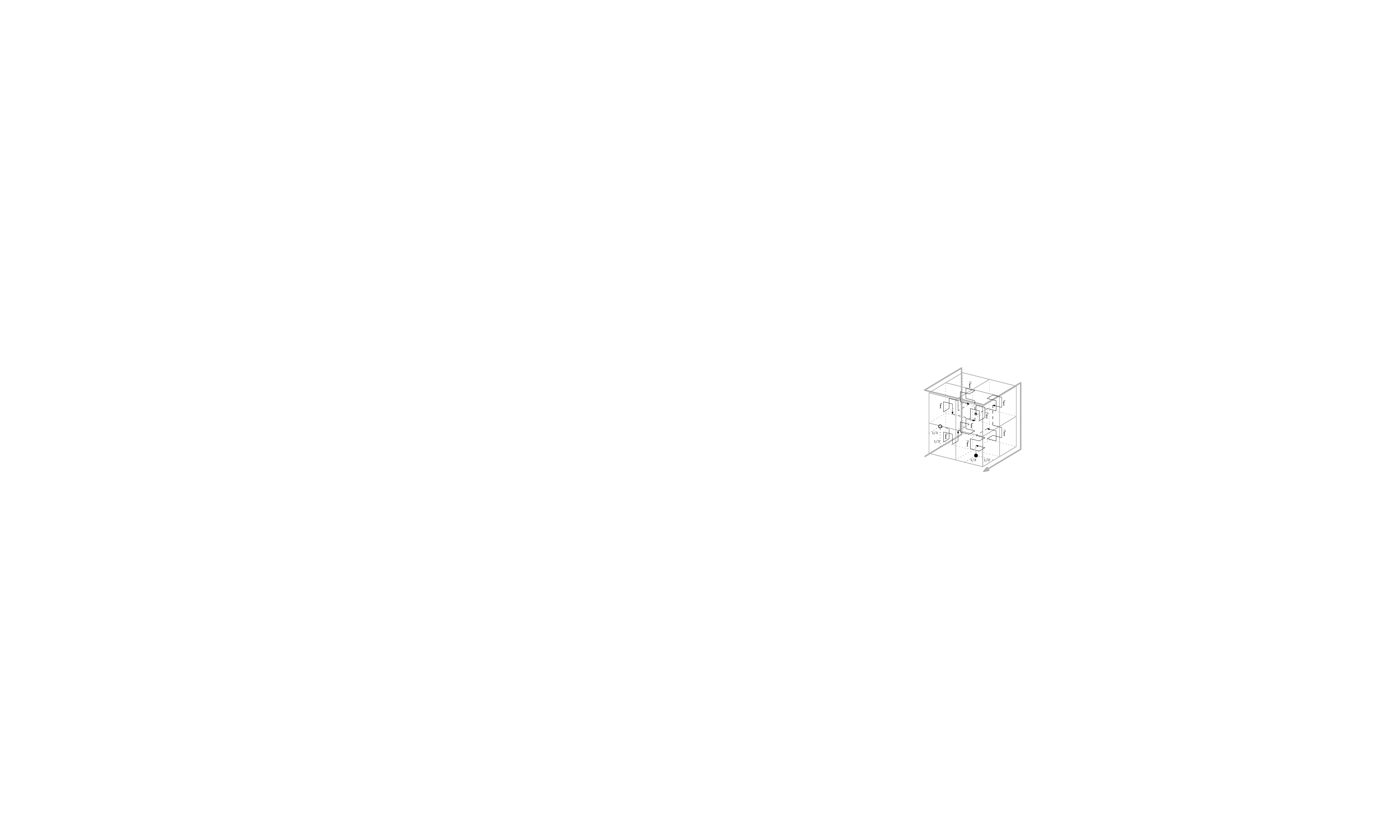}\hfill
(g)\includegraphics[width=0.21\hsize]{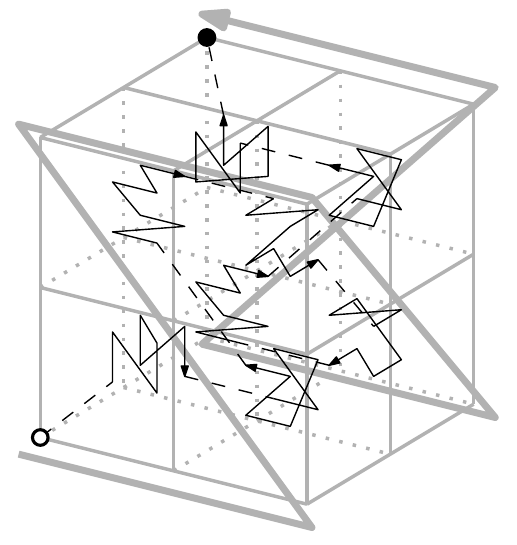}\hfill
(h)\includegraphics[width=0.21\hsize]{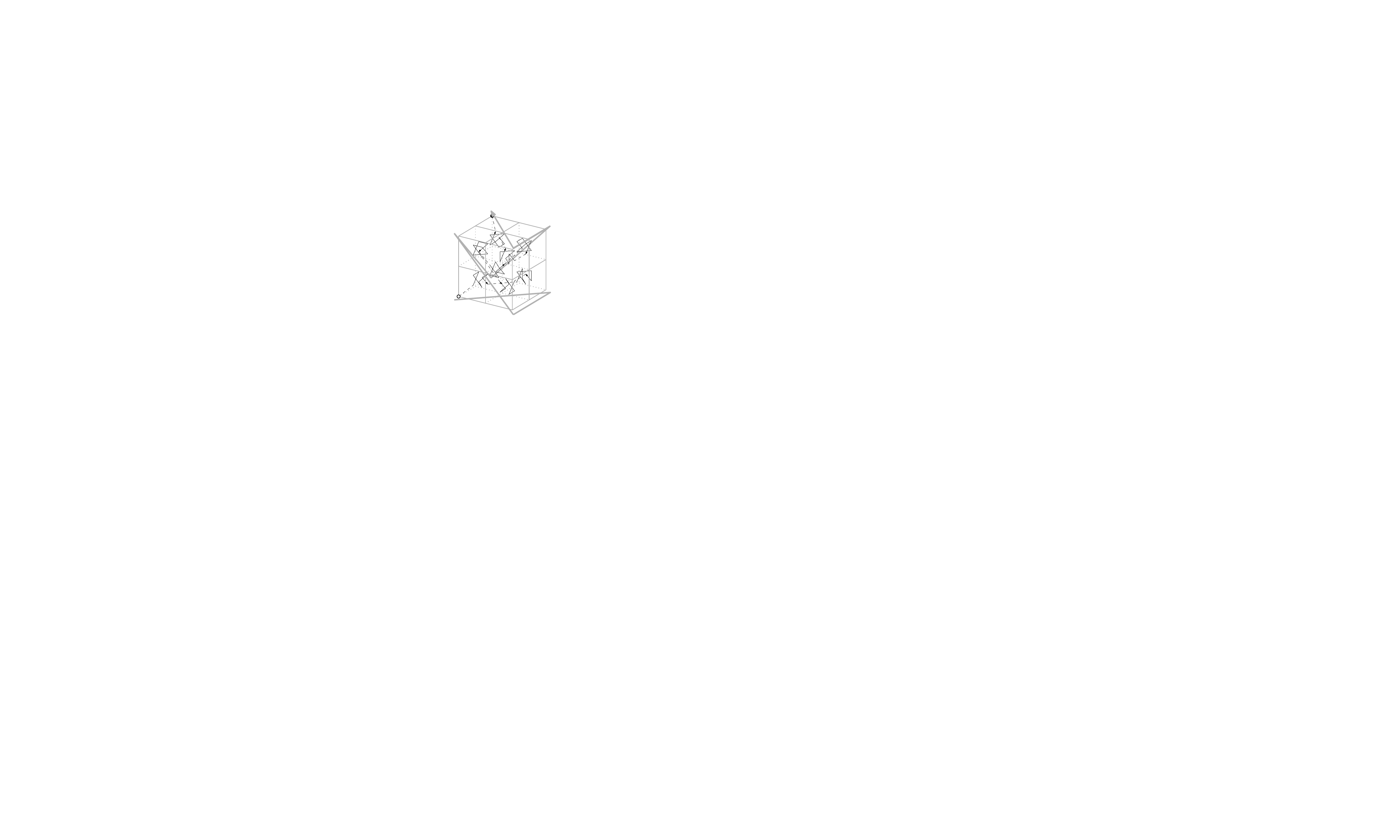}\hfill
}

\addvspace{.5\baselineskip}
Some three-dimensional vertex-continuous 2-regular mono-curves, with the names I gave them~\cite{Haverkort3D}.
(a)~A26.00.00, the three-dimensional Hilbert curve by Butz and Moore;\quad
(b)~A26.00.db, the three-dimensional harmonious Hilbert curve;\quad
(c)~A16.00.3c;\quad
(d)~A6.00.ff;\quad
(e)~A26.2b.b3;\quad
(f)~F;
(g)~B15.00.c3;\quad
(h)~B138.00.99 (not published before).%
\label{box:variety}
\end{inset}

\subsection{Numerical descriptions}
\label{sec:numerical}
Obviously, the graphical way of defining space-filling curves described above does not extend well to curves in more than three dimensions. Therefore, I will now introduce a concise numerical way of defining space-filling curves, tailored to $\base$-regular order-preserving curves.

In the numerical notation, we will use $d$-dimensional coordinate vectors and $d$-digit numbers. We use subscripts from 0 to $d-1$ to index the individual coordinates or digits, respectively. If $x$ is a $d$-digit number, $x_0$ is the \emph{most} significant and $x_{d-1}$ is the \emph{least} significant digit of $x$. By $\reps ad$ we denote the number that consists of $d$ times the digit $a$.

The \emph{rank} of each subregion is a base-$\base$ number from $\{0,...,\reps{\base-1}d\}$ that indicates the position of the subregion in the sorted order of subregions according to the rule that defines the curve: the first subregion has rank~0 and the last subregion has rank~$\reps{\base-1}d$. By $S(r)$ we denote the subregion with rank~$r$. For ease of notation we write $\pred{r}$ for $r-1$, and we write $\succ{r}$ for $r+1$.

The \emph{location} of $S(r)$ is denoted by $c(r)$, which is a $d$-digit base-$\base$ number indicating that the subregion's corner that is closest to the origin is located at $\frac1\base(c_0(r),...,c_{d-1}(r))$ and the subregion's corner that is farthest from the origin is located at $\frac1\base(c_0(r)+1,...,c_{d-1}(r)+1)$.

The curve within each subregion $S(r)$ can be obtained from the curve in the unit hypercube as a whole by first rotating it, then reflecting it, then scaling it with a factor $1/\base$ in all dimensions, and then translating it into place. I will now describe how rotations, reflections and translations are specified. The rotation is specified by a permutation $a(r)$, which is a $d$-digit base-$d$ number of which the digits form a permutation of $\{0,...,d-1\}$. This permutation specifies which axis $a_i(r)$ of the curve through the unit hypercube is rotated onto axis $i$ within subregion $r$. By $\inv{a}(r)$ we denote the inverse of this permutation, that is, $\inv{a}_i(r) = j$ if and only if $a_j(r) = i$. The reflections are specified by a $d$-digit binary number $m(r)$, where $m_i(r) = 1$ if and only if after rotation, the curve within subregion $r$ should be reflected in a plane orthogonal to axis~$i$. The required translation vector $o(r)$ is implied by the rotation, the reflections and the location of the subregion, and does not need to be specified separately.

More precisely, $a(r)$ and $m(r)$ define a transformation matrix $M(r)$ with rows and columns numbered from $0$ to $d-1$ and the following entries:\[\begin{array}{ll}
M(r)_{ij} = 1 & \hbox{if $j = a_i(r)$ and $m_i(r) = 0$;} \\
M(r)_{ij} = -1 & \hbox{if $j = a_i(r)$ and $m_i(r) = 1$;} \\
M(r)_{ij} = 0 & \hbox{if $j \neq a_i(r)$.}
\end{array}\]
By $\inv{M}(r)$ we denote the inverse of this matrix:\[\begin{array}{ll}
\inv{M}(r)_{ji} = 1 & \hbox{if $i = \inv{a}_j(r)$ and $m_i(r) = 0$;} \\
\inv{M}(r)_{ji} = -1 & \hbox{if $i = \inv{a}_j(r)$ and $m_i(r) = 1$;} \\
\inv{M}(r)_{ji} = 0 & \hbox{if $i \neq \inv{a}_j(r)$.}
\end{array}\]
The transformation $\tau(r): p \rightarrow \frac 1\base M(r) p + o(r)$ now maps each point $p$ of the curve through the unit hypercube to a point in subregion $S(r)$. Recall that the space-filling curve is a function $f$ from $[0,1]$ to the unit hypercube $[0,1]^d$. Considering the subregion with rank $r$, the function~$f$ maps the domain \hbox{$[r/\base^d,(r+1)/\base^d]$} to $S(r)$. The transformation $\tau(r)$ describes the relation between, on one hand, $f$~restricted to the domain $[r/\base^d,(r+1)/\base^d]$, and on the other hand, $f$~on the domain $[0,1]$: for $x \in [r/\base^d,(r+1)/\base^d]$ we have $f(x) = \tau(r)(f(s_r(x)))$, where $s_r(x) = \base^d x - r$ is the function that scales up the domain $[r/\base^d,(r+1)/\base^d]$ to $[0,1]$. Conversely, there is a valid ordering $\inv{f}$ of $f$ such that for any point $p$ in the interior of subregion $r$ we have: $\inv{f}(p) = \inv{s_r}\left(\inv{f}(\inv{\tau}(r)(p))\right)$, where $\inv{\tau}(r)(p) = \inv{M}(r) \base(p - o(r))$ and $\inv{s_r}(x) = (x + r) / \base^d$.

A curve can now be described by a table that lists, for each subregion $S(r)$, the location $c(r)$, the permutation $a(r)$, the reflections $m(r)$, and, redundantly (if desired for clarity), the entrance and/or exit gates of each subregion. Figure~\ref{fig:original-curves} shows the definitions of the two-dimensional curves by Hilbert and Peano in graphical notation and in table format.

\begin{figure}
\centering
\hbox to \hsize{\hfill
\includegraphics[width=\hsize]{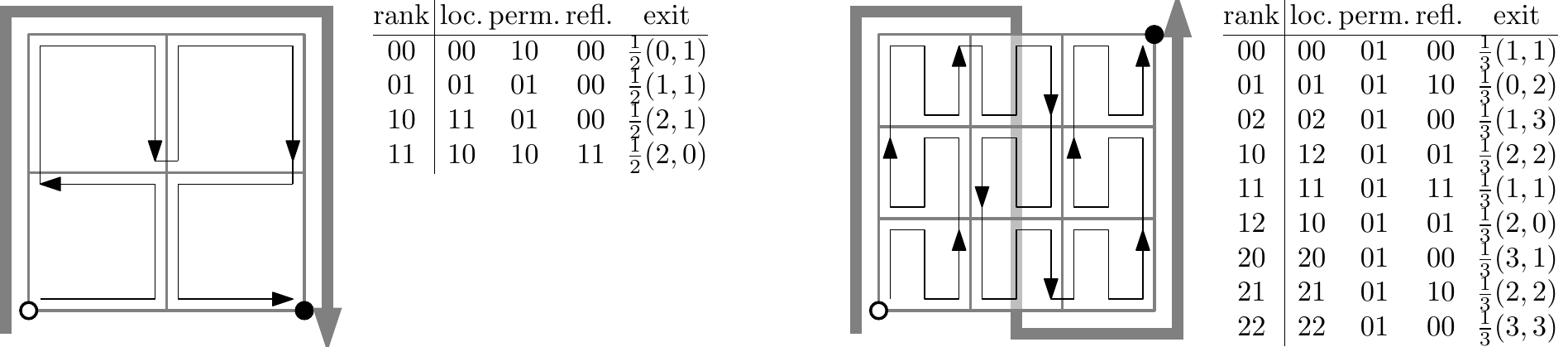}\hfill
}
\caption{Hilbert's curve (left) and Peano's curve (right) defined graphically and by means of a table.}
\label{fig:original-curves}
\end{figure}

\subsection{Implementing a comparison operator}
\label{sec:comparisonoperator}

We will know see how we can easily implement a comparison operator for an order-preserving $\base$-regular curve. It is not necessary to read this section to be able to understand the descriptions and proofs in the rest of this paper, and it is not necessary to read this section to be able to use the pseudocode presented in the rest of this paper. However, it may help in understanding why the pseudocode in Sections \ref{sec:ternary}--\ref{sec:composition} works.

Our goal in this section is the following: given two points $p, q$ in $[0,1]^d$, whose coordinates are given as base-$\base$ numbers, decide whether $\inv{f}(p) < \inv{f}(q)$. As an ordering $\inv{f}$ corresponding to $f$ we choose the ordering that assigns each point $p$ to the subregion that lies above/behind/to the right of it in all dimensions---that is, we define $\inv{f}(p)$ as the common limit of the values $t$ such that $f(t) = p + (\eps,...,\eps)$ as $\eps > 0$ approaches zero. This defines $\inv{f}$ only for points in $[0,1)^d$; it is not defined for points of which one or more coordinates are equal to~1, and our pseudocode will not deal with such points. One could define $\inv{f}$ also for points with one or more coordinates equal to 1, and one could easily adapt the algorithms presented in this paper to handle such coordinates, by representing the number 1 as 0.111... (if $\base = 2$) or 0.222... (if $\base = 3$).

The idea of our basic algorithm (Algorithm~\ref{alg:abstractoperator}) is as follows. The algorithm takes two points $p$ and~$q$, each given as a set of $d$ coordinates $(p[0],...,p[d-1])$ and $(q[0],...,q[d-1])$, respectively. Assume that each coordinate is in $[0,1)$ and it is given as the string of digits that constitute the fractional part of its representation as a base-$\base$ number. The algorithm first removes the first digit of the fractional part of each coordinate. This is done with a function $\Extract(x)$, which removes the first digit from the digit string $x$ and returns the removed digit---if \Extract\ is called on an empty string, it simply returns a zero. Now let $c_p$ and $c_q$ be the concatenated removed digits of $p$ and~$q$, respectively. Determine $r_p$ and $r_q$ such that $c(r_p) = c_p$ and $c(r_q) = c_q$; thus $r_p$ and $r_q$ are the ranks of the regions containing $p$ and $q$, respectively. If $r_p < r_q$, return true; if $r_p > r_q$, return false; otherwise, we apply the transformation $x \rightarrow \inv{M}(r_p)x$ to the points $p$ and $q$ that remained after removing the first digits of their coordinates, and recurse on the resulting points $p$ and $q$. In fact, in Algorithm~\ref{alg:abstractoperator} we do not apply the transformation right away, but instead we maintain a transformation matrix $T$ which is the composition of all transformations $\inv{M}(r_p)$ that should have been applied at all higher levels of recursion. We apply this composed transformation $T$ only to the digits extracted, on Line~\ref{alg:lazytransformation}.

\begin{algorithm}
\KwIn{Points $p = (p[0],...,p[d-1])$ and $q = (q[0],...,q[d-1])$, with coordinates in $[0,1)$, given by the digits of the fractional part in base-$\base$ representation}
\KwOut{\True\ if $\inv{f}(p) < \inv{f}(q)$, \False\ otherwise}
\BlankLine
$T \leftarrow$ identity matrix of size $d$\;
\Repeat{all remaining digit strings $p[0],...,p[d-1]$ and $q[0],...,q[d-1]$ are empty}{
  \For{$i \leftarrow 0$ \KwTo $d-1$}{
    $c_p[i] \leftarrow \Extract(p[i])$;
    $c_q[i] \leftarrow \Extract(q[i])$\;
  }
  $c_p \leftarrow T \cdot c_p$; $c_q \leftarrow T \cdot c_q$\;\label{alg:lazytransformation}
  add $\base-1$ to all elements of $c_p$ and $c_q$ that are negative\;\label{alg:correcttranslation}
  compute $r_p, r_q$ such that $c(r_p) = c_p$ and $c(r_q) = c_q$\;
  \lIf{$r_p < r_q$}{\Return \True}\ \lElseIf{$r_p > r_q$}{\Return \False}\;
  \BlankLine
  \tcp{$c_p = c_q$ and $r_p = r_q$}
  $T \leftarrow \inv{M}(r_p) \cdot T$\;
}
\Return\False\tcp*[r]{$p$ and $q$ are equal}
\caption{General comparison operator for a $\base$-regular curve: the main idea\label{alg:abstractoperator}}
\end{algorithm}

The correctness of Algorithm~\ref{alg:abstractoperator} is based on the following observations. To compare $\inv{f}(p) = \inv{s_r}\left(\inv{f}(\inv{\tau}(r)(p))\right)$ and $\inv{f}(q) =\inv{s_r}\left(\inv{f}(\inv{\tau}(r)(q))\right)$ we do not have to apply $\inv{s_r}$ if both points lie in the same region $S(r)$: we can simply compare $\inv{f}(\inv{\tau}(r)(p))$ to $\inv{f}(\inv{\tau}(r)(q))$. Algorithm~\ref{alg:abstractoperator} does this as follows. Removing the first digits of the coordinates implements the transformation $x \rightarrow \base(x - o'(r))$, where $o'(r)$ is the corner of $S(r)$ closest to the origin. Thus, assuming $o(r) = o'(r)$, subsequently applying $x \rightarrow \inv{M}(r)x$ implements $\inv{\tau}(r)$. If there are reflections, that is, negative entries in $\inv{M}(r)$, then the situation is slightly more complicated. The translation vector $o(r)$ should be the coordinates of the corner of region $S(r)$ that maps to the origin under the transformation $\inv{\tau}(r)$. Mapping $o'(r)$ to the origin instead, results in a translation of the unit hypercube $\inv{\tau}(r)(S(r))$ over a distance -1 in each of the reflected coordinates. Line~\ref{alg:correcttranslation} corrects this by adding $\base-1$ to each extracted digit of a negative coordinate---which amounts to adding $0.111... = 1$ (if $\base = 2$) or $0.222... = 1$ (if $\base = 3$).

Algorithm~\ref{alg:generaloperator} shows a version of Algorithm~\ref{alg:abstractoperator} which is more explicit about how to maintain and apply the transformation matrix $T$. Note that $T$ is always a matrix with exactly one non-zero entry in each row, and this entry is either $1$ or $-1$. Therefore, instead of maintaining a complete matrix, we can simply maintain where the non-zero entries are and whether they are $1$ or $-1$.

\begin{algorithm}
\KwIn{Points $p = (p[0],...,p[d-1])$ and $q = (q[0],...,q[d-1])$, with coordinates in $[0,1)$, given by the digits of the fractional part in base-$\base$ representation}
\KwOut{\True\ if $\inv{f}(p) < \inv{f}(q)$, \False\ otherwise}
\BlankLine
\lFor{$i \leftarrow 0$ \KwTo $d-1$}{$\Reflected[i] \leftarrow \False$\tcp*[r]{maintains which columns of $T$ have minus}}
\lFor{$i \leftarrow 0$ \KwTo $d-1$}{$\Permutation[i] \leftarrow i$\tcp*[r]{maintains which entry in each row of $T$ is non-zero}}\label{alg:general:identity}
initialize an array $\AltPermutation[0..d-1]$\tcp*[r]{for a new value of \Permutation under construction}\label{alg:general:altpermutation}
\Repeat{all remaining digit strings $p[0],...,p[d-1]$ and $q[0],...,q[d-1]$ are empty}{\label{alg:general:mainloop}
  \For{$i \leftarrow 0$ \KwTo $d-1$}{\label{alg:general:foralldimensions}
    $c_p[i] \leftarrow \Extract(p[\Permutation[i]])$;
    $c_q[i] \leftarrow \Extract(q[\Permutation[i]])$\;
    \If{$\Reflected[\Permutation[i]]$}{\label{alg:general:reflected}
      $c_p[i] \leftarrow (\base - 1 - c_p[i])$; $c_q[i] \leftarrow (\base - 1 - c_q[i])$
    }
  }
  compute $r_p, r_q$ such that $c(r_p) = c_p$ and $c(r_q) = c_q$\;
  \lIf{$r_p < r_q$}{\Return \True}\ \lElseIf{$r_p > r_q$}{\Return \False}\;
  \BlankLine
  \tcp{$c_p = c_q$ and $r_p = r_q$}
  \For{$i \leftarrow 0$ \KwTo $d-1$}{
    $j \leftarrow \inv{a}_i(r_p)$\;
    $\AltPermutation[i] \leftarrow \Permutation[j]$\;
    \lIf{$m_j(r_p) = 1$}{$\Reflected[\AltPermutation[i]] \leftarrow \Not \Reflected[\AltPermutation[i]]$}\;
  }
  swap $\Permutation$ and $\AltPermutation$\;
}\label{alg:general:endloop}
\Return\False\tcp*[r]{$p$ and $q$ are equal}
\caption{General comparison operator for a $\base$-regular curve: more detailed algorithm\label{alg:generaloperator}}
\end{algorithm}

\subsection{Reflected Gray codes}
\label{sec:graycodes}

Many variants of Peano's and Hilbert's curves are based on \emph{Gray codes}: sequences of base-$\base$ numbers of $d$ digits such that each number from $\{0,...,\base^d-1\}$ occurs exactly once and consecutive numbers differ in only one digit. More precisely, we require that the difference in this digit is always exactly one. When the numbers of such a Gray code are interpreted as the locations $c(r)$ of the regions $S(r)$ for $r = 0,...,\base^d-1$, the result is a space-filling curve in which each visited region $S(r)$ shares a $(d-1)$-dimensional face with its predecessor $S(\pred{r})$ (recall that we use $\pred{r}$ to denote $r-1$); for examples, see Figures (a)--(f) in Inset~\ref{box:variety}. Thus, each region fits snugly together with the previous region. This tends to result in space-filling curves with good locality-preserving properties, although by itself, it is not a guarantee~\cite{Haverkort3D}.

In many of the definitions of space-filling curves that follow, we will use a specific type of Gray codes called \emph{reflected Gray codes}. By $\kRGC^d$ we denote the \emph{reflected Gray code} with base~$\base$. This code is a sequence of $\base^d$ base-$\base$ numbers, containing each of the $\base^d$ different base-$\base$ numbers exactly once. The code is defined recursively as follows. Let $a \oplus \kRGC^d$ denote the sequence that is obtained by prefixing each number of $\kRGC^d$ with the digit $a$, let $\mathrm{rev}(\kRGC^d)$ denote $\kRGC^d$ in reverse order, and let $\mathrm{revifodd}(a,\kRGC^d)$ be $\kRGC^d$ if $a$ is even, and $\mathrm{rev}(\kRGC^d)$ if $a$ is odd. Then $\kRGC^1$ is the sequence $\langle 0,1,...,\base-1 \rangle$, and $\kRGC^{d+1}$, for $d \geq 1$, consists of the concatenation of the sequences $a \oplus \mathrm{revifodd}(a,\kRGC^d)$ for $a = 0,1,...,\base-1$ in order. For example, $\TRGC^1 = \langle 0,1,2\rangle$ and $\TRGC^2 = \langle 00,01,02,12,11,10,20,21,22\rangle$. We can index $\kRGC^d$ by $d$-digit base-$\base$ numbers. More precisely: we can consider $\kRGC^d$ to be a bijection between the $d$-digit base-$\base$ numbers and the $d$-digit base-$\base$ numbers, where $\kRGC^d(r)$ is the element $r$ (counting from zero) of the sequence $\kRGC^d$. Observe that $\kRGC^d(0) = 0$.

The reader may verify that by construction, two consecutive elements $\kRGC^d(\pred{r})$ and $\kRGC^d(r)$ differ in only one digit, and the difference in value of this digit is one.

Furthermore, $\kRGC^d(r)$ can be obtained from $r$ by going through the digits of $r$ in order from $r_0$ to $r_{d-1}$: whenever a digit $r_i$ is encountered such that $r_i$ is odd, we `reflect' all following digits, that is, we replace $r_j$ by $\base - 1 - r_j$ for all $j > i$. To understand why this works, consider that after having seen $r_i$, we have fixed the first $i+1$ digits of $\kRGC^d(r)$; if $\kRGC^d_i(r)$ is even, then the remaining digits must constitute the element of $\kRGC^{d-i-1}$ indexed by the base-$\base$ number $r_{i+1}...r_{d-1}$, whereas if $\kRGC^d_i(r)$ is odd, then the remaining digits must constitute the element of $\mathrm{rev}(\kRGC^{d-i-1})$ indexed by the base-$\base$ number $r_{i+1}...r_{d-1}$, or equivalently, the element of $\kRGC^{d-i-1}$ indexed by $r'_{i+1}...r'_{d-1}$, where $r'_j = \base-1 - r_j$.

Conversely, we compute $r$ from $\kRGC^d(r)$ with Algorithm~\ref{alg:graydecoder}.
For a $\base$-regular curve of which $c(r)$ is defined as $\kRGC^d(r)$, we can integrate Algorithm~\ref{alg:graydecoder} into Algorithm~\ref{alg:generaloperator} (which computes which of two points comes first along a space-filling curve) and obtain Algorithm~\ref{alg:RGCcurvesoperator}.

\begin{algorithm}
\KwIn{An element of $\kRGC^d$, given as an array $c = (c[0],...,c[d-1])$ of base-$\base$ digits}
\KwOut{$\Rank$ s.t. $\kRGC^d(\Rank) = c$, given as an array $\Rank[0..d-1]$ of base-$\base$ digits}
\BlankLine
$\Forward \leftarrow \True$\tcp*[r]{indicates whether we go through $\kRGC^{d-i}$ forwards or backwards}
\For{$i \leftarrow 0$ \KwTo $d-1$}{
\lIf{$\Forward = \True$}{$\Rank[i] = c[i]$}
\lElse{$\Rank[i] = (\base - 1 - c[i])$}\;
\lIf{$c[i]$ is odd}{$\Forward \leftarrow \Not \Forward$}
}
\caption{Gray-decoding algorithm for $\kRGC^d$\label{alg:graydecoder}}
\end{algorithm}

\begin{algorithm}
\KwIn{Points $p = (p[0],...,p[d-1])$ and $q = (q[0],...,q[d-1])$, with coordinates in $[0,1)$, given by the digits of the fractional part in base-$\base$ representation}
\KwOut{\True\ if $\inv{f}(p) < \inv{f}(q)$, \False\ otherwise}
\BlankLine
\lFor{$i \leftarrow 0$ \KwTo $d-1$}{$\Reflected[i] \leftarrow \False$\tcp*[r]{maintains which columns of $T$ have minus}}
\lFor{$i \leftarrow 0$ \KwTo $d-1$}{$\Permutation[i] \leftarrow i$\tcp*[r]{maintains which entry in each row of $T$ is non-zero}}
initialize an array $\AltPermutation[0..d-1]$\tcp*[r]{for a new value of \Permutation under construction}\label{alg:rgccurves:altpermutation}
\Repeat{all remaining digit strings $p[0],...,p[d-1]$ and $q[0],...,q[d-1]$ are empty}{\label{alg:rgccurves:digitloop}
  $\Forward \leftarrow \True$\tcp*[r]{indicates whether we go through $\kRGC^{d-i}$ forwards or backwards}\label{alg:rgccurves:initforward}
  \For{$i \leftarrow 0$ \KwTo $d-1$}{\label{alg:rgccurves:dimloop}
    $\pFirstDigit \leftarrow \Extract(p[\Permutation[i]])$; $\qFirstDigit \leftarrow \Extract(q[\Permutation[i]])$\;
    \If{$\Reflected[\Permutation[i]]$}{\label{alg:rgccurves:test}
      $\pFirstDigit \leftarrow (\base - 1 - \pFirstDigit)$; $\qFirstDigit \leftarrow (\base - 1 - \qFirstDigit)$\label{alg:rgccurves:reflectdigits}
    }
    \lIf{$\pFirstDigit < \qFirstDigit$}{\Return \Forward}\ 
    \lElseIf{$\pFirstDigit > \qFirstDigit$}{\Return \Not \Forward}\;\label{alg:rgccurves:return2}
    \lIf{\Forward}{$\Rank[i] \leftarrow \pFirstDigit$}\ \lElse{$\Rank[i] \leftarrow (\base - 1 - \pFirstDigit)$}\;\label{alg:rgccurves:calculaterank}
    \lIf{$\pFirstDigit$ is odd}{$\Forward \leftarrow \Not \Forward$}\label{alg:rgccurves:odddigit}
  }\label{alg:rgccurves:dimloopend}
  \For{$i \leftarrow 0$ \KwTo $d-1$}{
    $j \leftarrow \inv{a}_i(\Rank)$\;
    $\AltPermutation[i] \leftarrow \Permutation[j]$\;
    \lIf{$m_{j}(\Rank) = 1$}{$\Reflected[\AltPermutation[i]] \leftarrow \Not \Reflected[\AltPermutation[i]]$}\;\label{alg:rgccurves:reflect}
  }
  swap $\Permutation$ and $\AltPermutation$\;\label{alg:rgccurves:update}
}\label{alg:general:endloop}
\Return\False\tcp*[r]{$p$ and $q$ are equal}
\caption{Comparison operator for a $\base$-regular curve with $c(r) = \kRGC^d(r)$\label{alg:RGCcurvesoperator}}
\end{algorithm}

\section{Proof ingredients}
\label{sec:toolbox}

In this section we establish notation and basic results that are needed for the proofs in the following sections.
Readers who wish to learn about the final results and pseudocode for implementations only without studying the proofs, may skip this section.
More precisely, in Section~\ref{sec:operators} we set up notation common to most proofs; in Section~\ref{sec:visibleorders} we set up basic terminology and observations which will allow us to reason about interdimensional consistency; in Section~\ref{sec:orientations} we do the same for the concept of neutral orientation. In Section~\ref{sec:graycodeproperties} we establish basic properties of reflected Gray codes that we will need in the interdimensional-consistency proofs; this last section can also be skipped, to be read, if desired, when it is referred to later on.

\subsection{Notation for editing numbers, vectors and sequences}
\label{sec:operators}

\paragraph{The take-out operator}
If $a$ is a vector or a number, we will use $\takeout{i}a$ to denote the vector or number that results from removing the element or digit $a_i$ and renumbering the remaining elements consecutively, starting from zero. Thus, if $b = \takeout{i}a$, we have $b_j = a_j$ for $j < i$, and $b_j = a_{j+1}$ for $j \geq i$. When the take-out operator is applied to a permutation $a$, the remaining elements are numbered consecutively. Thus, we have:\[\begin{array}{ll}
\takeout{i}a_j = a_j & \hbox{if $0 \leq j < i$ and $a_j < a_i$}; \\
\takeout{i}a_j = a_j - 1 & \hbox{if $0 \leq j < i$ and $a_j > a_i$}; \\
\takeout{i}a_j = a_{j+1} & \hbox{if $i \leq j < d-1$ and $a_{j+1} < a_i$}; \\
\takeout{i}a_j = a_{j+1} - 1 & \hbox{if $i \leq j < d-1$ and $a_{j+1} > a_i$}. \\
\end{array}\]
For example, if $a = 30142$, then $\takeout{2}a = 2041$. The take-out operator can also be applied to sets and sequences: if $A$ is a set or sequence, then $\takeout{i}A$ is the set or sequence that results from applying $\takeout{i}$ to every member of the set or sequence. For example, if $A$ is the sequence of ternary numbers $\langle 010,000,210,222\rangle$, then $\takeout{2}A$ is the sequence $\langle 01,00,21,22\rangle$.

\paragraph{The reduce operator}
If $A$ is a set or sequence, then we use $\reduction{i}{j}A$ to denote the set or sequence that results from only selecting those elements $a \in A$ such that $a_i = j$, and applying $\takeout{i}$ to those elements. For example, if $A$ is the sequence of ternary numbers $\langle 010,000,210,222\rangle$, then $\reduction{2}{0}A$ is the sequence $\langle 01,00,21\rangle$.

\paragraph{The reflection operator}
If $a$ is a base-$\base$ number, then we use $\mirror{i}a$ to denote the number $a$ that results from replacing digit $a_i$ by its complement with respect to $\base - 1$. For example, if $a$ is the ternary number $20$, then $\mirror{1}a = 22$, and if $a$ is the binary number $1001$, then $\mirror{2}a = 1011$. Like the take-out operator, the reflection operator can also be applied to sets and sequences.

\paragraph{The insertion operator}
If $a$ is a vector or a number, we will use $\putin{i}{j}a$ to denote the vector or number that results from inserting an element or digit $a_i$ with value $j$, and re-indexing the original elements or digits $a_i,...$ as $a_{i+1},...$. For example, if $a$ is the ternary number 120, then $\putin{2}{1}a = 1210$. Like the take-out operator, the insertion operator can also be applied to sets and sequences.

\subsection{Visible orders of space-filling curves}
\label{sec:visibleorders}
Recall from Definition~\ref{def:consistent} that an (infinite) set $F$ of space-filling curves is \emph{interdimensionally consistent}, if $F$ contains a unique $d$-dimensional space-filling curve $f_d$ for any integer $d \geq 1$, and $f_j \in F$ is extradimensional to $f_i \in F$ whenever $j > i$. Observe that, to show that a given set $F$ of space-filling curves is interdimensionally consistent, we only need to show that each $f_j$ is extradimensional to $f_{j-1}$. Then it follows by induction that $f_j$ is also extradimensional to all $f_i$ with $j > i$.

To make the proofs in the rest of this paper easier to read, we introduce the following terminology. The \emph{front face} $F^0_i$ of a $d$-dimensional unit hypercube is the $(d-1)$-dimensional face that consists of the points $\{(q_0,...,q_d) \in [0,1]^d \mid q_i = 0\}$, and the \emph{back face} $F^1_i$ of a $d$-dimensional unit hypercube is the $(d-1)$-dimensional face that consists of the points $\{(q_0,...,q_{d-1}) \in [0,1]^d \mid q_i = 1\}$. With a \emph{face of a space-filling curve} $g$ we mean a face of the unit hypercube filled by~$g$.

The \emph{visible order} of a valid ordering $\inv{g}$ on a face $F^k_i$ of a $d$-dimensional space-filling curve $g$, is the function $\inv{f}: [0,1]^{d-1} \rightarrow [0,1]$ that is defined by the following two properties: (i)
for any pair of points $a,b \in F^k_i$ we have:\[
\inv{f}(\takeout{i}a) < \inv{f}(\takeout{i}b) \Leftrightarrow \inv{g}(a) < \inv{g}(b),\]
and (ii) the measure of the set $\{p \mid \inv{f}(p) \in [a,b]\}$ is equal to $b-a$.
The second property only serves to make $\inv{f}$ well-defined. The first property is what matters: this property defines the ordering of the points $\takeout{i}F^k_i$ induced by $\inv{f}$. If any valid ordering $\inv{f}$ of a space-filling curve $f$ can be defined as the visible order on $F^k_i$ of some valid ordering $\inv{g}$ of $g$, then we say that $g$ \emph{shows} $f$ on face $F^k_i$. 

Using the above terminology, we get the following:
\begin{lemma}\label{lem:goal}
A set $F$ of $\base$-regular space-filling curves, containing a $d$-dimensional space-filling curve $f_d$ for any $d \geq 1$, is interdimensionally consistent if and only if each curve $f_d$ (for $d > 1$) shows $f_{d-1}$ on each front face.
\end{lemma}

\subsection{Orientations}
\label{sec:orientations}
Recall from Section~\ref{sec:numerical} that for each of the $\base^d$ regions $S(r)$ in a $d$-dimensional $\base$-regular mono-curve, a permutation $a(r)$ defines the rotation component of the transformation that maps the curve through the unit hypercube to the curve in region $S(r)$. When we expand the recursion, we find subregions within regions. We can identify each such region by a sequence $R$ in the following way. Let $\mathrm{prefix}(R)$ be $R$ without its last element, let $\mathrm{last}(R)$ be the last element of $R$, and let $\langle\rangle$ be the empty sequence. Then $S(\langle\rangle)$ is the $d$-dimensional unit hypercube, and $S(R)$ is region $S(\mathrm{last}(R))$ within region $S(\mathrm{prefix}(R))$. The rotation component $a(R)$ of the transformation that maps the curve through the unit hypercube to the curve within $S(R)$, is now the permutation $a(\mathrm{last}(R))$ applied to the permutation $a(\mathrm{prefix}(R))$, where $a(\langle\rangle)$ is the identity permutation.

We say that a $\base$-regular mono-curve has neutral orientation, if, in the limit for increasing recursion depth $\depth$, all $d!$ permutations are equally frequent among the permutations $a(R)$ of the regions identified by sequences $R$ of length $\depth$. In other words, if $R$ is a long random sequence of $d$-digit base-$\base$ numbers, then the probability that $a(R)$ is any particular permutation should be $1/d!$.

\begin{lemma}\label{lem:neutral}
A $d$-dimensional $\base$-regular mono-curve has neutral orientation if and only if there is a number~$\depth$, such that each of the $d!$ possible permutations of $d$ numbers can be constructed as the composition of $\depth$ permutations from the set $\{a(0),...,a(\reps{\base-1}{d})\}$.
\end{lemma}
\begin{proof}
Let all possible permutations be numbered from $1$ to $d!$, where permutation 1 is the identity permutation.
Let $P$ be the $d! \times d!$ matrix defined by $P_{ij} = z/\base^d$ if exactly $z$ out of $\base^d$ regions $S(r)$ have the permutation that, applied after permutation $i$, results in permutation $j$. Note that for each permutation $a(r)$ used in $z$ regions $S(r)$, we put an entry with value $z/k^d$ in each row $i$ (namely at $P_{ij}$, where $j$ is the permutation that results from applying $a(r)$ to permutation $i$) and in each column $j$ (namely at $P_{ij}$, where $i$ is the permutation that results from applying $\inv{a}(r)$ to permutation $j$). Thus each row sums up to one and each column sums up to one. Let $u$ be the $d!$-vector with $u_1 = 1$ and $u_i = 0$ for all $i > 1$.

Now, if we choose a random region identified by a sequence $R$ of length $\depth$, the probability that it has permutation $i$ is given by element $i$ of the vector $P^\depth u$. As $\depth$ goes to infinity, $P^\depth u$ converges if and only if $P$ is a regular square matrix, that is, if and only if there is a $k$ such that all elements of $P^k$ are strictly positive. Note that all entries in the first column of $P^k$ are strictly positive if and only if each permutation can be constructed from the identity permutation as the composition of $k$ permutations from the set $\{a(0),...,a(\reps{\base-1}{d})\}$; in this case each permutation can actually be constructed from any other permutation and therefore all entries in all columns of $P^k$ are strictly positive.

What is left to prove is that if there is a $k$ such that all entries of $P^k$ are strictly positive, then $\lim_{\depth\rightarrow\infty} P^\depth u$ converges to a vector of which all elements have the value $1/d!$. Note that the vector to which $\lim_{\depth\rightarrow\infty} P^\depth u$ converges is the unique vector $v$ that satisfies $P v = v$. Since all rows of $P$ sum up to one, the vector of which all elements are $1/d!$ is a solution to this equation. The lemma follows.
\end{proof}

\subsection{Properties of reflected Gray codes}
\label{sec:graycodeproperties}

Recall from Section~\ref{sec:graycodes} that the reflected Gray code $\kRGC^d$ is defined as follows: $\kRGC^1$ is the sequence $\langle 0,1,...,\base-1 \rangle$, and $\kRGC^{d+1}$, for $d \geq 1$, consists of the concatenation of the sequences $a \oplus \mathrm{revifodd}(a,\kRGC^d)$ for $a = 0,1,...,\base-1$ in order. As noted before, two consecutive elements $\kRGC^d(\pred{r})$ and $\kRGC^d(r)$ differ in only one digit, and the difference in value of this digit is one. In fact, we observe that the difference is in the most significant digit in which the base-$\base$ number $r$ differs from $\pred{r}$:
\begin{lemma}\label{lem:changingdigit}
Let $C$ be a reflected Gray code with base~$\base$, indexed by base-$\base$ numbers.
If $C_i(\pred{r}) \neq C_i(r)$, then $i$ has the smallest value such that $\pred{r}_i \neq r_i$.
\end{lemma}

We will now explain some observations about reflected Gray codes with base 2 and 3 that are relevant to Sections \ref{sec:ternary} and~\ref{sec:hilbert} of this paper.

\paragraph{Ternary reflected Gray codes ($\base = 3$)}
Recall that $\kRGC^d(r)$ can be obtained from $r$ by going through the digits of $r$ in order from $r_0$ to $r_{d-1}$: whenever a digit $r_i$ is encountered such that $r_i = 1$, we `reflect' all following digits, that is, we replace $r_j$ by $\base - 1 - r_j$ for all $j > i$.
Note that if $\base = 3$, this operation leaves the digits 1 unaffected. Thus, $\TRGC^d_i(r) = 1$ if and only if $r_i = 1$, and $\TRGC^d_i(r) \in \{0,2\}$ if and only if $r_i \in \{0,2\}$.
\begin{lemma}\label{lem:ternarydigitremoval}
(i) If $r_i$ is even, then $\takeout{i}\TRGC^d(r) = \TRGC^{d-1}(\takeout{i}r)$.
(ii) If $k$ is even, $\reduction{i}{k}\TRGC^d = \TRGC^{d-1}$.
\end{lemma}
\begin{proof}
Under the above mentioned step-by-step transformation that transforms $r$ into $\TRGC^d(r)$, even digits remain even and do not affect other digits. Therefore we have: if $r_i \in \{0,2\}$ (and therefore, $\TRGC^d_i(r) \in \{0,2\}$), then $\takeout{i}\TRGC^d(r) = \TRGC^{d-1}(\takeout{i}r)$. This establishes part (i) of the lemma.

Furthermore, note that each prefix of $\TRGC^d(r)$ only depends on the corresponding prefix of $r$. Therefore, selecting all pairs $(r,\TRGC^d(r))$ such that $\TRGC^d_i(r) = k$, for certain values $i \in \{0,...,d-1\}$ and $k \in \{0,2\}$, results in selecting $\base^i$ bundles with consecutive values of $r$, such that the values of $r$ in each bundle are exactly all values that have a common prefix of $i+1$ digits, and different bundles differ in at least one of the first $i$ digits of $r$. Thus, replacing each selected pair $(r,\TRGC^d(r))$ by $(\takeout{i}r,\takeout{i}\TRGC^d(r)) = (\takeout{i}r,\TRGC^{d-1}(\takeout{i}r))$ results in a sequence ordered by $\takeout{i}r$, and thus, the second elements of the pairs constitute exactly the sequence $\TRGC^{d-1}$ in order. This establishes part (ii) of the lemma.
\end{proof}
Observe that $\TRGC^d(\reps 2d) = \reps 2d$.

\paragraph{Binary reflected Gray codes ($\base = 2$)}
Recall again that $\kRGC^d(r)$ can be obtained from $r$ by going through the digits of $r$ in order from $r_0$ to $r_{d-1}$: whenever a digit $r_i$ is encountered such that $r_i = 1$, we `reflect' all following digits, that is, we replace $r_j$ by $\base - 1 - r_j$ for all $j > i$. For ease of notation, define $r_{-1} = 0$.
If $\base=2$, one can verify easily by induction that the end result is the following:
\begin{lemma}\label{lem:XORconstruction}
$\BRGC^d_i(r) = |r_i - r_{i-1}|$.
\end{lemma}
This leads to the following lemma:
\begin{lemma}\label{lem:binarydigitremoval}
(i) If $r_i = r_{i-1}$ or $i = d-1$, then $\takeout{i}\BRGC^d(r) = \BRGC^{d-1}(\takeout{i}r)$. If $i < d-1$ and $r_i \neq r_{i-1}$, then $\takeout{i}\BRGC^d(r) = \mirror{i}\BRGC^{d-1}(\takeout{i}r)$.
(ii) For $0 \leq i < d$ we have $\reduction{i}{0}\BRGC^d = \BRGC^{d-1}$; for $0 \leq i < d-1$ we have $\reduction{i}{1}\BRGC^d = \mirror{i}\BRGC^{d-1}$. 
\end{lemma}
\begin{proof}
From Lemma~\ref{lem:XORconstruction} we get that the removal of a single digit $r_i$ does not affect the remaining digits in any of the following cases: (a) $r_i = r_{i-1}$ (and thus, $\BRGC^d_i(r) = 0$); and (b) $i = d-1$ (because there are no further digits). Thus, in these cases we would also have $\takeout{i}\BRGC^d(r) = \BRGC^{d-1}(\takeout{i}r)$. On the other hand, if we remove a digit $r_i$ such that $i < d-1$ and $r_i \neq r_{i-1}$ (and thus, $\BRGC^d_i(r) = 1$), the next digit of the Gray code changes. Thus we get $\takeout{i}\BRGC^d(r) = \mirror{i}\BRGC^{d-1}(\takeout{i}r)$.

Following the same reasoning as above in the proof of Lemma~\ref{lem:ternarydigitremoval}, selecting all pairs $(r,\BRGC^d(r))$ such that $\BRGC^d_i(r) = 0$ (for any $i$) or such that $\BRGC^d_{d-1}(r) = 1$ and replacing each selected pair $(r,\BRGC^d(r))$ by $(\takeout{i}r,\takeout{i}\BRGC^d(r)) = (\takeout{i}r,\BRGC^{d-1}(\takeout{i}r))$ results in the sequence $\BRGC^{d-1}$ in order. On the other hand, if we remove a digit $r_i$ such that $\BRGC^d_i(r) = 1$, the next digit of the Gray code changes. Thus, in this case we get: $\takeout{i}\BRGC^d(r) = \mirror{i}\BRGC^{d-1}(\takeout{i}r)$; selecting all pairs $(r,\BRGC^d(r))$ such that $\BRGC^d_i(r) = 1$ (for any $i < d-1$) and replacing each selected pair $(r,\BRGC^d(r))$ by $(\takeout{i}r,\takeout{i}\BRGC^d(r)) = (\takeout{i}r,\mirror{i}\BRGC^{d-1}(\takeout{i}r))$ results in the sequence $\mirror{i}\BRGC^{d-1}$ in order.
\end{proof}
Observe that $\BRGC^d(\reps 1d) = 1\reps 0{d-1}$.

\section{Extradimensional 3-regular curves}
\label{sec:ternary}

In this section we study 3-regular vertex-continuous mono-curves, as defined in Section~\ref{sec:classes}. We will call such curves \emph{mono-Wunderlich curves}, after Wunderlich~\cite{Wunderlich}.
Besides Peano's curve, there are many other two-dimensional mono-Wunderlich curves: Figure~\ref{fig:2DWunderlich} shows a selection of order-preserving mono-Wunderlich curves. Four of these curves we will generalize to higher dimensions: Peano's curve, the coil curve~\cite{Haverkort2D} (Luxburg's variation~1~\cite{Luxburg}), Luxburg's variation~2~\cite{Haverkort2D,Luxburg}---which we will henceforth call \emph{half-coil}---and the Meurthe curve~\cite{Haverkort2D}.

\begin{figure}
\centering
\hbox to \hsize{\hfill
\includegraphics[width=\hsize]{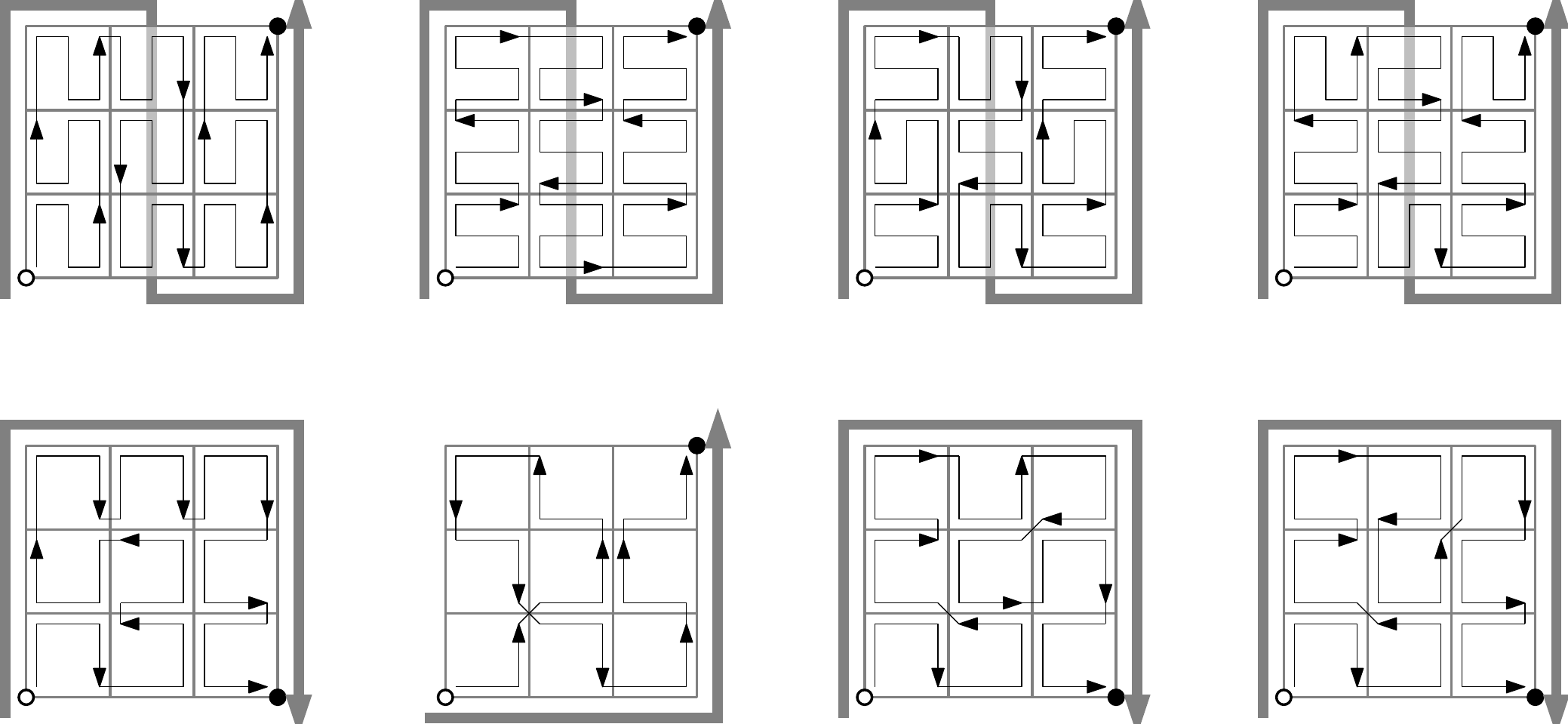}\hfill
}
\caption{Some examples of order-preserving mono-Wunderlich curves. On top we see four curves that fit Wunderlich's Serpentine framework, from left to right: Peano, coil, half-coil, and Meurthe. Below, from left to right: the {\cyrrm {Ya}}- or mirrored R-curve (Wunderlich's Meander), and three unnamed curves.}
\label{fig:2DWunderlich}
\end{figure}

\subsection{Interdimensionally consistent mono-Wunderlich curves}
Each curve described in this section is defined by a single recursive rule. Each curve traverses the $3^d$ subcubes of the unit hypercube in the order of the ternary reflected Gray code explained in Section~\ref{sec:graycodes}, that is, for each region $S(r)$ we have:
\[c(r) = \TRGC_d(r).\]
In recursion, this results in each subcube being traversed from one corner to a corner that is opposite of it in all dimensions. To make sure that these curve fragments connect up properly, we define the reflections $m(r)$ as follows:
\[m_i(r) = (\mathrm{digitsum}(r) - r_i) \bmod 2.\]
With these definitions of $c(r)$ and $m(r)$ we get a proper mono-Wunderlich curve, regardless of our choice of the permutations $a(r)$ (Lemma~\ref{lem:ternarycontinuous} in Section~\ref{sec:proofternarycontinuous} proves this). As we will prove in Section~\ref{sec:proofternary}, the following choices for $a(r)$ result in interdimensionally consistent sets of mono-Wunderlich curves:\begin{itemize}
\item generalized Peano curves: $a(r)$ is simply the identity permutation ($a_i(r) = i$), regardless of $r$;
\item generalized coil curves: $a(r)$ is the reversal permutation ($a_i(r) = d-1-i$), regardless of $r$;
\item generalized half-coil curves: $a(r)$ is the identity permutation if $r$ is odd, and the reversal permutation if $r$ is even;
\item generalized Meurthe curves:
  when $r_i \in \{0,1\}$, then $a_i(r)$ is the number of zeros and ones among $r_{i+1},...,r_{d-1}$;
  when $r_i = 2$, then $a_i(r)$ is $d-1$ minus the number of twos in $r_{i+1},...,r_{d-1}$.
  Alternatively, we can define $a(r)$ through its inverse $\inv{a}(r)$: the inverse is constructed from the identity permutation by moving all indices $i$ such that $r_i \in \{0,1\}$ to the front and reversing their order.
\end{itemize}
Table~\ref{tab:3reg3Dsets} shows the definitions of the three-dimensional versions of each of these curves.
The inverse permutations are shown explicitly only for the Meurthe curve; for the other curves, each permutation is its own inverse.
For illustration, Figure~\ref{fig:3DMeurthe} shows the graphical definition of the three-dimensional Meurthe curve.
Of the curves described above, 
in three and more dimensions only the Meurthe curves have neutral orientation---we prove this in Section~\ref{sec:proofternaryneutral}. The main results of Section~\ref{sec:ternary} can therefore be summarized as follows:

\begin{theorem}
The Peano, coil, and half-coil curves constitute interdimensionally consistent sets of order-preserving 
3-regular vertex-continuous mono-curves.
The Meurthe curves constitute an interdimensionally consistent set of order-preserving 3-regular vertex-continuous mono-curves with neutral orientation.
\end{theorem}

\begin{figure}
\centering
\hbox to \hsize{\hfill
\includegraphics[scale=0.85]{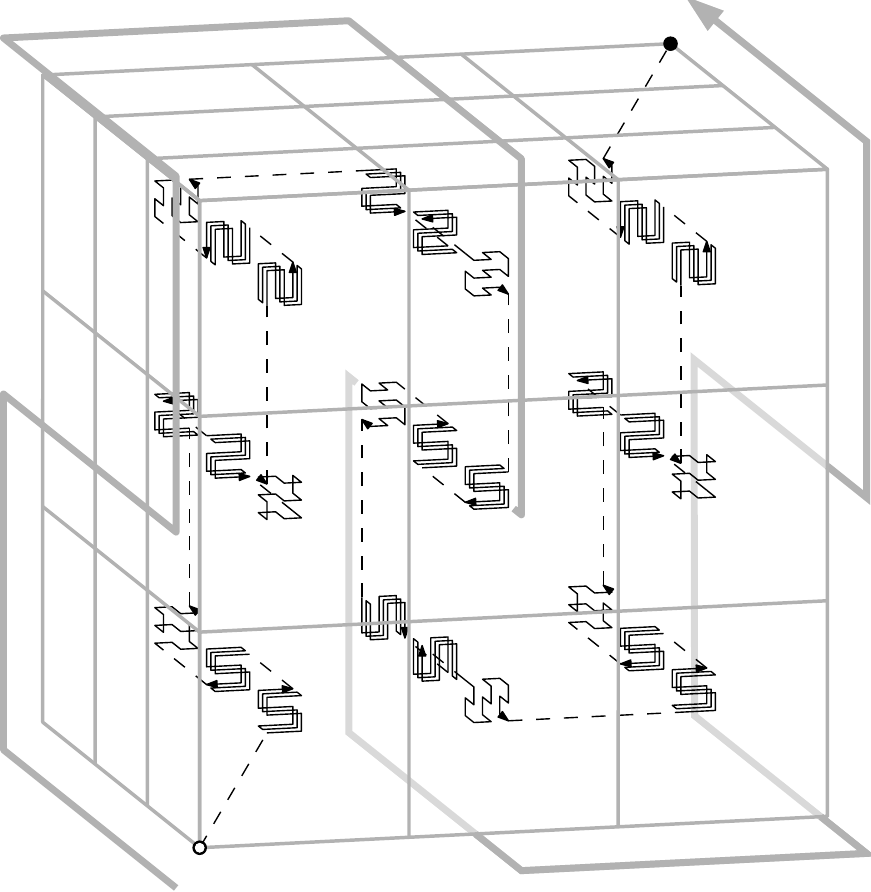}\hfill
}
\caption{The three-dimensional Meurthe curve.}
\label{fig:3DMeurthe}
\end{figure}

\begin{table}
\label{tab:3reg3Dsets}
\caption{Three-dimensional mono-Wunderlich curves from interdimensionally consistent sets}
\addvspace{.5\baselineskip}\centering\footnotesize
\begin{tabular}{@{}c@{\,}|@{\,}c@{ }c@{/}c@{/}c@{/}c@{ }c@{ }c@{ }c||c@{\,}|@{\,}c@{ }c@{/}c@{/}c@{/}c@{ }c@{ }c@{ }c@{}}
rank      & loc. & \multicolumn{4}{c}{permutation} & inv. & refl. & exit &
rank      & loc. & \multicolumn{4}{c}{permutation} & inv. & refl. & exit \\
          &      & Pea & coil & $\frac12$cl & Mt & Mt & & &
          &      & Pea & coil & $\frac12$cl & Mt & Mt & & \\
\hline
000 & 000 & 012&210&210&210 & 210 & 000 & $\frac13(1,1,1)$ & 112 & 112 & 012&210&210&102 & 102 & 110 & $\frac13(1,1,3)$ \\
001 & 001 & 012&210&012&210 & 210 & 110 & $\frac13(0,0,2)$ & 120 & 102 & 012&210&012&120 & 201 & 011 & $\frac13(2,0,2)$ \\
002 & 002 & 012&210&210&102 & 102 & 000 & $\frac13(1,1,3)$ & 121 & 101 & 012&210&210&120 & 201 & 101 & $\frac13(1,1,1)$ \\
010 & 012 & 012&210&012&210 & 210 & 101 & $\frac13(0,2,2)$ & 122 & 100 & 012&210&012&012 & 012 & 011 & $\frac13(2,0,0)$ \\
011 & 011 & 012&210&210&210 & 210 & 011 & $\frac13(1,1,1)$ & 200 & 200 & 012&210&210&210 & 210 & 000 & $\frac13(3,1,1)$ \\
012 & 010 & 012&210&012&102 & 102 & 101 & $\frac13(0,2,0)$ & 201 & 201 & 012&210&012&210 & 210 & 110 & $\frac13(2,0,2)$ \\
020 & 020 & 012&210&210&120 & 201 & 000 & $\frac13(1,3,1)$ & 202 & 202 & 012&210&210&102 & 102 & 000 & $\frac13(3,1,3)$ \\
021 & 021 & 012&210&012&120 & 201 & 110 & $\frac13(0,2,2)$ & 210 & 212 & 012&210&012&210 & 210 & 101 & $\frac13(2,2,2)$ \\
022 & 022 & 012&210&210&012 & 012 & 000 & $\frac13(1,3,3)$ & 211 & 211 & 012&210&210&210 & 210 & 011 & $\frac13(3,1,1)$ \\
100 & 122 & 012&210&012&210 & 210 & 011 & $\frac13(2,2,2)$ & 212 & 210 & 012&210&012&102 & 102 & 101 & $\frac13(2,2,0)$ \\
101 & 121 & 012&210&210&210 & 210 & 101 & $\frac13(1,3,1)$ & 220 & 220 & 012&210&210&120 & 201 & 000 & $\frac13(3,3,1)$ \\
102 & 120 & 012&210&012&102 & 102 & 011 & $\frac13(2,2,0)$ & 221 & 221 & 012&210&012&120 & 201 & 110 & $\frac13(2,2,2)$ \\
110 & 110 & 012&210&210&210 & 210 & 110 & $\frac13(1,1,1)$ & 222 & 222 & 012&210&210&012 & 012 & 000 & $\frac13(3,3,3)$ \\
111 & 111 & 012&210&012&210 & 210 & 000 & $\frac13(2,2,2)$ & \\
\end{tabular}
\end{table}

\begin{inset}
\caption{How the interdimensionally consistent mono-Wunderlich curves were found.}
\footnotesize
\label{box:howfoundternary}
It was not straightforward to guess the permutations for the interdimensionally consistent generalizations of the mono-Wunderlich curves from the two-dimensional curves, not even for the coil and half-coil curves. In two dimensions there are only two numbers to permute and one cannot tell the difference between, for example, a reversal, a left rotation, and a right rotation, or the difference between the identity permutation and the permutation that reverses the whole sequence except the first element. Even in three dimensions, one cannot tell the difference between a complete reversal and a simple swap of the first and the last element.

To get started, I worked from the assumption that I could require more than what is strictly necessary: I tried to ensure that the three-dimensional curves show the two-dimensional curves on \emph{all} faces instead of only on the front faces. Thus it was possible to deduce what the permutations in the eight corner regions had to be, and to some extent, also what the permutations in other regions had to be. Having a hypothesis for the necessary permutations in three dimensions, and with the experience gained with the harmonious Hilbert curves that are explained in Section~\ref{sec:hilbert}, it became possible to see the patterns emerge.

Those who are familiar with our previous work in two dimensions may miss two generalizations, namely interdimensionally consistent generalizations of Wunderlich's Serpentine~011\,010\,110 and R-order (Wunderlich's Meander)~\cite{Haverkort2D}. These curves seem to be considerably more difficult to generalize in an interdimensionally consistent way than the Peano, coil, half-coil and Meurthe curves. The above approach, unfortunately, results in inconsistent requirements on the permutations for the three-dimensional generalizations of Serpentine~011\,010\,110 and R-order. One would therefore have to aim for a solution in which some of the back faces do \emph{not} show the two-dimensional curve, or, for example, a mirrored version of it. \end{inset}

\subsection{Implementation of comparison operators}
\label{sec:implementternary}

In this section we discuss how to implement comparison operators for each of the four sets of curves presented above. We will first explain how to do this for points in the unit hypercube, then explain how to extend the operators to points with arbitrary non-negative coordinates, and finally how to extend the operators to points with arbitrary real coordinates. The main purpose of this section is to clarify how such comparison operators can be implemented, and to enable us to discover more properties of the curves by analysing the algorithms. For real implementations, one would have to implement the \Extract function, one would have to implement an efficient check to determine if no more digits remain to be processed, and one may consider further optimizations, such as look-up tables for permutations.

\LinesNotNumbered
\RestyleAlgo{plain}

\paragraph{Modifications for all curves}
For points in the unit hypercube we can simply fill in the details of Algorithm~\ref{alg:RGCcurvesoperator}.

To start with, we substitute ``2'' for ``$b - 1$''.

Second, we implement the reflections. The condition $m_j(\Rank) = 1$ translates to: $\Rank[j]$ is odd and $\Rank$ has even digit sum, or $\Rank[j]$ is even and $\Rank$ has odd digit sum. Therefore we implement the reflections correctly if we replace the condition on Line~\ref{alg:rgccurves:reflect} by ``\Rank has odd digit sum'', and replace Line~\ref{alg:rgccurves:odddigit} by:

\medskip
\begin{algorithm}[H]
\If{$\pFirstDigit$ is odd}{$\Forward \leftarrow \Not \Forward$; $\Reflected[\Permutation[i]] \leftarrow \Not \Reflected[\Permutation[i]]$}
\end{algorithm}

\medskip\noindent
Even better, we move Line~\ref{alg:rgccurves:initforward} to just before Line~\ref{alg:rgccurves:digitloop}, and remove Line~\ref{alg:rgccurves:reflect} altogether. The only effect of this is that in some iterations of the main loop the values of \Forward and all entries of \Reflected are `wrong', but since all of them are wrong, the `error' made on Line~\ref{alg:rgccurves:reflectdigits} will be undone on Line \ref{alg:rgccurves:return2} or~\ref{alg:rgccurves:calculaterank} (whichever is applicable).

A last simplification we make, is that we will work with reflected versions of \pFirstDigit and \qFirstDigit whenever \Forward is \False. This is implemented by changing the condition in Line~\ref{alg:rgccurves:test} to $(\Reflected[\Permutation[i]] = \Forward)$; changing the return values on Line~\ref{alg:rgccurves:return2} to \True and \False, respectively; and changing Line~\ref{alg:rgccurves:calculaterank} to simply $\Rank[i] \leftarrow \pFirstDigit$. This may change the value of $\pFirstDigit$ as used on Line~\ref{alg:rgccurves:odddigit}, but it does not change the parity of $\pFirstDigit$, which is all that matters.

Thus we get Algorithm~\ref{alg:ternarycurvesoperator}.

\LinesNumbered
\RestyleAlgo{ruled}
\begin{algorithm}
\lFor{$i \leftarrow 0$ \KwTo $d-1$}{$\Reflected[i] \leftarrow \False$}\;
\lFor{$i \leftarrow 0$ \KwTo $d-1$}{$\Permutation[i] \leftarrow i$\tcp*[r]{maintains which entry in each row of $T$ is non-zero}}\label{alg:ternarycurves:identity}
initialize an array $\AltPermutation[0..d-1]$\tcp*[r]{for a new value of \Permutation under construction}\label{alg:ternarycurves:altpermutation}
$\Forward \leftarrow \True$\;
\Repeat{all remaining digit strings $p[0],...,p[d-1]$ and $q[0],...,q[d-1]$ are empty}{
  \For{$i \leftarrow 0$ \KwTo $d-1$}{\label{alg:ternarycurves:dimloop}
    $\pFirstDigit \leftarrow \Extract(p[\Permutation[i]])$; $\qFirstDigit \leftarrow \Extract(q[\Permutation[i]])$\;
    \If{$\Reflected[\Permutation[i]] = \Forward$}{
      $\pFirstDigit \leftarrow (2 - \pFirstDigit)$; $\qFirstDigit \leftarrow (2 - \qFirstDigit)$\label{alg:ternarycurves:reflect}
    }
    \lIf{$\pFirstDigit < \qFirstDigit$}{\Return \True}\ \lElseIf{$\pFirstDigit > \qFirstDigit$}{\Return \False}\;
    $\Rank[i] \leftarrow \pFirstDigit$\;\label{alg:ternarycurves:updaterank}
    \If{$\pFirstDigit$ is odd}{
      $\Forward \leftarrow \Not \Forward$;
      $\Reflected[\Permutation[i]] \leftarrow \Not \Reflected[\Permutation[i]]$\label{alg:ternarycurves:odddigit}
    }\label{alg:ternarycurves:dimloopend}
  }
  \lFor{$i \leftarrow 0$ \KwTo $d-1$}{$\AltPermutation[i] \leftarrow \Permutation[\inv{a}_i(\Rank)]$}\;\label{alg:ternarycurves:rotate}
  swap $\Permutation$ and $\AltPermutation$\;\label{alg:ternarycurves:update}
}\label{alg:general:endloop}
\Return\False\tcp*[r]{$p$ and $q$ are equal}
\caption{Framework for implementation of Peano, coil, half-coil and Meurthe curves\label{alg:ternarycurvesoperator}}
\end{algorithm}

\LinesNotNumbered
\RestyleAlgo{plain}

\paragraph{Peano curves}
Since $\inv{a}_i(\Rank) = i$, regardless of the value of $\Rank$, we can simply omit Lines \ref{alg:ternarycurves:identity}, \ref{alg:ternarycurves:altpermutation}, \ref{alg:ternarycurves:updaterank}, \ref{alg:ternarycurves:rotate} and~\ref{alg:ternarycurves:update} from Algorithm~\ref{alg:ternarycurvesoperator}, and replace each occurrence of ``$\Permutation[i]$'' by ``$i$''.

\paragraph{Coil curves}
Since $\inv{a}_i(\Rank) = d-1-i$, regardless of the value of $\Rank$, we can implement the updating of $\Permutation$ as follows. Omit Lines \ref{alg:ternarycurves:updaterank} and~\ref{alg:ternarycurves:rotate}, and replace Line~\ref{alg:ternarycurves:altpermutation} by:

\medskip
\begin{algorithm}[H]
\lFor{$i \leftarrow 0$ \KwTo $d-1$}{$\AltPermutation[i] \leftarrow (d-1-i)$}
\end{algorithm}

\paragraph{Half-coil curves}
In half-coil curves, the permutation is reversed if and only if \Rank is even, that is, if and only if \Forward remains unchanged. We implement this as follows. Omit Lines \ref{alg:ternarycurves:updaterank} and~\ref{alg:ternarycurves:rotate}. Before Line~\ref{alg:ternarycurves:dimloop}, add:

\medskip
\begin{algorithm}[H]
$\Direction \leftarrow \Forward$\;
\end{algorithm}

\medskip\noindent Replace Line~\ref{alg:ternarycurves:update} by:

\medskip
\begin{algorithm}[H]
\lIf{\Forward = \Direction}{swap \Permutation and \AltPermutation}
\end{algorithm}

\paragraph{Meurthe curves}
Replace Line~\ref{alg:ternarycurves:rotate} by:

\medskip
\begin{algorithm}[H]
\SetKwData{NextTwo}{indexForNextTwo}
\SetKwData{NextNonTwo}{indexForNextZeroOrOne}
$\NextTwo \leftarrow d$; \lFor{$i \leftarrow 0$ \KwTo $d-1$}{\lIf{$\Rank[i] = 2$}{decrement \NextTwo}}\;
$\NextNonTwo \leftarrow \NextTwo - 1$\;
\For{$i \leftarrow 0$ \KwTo $d-1$}{
  \If{$\Rank[i] = 2$}{
    $\AltPermutation[\NextTwo] \leftarrow \Permutation[i]$; increment \NextTwo
  }
  \Else($\Rank[i]$ is 0 or 1){
    $\AltPermutation[\NextNonTwo] \leftarrow \Permutation[i]$; decrement \NextNonTwo
  }
}
\end{algorithm}

\medskip\noindent Of course the first \textbf{for} loop can be avoided by integrating it into the loop on Line~\ref{alg:ternarycurves:dimloop}--\ref{alg:ternarycurves:dimloopend} of Algorithm~\ref{alg:ternarycurvesoperator}.
When implementing this curve, make sure that the algorithm applies $\inv{a}$ (like the above piece of code does), not $a$. A mistake is easily made and can be hard to detect: as can be seen in Table~\ref{tab:3reg3Dsets}, in the three-dimensional Meurthe curve $a(r)$ equals $\inv{a}(r)$ for most, but not all, values of $r$.

\paragraph{Points with non-negative coordinates}
We can extend our implementation so that it can also compare points outside the unit hypercube. To this end, we should extend the space-filling curve outside the unit hypercube. This can be done by considering the unit hypercube as the first subregion of a larger cube. The easiest is to consider it as the first subregion of the first subregion, because for all four sets of curves, the first subregion of the first subregion is neither reflected nor rotated. Thus, the outcome of a comparison between two points that lie in the unit hypercube already, remains the same. Applying this idea recursively, we get the following algorithm:

\medskip
\begin{algorithm}[H]
\lWhile{any coordinate of $p$ or $q$ is at least 1}{divide all coordinates of $p$ and $q$ by 9}\;
run Algorithm~\ref{alg:ternarycurvesoperator} on $p$ and $q$
\end{algorithm}

\paragraph{Points with negative coordinates}
The above implementation does not work for points with negative coordinates. We could try to solve this by considering the unit hypercube to be the subregion in the centre of a large cube. However, in that case the entrance point of the larger cube does not lie in the origin anymore, and we loose the property that the curve traverses points in an axis-parallel hyperplane through the origin in the same order as the lower-dimensional curve. We can solve this as follows.

Note that, for the coordinates common to all points of an axis-parallel hyperplane through the origin, the implementations of the Peano, coil and half-coil curves effectively do not distinguish between the digits zero and two: in Line~\ref{alg:ternarycurves:reflect}, zeros may become twos and vice versa, but as these digits remain even and remain the same between $p$ and $q$, this does not affect the decisions taken by the algorithm, and they do not affect \Forward and \Reflected. This is also the case for the implementation of the Meurthe curves: a digit zero or two for coordinate $i$ only affects where this coordinate ends up between the other coordinates in the permutation on Line~\ref{alg:ternarycurves:rotate}, but it does not affect the permutation of the other coordinates relative to each other. This is all that matters, since coordinate $i$, being shared by $p$ and $q$, is not going to decide which of $p$ or $q$ comes first in the order in any case. Therefore, all four sets of mono-Wunderlich curves presented in this section remain consistent even if the origin of the coordinate system is translated into the hypercube by any vector whose coordinates, in ternary representation, only contain zeros and twos.

Therefore, to allow for negative coordinates, we could consider, for $k = 0,1,2,...$, each hypercube $U_k = [\frac{1}{4}-\frac{1}{4}81^k,\frac{1}{4}+\frac{3}{4}81^k)^d$ to be the last subregion of the first subregion of the last subregion of the first subregion of the hypercube $U_{k+1} = [\frac{1}{4}-\frac{1}{4}81^{k+1},\frac{1}{4}+\frac{3}{4}81^{k+1})^d$ (note that the last of the first of the last of the first subregion is never reflected nor rotated, in all four sets of curves). Using ternary numbers, we write the coordinates of $U_k$ as $[\frac{1}{11}-\frac{1}{11}10\,000^k,\frac{1}{11}+\frac{10}{11}10\,000^k)^d$. We will continue to use ternary numbers in the rest of this paragraph until further notice. Define $s_0(x) = x$ and let $s_1(x) = (x+202)/10\,000$ be the function that maps the coordinate range $[\frac1{11}-\frac1{11} 10\,000^k, \frac1{11} + \frac{10}{11} 10\,000^k]$ of $U_k$ to the coordinate range $[\frac1{11}-\frac1{11} 10\,000^{k-1}, \frac1{11} + \frac{10}{11} 10\,000^{k-1}]$ of $U_{k-1}$. Let $s_k(x)$, for $k > 1$, be given by $s_k(x) = s_1(s_{k-1}(x))$. We can now define a local coordinate system for each hypercube $U_k$, were $s_k(x)$ gives the value in the local coordinate system that corresponds to $x$ in the global coordinate system.
The origin of the unit hypercube $U_0$ now lies at position $(s_k(0), s_k(0), ...)$ in the local coordinate system of hypercube $U_k$.
Note that $s_k(0)$ is the number $0.\reps{0202}{k}$. Thus, the origin lies at coordinates that only contain zeros and twos in the local coordinate system of each hypercube $U_k$, and hence, the space-filling curve filling $U_k$ is extradimensional to the lower-dimensional space-filling curves with respect to this origin.

The following algorithm (with numbers in decimal format again) zooms out to the smallest hypercube $U_k$ that contains the points to be compared, while converting the coordinates to the local coordinate system of $U_k$:

\medskip
\begin{algorithm}[H]
\While{any coordinate of $p$ or $q$ is less than zero or more than one}{
  add 20 to all coordinates of $p$ and $q$\tcp*[r]{adding the ternary number 202}
  divide all coordinates of $p$ and $q$ by 81\tcp*[r]{dividing by the ternary number 10\,000}
}
run Algorithm~\ref{alg:ternarycurvesoperator} on $p$ and $q$
\end{algorithm}

\medskip\noindent
Note that this solution produces different results than the solution for points with non-negative coordinates that was presented before.

\subsection{Proof of vertex-continuity}
\label{sec:proofternarycontinuous}

\begin{lemma}\label{lem:ternarycontinuous}
The generalizations of the Peano, coil, half-coil and Meurthe curves defined above, indeed constitute mono-Wunderlich curves.
\end{lemma}
\begin{proof}
All we have to prove is that the curves are vertex-continuous.

Recall that the curves start in the region with location $\reps 0d$, and end in the region with location $\reps 2d$. Thus, for the first and the last subregion $r$, we have $m_i(r) = 0$ for all $0 \leq i < d$. As a result, in the recursion within $S(0)$, the curve will always start in the region closest to the origin, and in the recursion within $S(\reps 2d)$, the curve will always end in the region farthest from the origin. Thus the entrance gate of the curve is at the origin, and the exit gate is at the corner of the $d$-dimensional hypercube that is farthest from the origin.

Thus, for the entrance gate of a subregion $S(r)$ to be connected to the exit gate of the previous subregion $S(\pred{r})$, region $S(r)$ should be mirrored in all dimensions with respect to $S(\pred{r})$, except in the one dimension in which the locations of $S(r)$ and $S(\pred{r})$ differ. More precisely, (i) if $c_i(\pred{r}) < c_i(r)$, we should have $m_i(\pred{r}) = m_i(r) = 0$; (ii) if $c_i(\pred{r}) = c_i(r)$, we should have $m_i(\pred{r}) + m_i(r) = 1$, and (iii) if $c_i(\pred{r}) > c_i(r)$, we should have $m_i(\pred{r}) = m_i(r) = 1$.

First consider case (i): $c_i(\pred{r}) < c_i(r)$. Since $c_i(\pred{r}) \neq c_i(r)$, we have that $i$ has the smallest value such that $\pred{r}_i \neq r_i$, and therefore, $r_i = \pred{r}_i + 1$. We have $\pred{r}_0...\pred{r}_{i-1} = r_0...r_{i-1}$. Moreover, since $c_i(\pred{r}) < c_i(r)$, we have, by definition, $\TRGC_i(\pred{r}) < \TRGC_i(r)$, and the number of reflections caused by the digits $r_0,...,r_{i-1}$, that is, the number of odd digits among $\TRGC_0(r),...,\TRGC_{i-1}(r)$, must be even. Since $\TRGC_j(r)$ is odd if and only if $r_j$ is odd, for any $j \in \{0,...,d-1\}$, this implies that the number of odd digits among $r_0,...,r_{i-1}$ must be even. Therefore, $\sum_{j=0}^{i-1}\pred{r}_j = \sum_{j=0}^{i-1}r_j \equiv 0 \pmod 2$. By definition, $r = \pred{r} + 1$, and therefore, since $r$ and $\pred{r}$ are ternary numbers, we have $\pred{r}_{i+1}...\pred{r}_{d-1} = \reps2{d-1-i}$ and $r_{i+1}...r_{d-1} = \reps0{d-1-i}$, so $\sum_{j=i+1}^{d-1}\pred{r}_j \equiv \sum_{j=i+1}^{d-1}r_j \equiv 0 \pmod 2$. Recall that $m_i(r)$ is defined as $(\mathrm{digitsum}(r) - r_i) \bmod 2$, so we have $m_i(\pred{r}) = (\sum_{j=0}^{i-1}\pred{r}_j + \sum_{j=i+1}^{d-1}\pred{r}_j) \bmod 2 = 0 = (\sum_{j=0}^{i-1}r_j + \sum_{j=i+1}^{d-1}r_j) \bmod 2 = m_i(r)$, QED.

In case (ii), $c_i(\pred{r}) = c_i(r)$, we know that $i$ is \emph{not} the smallest value $j$ such that $r_j \neq \pred{r}_j$. We have either $i < j$ and $r_i = \pred{r}_i$, or $i > j$ and $\pred{r}_i = 2$ and $r_i = 0$; in both cases $r_i \equiv \pred{r}_i \pmod 2$. Note that for ternary numbers, $r = \pred{r} + 1$ implies $\mathrm{digitsum}(r) \equiv \mathrm{digitsum}(\pred{r}) + 1 \pmod 2$. Thus we have $m_i(r) \equiv \mathrm{digitsum}(r) - r_i \equiv \mathrm{digitsum}(\pred{r}) + 1 - r_i \equiv \mathrm{digitsum}(\pred{r}) + 1 - \pred{r}_i \equiv m_i(\pred{r}) + 1 \pmod 2$, which implies $m_i(r) = (m_i(\pred{r}) + 1) \bmod 2$, QED.

Third, consider case (iii): $c_i(\pred{r}) > c_i(r)$. The proof is analogous to case (i), except that now, since $c_i(\pred{r}) > c_i(r)$, the number of odd digits among $r_0...r_{i-1}$ must be odd, and thus we get $m_i(\pred{r}) = (\sum_{j=0}^{i-1}\pred{r}_j + \sum_{j=i+1}^{d-1}\pred{r}_j) \bmod 2 = 1 = (\sum_{j=0}^{i-1}r_j + \sum_{j=i+1}^{d-1}r_j) \bmod 2 = m_i(r)$, QED.
\end{proof}

\subsection{Proof of interdimensional consistency}
\label{sec:proofternary}
In this section we will prove that the above mentioned generalizations of the Peano, coil, half-coil and Meurthe curves are interdimensionally consistent. Arguments to this effect were also given in the paragraph of negative coordinates in Section~\ref{sec:implementternary}. However, those arguments depend on the specific tie-breaking rule for points that lie on the boundary between regions, that is, the choice of $\inv{f}$, that is implemented in the pseudocode in this paper. Moreover, we need to develop a different proof technique, because in Section~\ref{sec:hilbert}, a look at the pseudocode will not reveal the secrets of the curves as easily. Therefore, in the present section, we will prove interdimensional consistency directly from the definitions of the curves, rather than from the implementation in pseudocode. I will first describe the part of the proof that is common to all four generalizations, and after that we will discuss the details of each separate generalization.

Consider any of the four sets of space-filling curves defined above, containing a $d$-dimensional space-filling curve $f_d$ for any $d \geq 1$. As observed in Lemma~\ref{lem:goal}, to prove that the set is interdimensionally consistent, it is sufficient to show that each curve $f_d$ (for $d > 1$) shows $f_{d-1}$ on each front face. In fact, we will prove something stronger, namely that each curve $f_d$ shows $f_{d-1}$ on \emph{each} face.

Below, we write $f$ for $f_d$, and by $c$, $a$, $m$, $o$, $M$, and $\tau$ we denote the location, permutation, reflection, translation, transformation matrix and transformation functions that define $f$.
We write $f'$ for $f_{d-1}$, and by $c'$, $a'$, $m'$, $o'$, $M'$ and $\tau'$ we denote the location, permutation, reflection, translation, transformation matrix and transformation functions that define $f'$.
Our proof that each curve $f$ shows $f'$ on each face goes by induction on increasing level of refinement. More precisely, for a level of refinement $\ell$, we consider the recursive balanced subdivision of the $(d-1)$-dimensional unit hypercube $U'$ into $3^{\ell(d-1)}$ regions, and prove the following induction hypothesis: the visible order of $f$ on any face $F' = F^k_i$ visits these regions in the same order as $f'$.

\paragraph{Base case $\ell=1$}
As a base case, consider refinement level $\ell = 1$, and a face $F' = F^k_i$. Recall that $f$ is defined based on subdivision of a $d$-dimensional unit hypercube---which we will denote $U$---into $3^d$ regions $S(0),...,S(\reps 2d)$. This induces a subdivision of $F'$ into $3^{d-1}$ regions $F'(0),...,F'(\reps 2{d-1})$, each of which is a $(d-1)$-dimensional face of one of the regions $S(0),...,S(\reps 2d)$. Let $F'(0),...,F'(\reps 2{d-1})$ be indexed in the order in which the $d$-dimensional regions that contain them appear in $S(0),...,S(\reps 2d)$. Thus, there is a monotonously increasing function $\sigma: \{0,...,\reps 2{d-1}\} \rightarrow \{0,...,\reps 2d\}$ such that $F'(r')$ is a face of $S(\sigma(r'))$. Note that for $r' \in \{0,...,\reps 2{d-1}\}$, the locations of the regions $F'(r')$ within $U'$ are given by $\takeout{i}c(\sigma(r')) = \takeout{i}\TRGC^d(\sigma(r'))$. Furthermore, the co-domain of $\sigma$ consists of exactly those values $r$ such that $c_i(r) = \TRGC^d_i(r) = 2k$. Thus, the locations of the regions $F'(r')$ within $U'$, in order of increasing~$r'$, form the sequence $\reduction{i}{2k}\TRGC^d$, which, by Lemma~\ref{lem:ternarydigitremoval}, is exactly the sequence $\TRGC^{d-1} = c'$. Hence, the visible order of $f$ on any face $F' = F^k_i$ visits the regions of refinement level $\ell = 1$ in the same order as~$f'$.

\paragraph{Induction}
We will now prove that if the induction hypothesis holds for refinement level $\ell-1$, then it also holds for refinement level $\ell$. Again, consider a face $F' = F^k_i$, subdivided into $3^{d-1}$ top-level regions $F'(0),...,F'(\reps 2{d-1})$, each subdivided into $3^{(\ell-1)(d-1)}$ smaller regions. By the analysis of the base case, the $3^{d-1}$ top-level regions $F'(0),...,F'(\reps 2{d-1})$ of $F'$ are traversed in the correct order. What is left to prove, is that for each $r' \in \{0,...,\reps 2{d-1}\}$, the visible order of the region $F'(r')$ traverses its recursive subdivision into $3^{(\ell-1)(d-1)}$ smaller regions in the correct order.

The transformation $\tau(r)$ that transforms the curve $f$ through $U$ into the section of $f$ through $S(\sigma(r'))$, also transforms a face $F$ of $U$ into $F'(r')$. Denote this face $F$ by $F^*(r')$, denote $a_i(\sigma(r'))$ by $j$, and observe that we have:\[F^*(r') = F^{|k - m_i(\sigma(r'))|}_{j}.\]
Note that by the inductive assumption $F^*(r')$ shows the correct order $f'$ of its subregions of level $\ell-1$ (which map to subregions of level $\ell$ on $F'$); all that is left to prove is that the transformation $\tau(\sigma(r'))$ that operates on $S(\sigma(r'))$ transforms $f'$ on $F^*(r')$ correctly into $f'$ on $F'(r')$, that is, we need to show \[\takeout{i}\tau(\sigma(r'))(\putin{j}{2k}U') = \tau'(r')(U'),\]
where $\putin{j}{2k}U'$ is the operation that inserts a coordinate $p_j = 2k$ in the coordinate sequence of each point $p \in U'$. Recall that $\tau(r)$ is the function $p \rightarrow \frac 1k M(r) p + o(r)$, so what we have to prove is:
\[\takeout{i}(\frac 1k M(\sigma(r')) (\putin{j}{2k}p) + \takeout{i}o(\sigma(r'))) = \frac 1k M'(r') p + o'(r').\]
By the definition of $s(r')$ the translations $o(\sigma(r'))$ and $o'(r')$ are such that $F^*(r')$ is mapped to $F'(r')$, so we do not need to worry about those. Our challenge is to show that the matrices $M(\sigma(r'))$ and $M'(r')$ match, that is, $M'(r')$ must equal $M(\sigma(r'))$ with row $i$ and column $j$ removed. Given the definition of these matrices, we can separate this challenge into two subproblems. First, the reflections should match:\begin{equation}\label{eq:matchreflections}
\takeout{i}m(\sigma(r')) = m'(r'),
\end{equation}
and second, the permutations should match:\begin{equation}\label{eq:matchpermutations}
\takeout{i}a(\sigma(r')) = a'(r').
\end{equation}

We start with the reflections. To prove is Equation~\ref{eq:matchreflections}.
Let $r$ be $\sigma(r')$. By definition of $r = \sigma(r')$, we have $c(r) = \putin{i}{2k}c(r')$, and hence, $\TRGC^d_i(r) = 2k$ and $\takeout{i}\TRGC^d(r) = \TRGC^{d-1}(r')$. Recall that $\TRGC^d_i(r)$ is even if and only if $r_i$ is even. Therefore, $r_i$ is even, and by Lemma~\ref{lem:ternarydigitremoval}, part~(i), we have:
\begin{equation}\label{eq:reductionofsigma}
r' = \takeout{i}r = \takeout{i}\sigma(r'),
\end{equation}
and:
\begin{equation}\label{eq:maintainrankparity}
r' \equiv r = \sigma(r') \pmod 2.
\end{equation}
By the definition of the curves, $m'_j(r') \equiv \mathrm{digitsum}(r') - r'_j \pmod 2$, and because $r'$ is a ternary number, we have $\mathrm{digitsum}(r') \equiv r' \pmod 2$. Now we have, for $j < i$,
$m'_j(r') \equiv \mathrm{digitsum}(r') - r'_j \equiv r' - r'_j \equiv \sigma(r') - \sigma_j(r') \equiv m_j(\sigma(r')) \pmod 2$, and for $j \geq i$,
$m'_j(r') \equiv \mathrm{digitsum}(r') - r'_j \equiv r' - r'_j \equiv \sigma(r') - \sigma_{j+1}(r') \equiv m_{j+1}(\sigma(r')) \pmod 2$.
Thus $m'(r') = \takeout{i}m(\sigma(r'))$, QED.

Now all that is left to prove is that the permutation components are correct, that is, we need to prove Equation~\ref{eq:matchpermutations}. Below we prove this separately for each of the four sets of curves described above.

\paragraph{Peano curves}
Since $a'(r')$ and $a(\sigma(r'))$ maintain the order of all coordinate axes regardless of $r'$ and $\sigma(r')$, we always have $a'(r') = \takeout{i}a(\sigma(r'))$, QED.

\paragraph{Coil curves}
Since $a'(r')$ and $a(\sigma(r'))$ reverse the order of all axes regardless of $r'$ and $\sigma(r')$, we always have $a'(r') = \takeout{i}a(\sigma(r'))$, QED.

\paragraph{Half-coil curves}
The permutation $a'(r')$ reverses the order of all axes if $r$ is odd, and maintains the order of all axes if $r'$ is even.
The permutation $a(\sigma(r'))$ reverses the order of all axes if $\sigma(r')$ is odd, and maintains the order of all axes if $\sigma(r')$ is even. By Equation~\ref{eq:maintainrankparity} we have that $r'$ is odd if and only if $\sigma(r')$ is odd, and thus, the permutations $a'(r')$ and $a(\sigma(r'))$ match, that is, $a'(r') = \takeout{i}a(\sigma(r'))$, QED.

\paragraph{Meurthe curves}
We will again use Equation~\ref{eq:reductionofsigma}. Recall that the respective inverses of $a'(r')$ and $a(r)$ (where $r = \sigma(r')$) can be constructed from the identity permutation by moving all indices $h$ such that $r'_h$, or $r_h$, respectively, is 0 or 1 to the front, and reversing the order of these indices. Obviously, inserting a digit in $r'$ at position $i$ to obtain $r = \sigma(r')$, moving the index $i$ in $\inv{a}(r)$, and subsequently removing index $i$ from $\inv{a}(r)$ (which is what $\takeout{i}a(r)$ does), does not affect the order of the remaining indices relative to each other. Thus we have $a'(r') = \takeout{i}a(\sigma(r'))$, QED.

\subsection{Orientation of the interdimensionally consistent mono-Wunderlich curves}
\label{sec:proofternaryneutral}

By Lemma~\ref{lem:neutral}, a $d$-dimensional mono-Wunderlich curve has neutral orientation if and only if there is a number $\depth$, such that each of the $d!$ possible permutations of $d$ numbers can be constructed as the composition of $\depth$ permutations from $a$. Trivially, each 1-dimensional curve has neutral orientation. Below we analyse whether the $d$-dimensional curves for $d \geq 2$ have neutral orientation.

\paragraph{Peano curves} In the Peano curves, $a$ only contains the identity permutation; therefore no other permutation can be constructed and the Peano curves do \emph{not} have neutral orientation.

\paragraph{Coil curves} In the coil curves, $a$ only contains the reversal permutation. As a result, there is no $\depth$ such that each permutation can be constructed as the composition of $\depth$ permutations from $a$, not even if $d = 2$: if $\depth$ is odd, only the reversal permutation can be constructed, and if $\depth$ is even, only the identity permutation can be constructed. Therefore the coil curves do \emph{not} have neutral orientation.

\paragraph{Half-coil curves} In the half-coil curves, $a$ only contains the identity and the reversal permutation. As a result, there is no $\depth$ such that each permutation can be constructed as the composition of $\depth$ permutations from $a$, unless $d = 2$. Thus, the two-dimensional half-coil curve has neutral orientation, but the higher-dimensional half-coil curves have not.

\paragraph{Meurthe curves} In the Meurthe curves, $a(00\reps{2}{d-2})$ is the permutation that swaps the first two elements, and $a(\reps{2}{d-1}0)$ is the permutation that rotates the whole sequence one step to the right. Thus, we can swap any two adjacent elements by composing at most $d+1$ permutations: $d$ rotations $a(\reps{2}{d-1}0)$ and one swap $a(00\reps{2}{d-2})$, placed between the rotations at the point where the two elements to be swapped have been rotated to the first two positions. Since any permutation can be created from any other by at most $d(d-1)/2$ swaps of adjacent elements, we get that any permutation can be created from any other by composing at most $(d+1)d(d-1)/2$ permutations $a(00\reps{2}{d-2})$ and $a(\reps{2}{d-1}0)$. However, to satisfy the condition in Lemma~\ref{lem:neutral}, there should be a depth $\depth$ such that each permutation can be created from the identity permutation by composing \emph{exactly} $\depth$ permutations. In the Meurthe curves, this can indeed be done for $\depth = (d+1)d(d-1)/2$, since any shorter sequence of permutations can be completed by the identity permutation found at $a(\reps{2}{d})$. If follows that the Meurthe curves have neutral orientation.

Note that the analysis given above is by no means tight: in reality much less than $(d+1)d(d-1)/2$ permutations from $a$ may suffice to produce all possible permutations. For example, in four dimensions two permutations suffice, rather than $(d+1)d(d-1)/2 = 30$.

\section{Extradimensional 2-regular curves}
\label{sec:hilbert}

In this section we study 2-regular vertex-continuous mono-curves, as defined in Section~\ref{sec:classes}. We will call such curves \emph{mono-Hilbert curves}~\cite{Haverkort3D}. In two dimensions, the properties that define a mono-Hilbert curve form a minimal sufficient set of properties that defines Hilbert's original curve (see Figure~\ref{fig:original-curves}). Therefore one can argue, although it is ultimately a matter of taste, that in higher dimensions these properties are necessary and sufficient for a curve to be considered a true generalization of a Hilbert curve~\cite{Haverkort3D}. Below we will discuss two sets of mono-Hilbert curves in arbitrary dimensions: the generalization implemented by Moore~\cite{Moore} (citing Butz~\cite{Butz}), and a new generalization called \emph{harmonious Hilbert curves}.

\subsection{Standard Hilbert curves}
\label{sec:standard}
We will first describe a class of Hilbert curves in higher dimensions, in which the order of the subregions and the connection points between them follow a very regular pattern. We call these curves \emph{standard Hilbert curves} (by lack of a better name).

Each standard Hilbert curve is defined by a single recursive rule. Each curve traverses the $2^d$ subcubes of the unit hypercube in the order of the binary reflected Gray code explained in Section~\ref{sec:graycodes}, that is, for each region $S(r)$ we have:
\[c(r) = \BRGC_d(r).\]
Standard Hilbert curves differ in their permutation functions $a$, but $\inv{a}_0$ is fixed, as follows (recall that we write $\pred{r}$ for $r - 1$, and we write $\succ(r)$ for $r + 1$):\[\begin{array}{ll}
a_{d-1}(r) = 0 & \hbox{if $r = 0$ (the first subregion) or $r = \reps 1d$ (the last subregion)};\\
a_i(r) = 0     & \hbox{if $0 < r < \reps 1d$, $r$ is odd and $c_i(r) \neq c_i(\succ{r})$};\\
a_i(r) = 0     & \hbox{if $0 < r < \reps 1d$, $r$ is even and $c_i(r) \neq c_i(\pred{r})$}.\\
\end{array}\]
Reflections are defined as follows:\[\begin{array}{ll}
m_i(r) = c_i(r)        & \hbox{if $r = 0$};\\
m_i(r) = 1 - c_i(r)    & \hbox{if $r > 0$ and $i = d-1$};\\
m_i(r) = c_i(\pred{r}) & \hbox{if $r > 0$ and $i < d-1$}.\\
\end{array}\]

\noindent We consider two sets of standard Hilbert curves:\begin{itemize}
\item Butz-Moore curves, as implemented by Moore~\cite{Moore}, based on the work of Butz~\cite{Butz}: here $\inv{a}_j(r) = (\inv{a}_0(r) + j) \bmod d$.
\item our new \emph{harmonious Hilbert curves}:
  when $r_i \neq r_{d-1}$, then $a_i(r)$ is the number of digits among $r_{i+1},...,r_{d-1}$ that differ from $r_{d-1}$;
  when $r_i = r_{d-1}$, then $a_i(r)$ is $d-1$ minus the number of digits among $r_0,...,r_{i-1}$ that equal $r_{d-1}$.
  Alternatively, we can define $a(r)$ through its inverse $\inv{a}(r)$: the inverse is constructed from the identity permutation by moving all indices $i$ such that $r_i \neq r_{d-1}$ to the front and reversing their order, and reversing the order of the remaining indices as well.
\end{itemize}

\begin{figure}
\centering
\hbox to \hsize{\hfill
\includegraphics[scale=0.85]{A26-00-db}\hfill\includegraphics[scale=0.85]{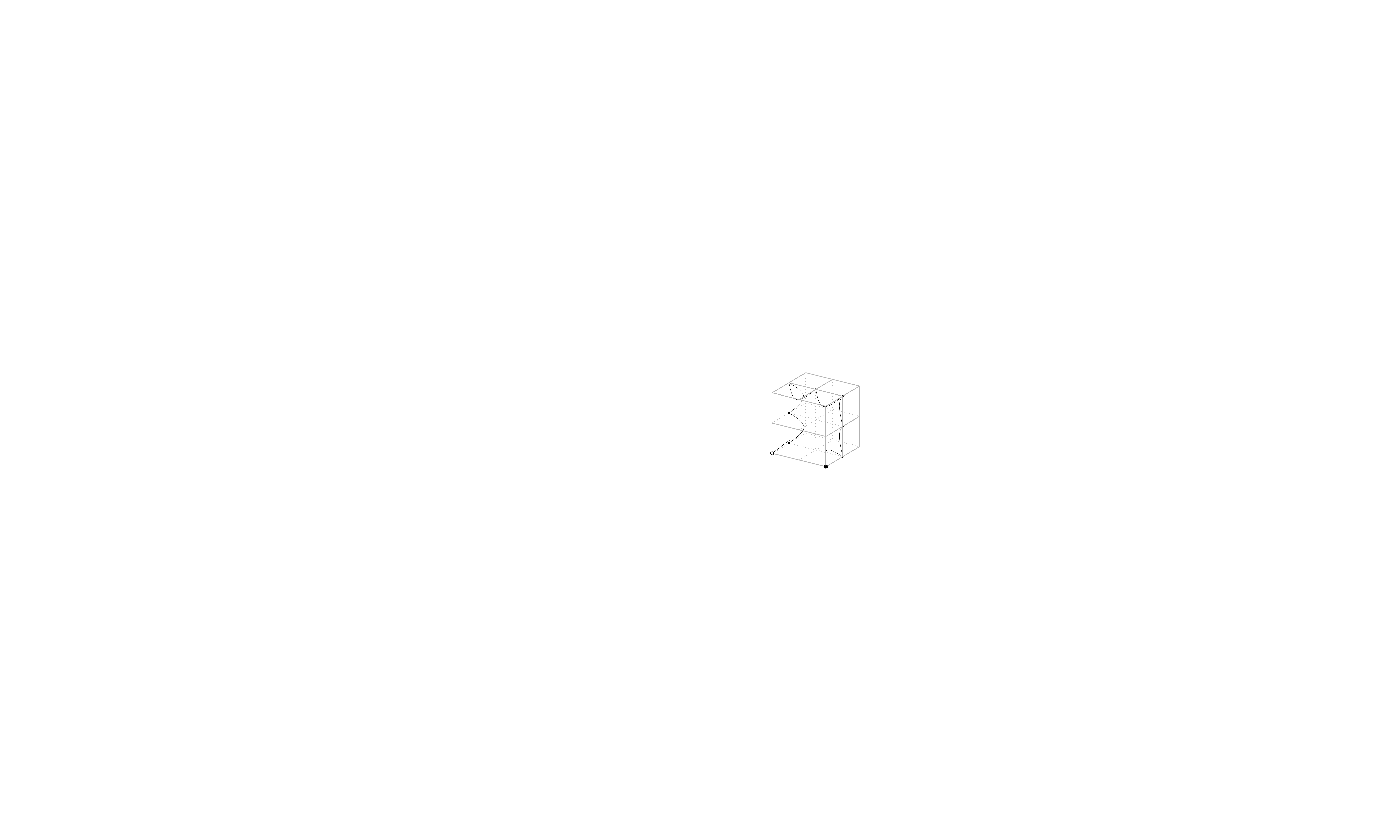}\hfill
}
\caption{The three-dimensional harmonious Hilbert curve. Left: definition. Right: location of the gates between the subregions.}
\label{fig:3DHarmonious}
\end{figure}

Table~\ref{tab:2reg5Dsets} shows the definitions of the five-dimensional versions of each of these curves. For illustration, Figure~\ref{fig:3DHarmonious} shows the graphical definition of the three-dimensional harmonious Hilbert curve.
The harmonious Hilbert curves are interdimensionally consistent and have neutral orientation (as we will see below), the Butz-Moore curves are not interdimensionally consistent (as observed in Section~\ref{sec:intro}) and do not have neutral orientation. The main result of Section~\ref{sec:hilbert} can be summarized as follows:

\begin{theorem}
The harmonious Hilbert curves constitute an interdimensionally consistent set of order-preserving 
2-regular vertex-continuous mono-curves with neutral orientation.
\end{theorem}

\begin{table}
\label{tab:2reg5Dsets}
\caption{Butz-Moore curves and harmonious Hilbert curves in five dimensions.}
\addvspace{.5\baselineskip}\centering\footnotesize
\begin{tabular}{@{}c@{\,}|@{\,}c@{ }c@{/}c@{ }c@{/}c@{ }c@{ }c@{ }||@{ }c@{\,}|@{\,}c@{ }c@{/}c@{ }c@{/}c@{ }c@{ }c@{}}
rank      & loc. & \multicolumn{2}{c}{permutat.} & \multicolumn{2}{c}{inv. perm.} & refl. & exit &
rank      & loc. & \multicolumn{2}{c}{permutat.} & \multicolumn{2}{c}{inv. perm.} & refl. & exit \\
          &      & B-M & HH & B-M & HH & & &
          &      & B-M & HH & B-M & HH & & \\
\hline
00000 & 00000 & 12340 & 43210 & 40123 & 43210 & 00000 & $\frac12(0,0,0,0,1)$ & 10000 & 11000 & 01234 & 04321 & 01234 & 04321 & 01001 & $\frac12(2,2,0,0,1)$ \\
00001 & 00001 & 23401 & 32104 & 34012 & 32104 & 00000 & $\frac12(0,0,0,1,1)$ & 10001 & 11001 & 23401 & 42103 & 34012 & 32140 & 11000 & $\frac12(2,2,0,1,1)$ \\
00010 & 00011 & 23401 & 43201 & 34012 & 34210 & 00000 & $\frac12(0,0,0,2,1)$ & 10010 & 11011 & 23401 & 14302 & 34012 & 30421 & 11000 & $\frac12(2,2,0,2,1)$ \\
00011 & 00010 & 34012 & 21043 & 23401 & 21043 & 00011 & $\frac12(0,0,1,2,1)$ & 10011 & 11010 & 34012 & 41032 & 23401 & 21430 & 11011 & $\frac12(2,2,1,2,1)$ \\
00100 & 00110 & 34012 & 43021 & 23401 & 24310 & 00011 & $\frac12(0,0,2,2,1)$ & 10100 & 11110 & 34012 & 14032 & 23401 & 20431 & 11011 & $\frac12(2,2,2,2,1)$ \\
00101 & 00111 & 23401 & 21403 & 34012 & 31042 & 00110 & $\frac12(0,0,2,1,1)$ & 10101 & 11111 & 23401 & 41302 & 34012 & 31420 & 11110 & $\frac12(2,2,2,1,1)$ \\
00110 & 00101 & 23401 & 43102 & 34012 & 32410 & 00110 & $\frac12(0,0,2,0,1)$ & 10110 & 11101 & 23401 & 24103 & 34012 & 32041 & 11110 & $\frac12(2,2,2,0,1)$ \\
00111 & 00100 & 40123 & 10432 & 12340 & 10432 & 00101 & $\frac12(0,1,2,0,1)$ & 10111 & 11100 & 40123 & 40321 & 12340 & 14320 & 11101 & $\frac12(2,1,2,0,1)$ \\
01000 & 01100 & 40123 & 40321 & 12340 & 14320 & 00101 & $\frac12(0,2,2,0,1)$ & 11000 & 10100 & 40123 & 10432 & 12340 & 10432 & 11101 & $\frac12(2,0,2,0,1)$ \\
01001 & 01101 & 23401 & 24103 & 34012 & 32041 & 01100 & $\frac12(0,2,2,1,1)$ & 11001 & 10101 & 23401 & 43102 & 34012 & 32410 & 10100 & $\frac12(2,0,2,1,1)$ \\
01010 & 01111 & 23401 & 41302 & 34012 & 31420 & 01100 & $\frac12(0,2,2,2,1)$ & 11010 & 10111 & 23401 & 21403 & 34012 & 31042 & 10100 & $\frac12(2,0,2,2,1)$ \\
01011 & 01110 & 34012 & 14032 & 23401 & 20431 & 01111 & $\frac12(0,2,1,2,1)$ & 11011 & 10110 & 34012 & 43021 & 23401 & 24310 & 10111 & $\frac12(2,0,1,2,1)$ \\
01100 & 01010 & 34012 & 41032 & 23401 & 21430 & 01111 & $\frac12(0,2,0,2,1)$ & 11100 & 10010 & 34012 & 21043 & 23401 & 21043 & 10111 & $\frac12(2,0,0,2,1)$ \\
01101 & 01011 & 23401 & 14302 & 34012 & 30421 & 01010 & $\frac12(0,2,0,1,1)$ & 11101 & 10011 & 23401 & 43201 & 34012 & 34210 & 10010 & $\frac12(2,0,0,1,1)$ \\
01110 & 01001 & 23401 & 42103 & 34012 & 32140 & 01010 & $\frac12(0,2,0,0,1)$ & 11110 & 10001 & 23401 & 32104 & 34012 & 32104 & 10010 & $\frac12(2,0,0,0,1)$ \\
01111 & 01000 & 01234 & 04321 & 01234 & 04321 & 01001 & $\frac12(1,2,0,0,1)$ & 11111 & 10000 & 12340 & 43210 & 40123 & 43210 & 10001 & $\frac12(2,0,0,0,0)$ \\
\end{tabular}
\end{table}

\begin{inset}
\caption{How the harmonious Hilbert curves were found.}
\footnotesize
The first non-trivial harmonious Hilbert curve is the three-dimensional curve. This curve was discovered when generating all 10,694,807 different three-dimensional mono-Hilbert curves, including curves that are not order-preserving~\cite{Haverkort3D}. In fact, it turned out that in three dimensions, the harmonious Hilbert curve, showing the two-dimensional Hilbert curve on five faces, is the only mono-Hilbert curve that is extradimensional to the two-dimensional Hilbert curve. There is no three-dimensional mono-Hilbert curve that shows the two-dimensional Hilbert curve on all six faces (even when allowing reflections on some faces).

Having discovered the three-dimensional curve, the question arose whether similar curves would exist in higher dimensions, and if so, how many of them. Unfortunately, the approach to finding interdimensionally consistent generalizations of mono-Wunderlich curves, as explained in Inset~\ref{box:howfoundternary}, could not be applied to mono-Hilbert curves: it would already fail in three dimensions. Instead, I programmed an exhaustive search algorithm that, given a $d$-dimensional curve $f_d$, would search all possible symmetric $(d+1)$-dimensional curves that traverse the subregions in binary reflected Gray code order, reporting those curves that are extradimensional to~$f_d$. Using effective pruning, unique solutions for four-, five- and six-dimensional curves were quickly found. One may regard this as a somewhat independent confirmation that the proofs given in Section~\ref{sec:proofhilbert} are likely to be correct.

It became clear that the challenge was not to run the exhaustive search, but to understand the structure in the solutions. Freek van Walderveen was the first to recognize structure in the permutation patterns, and managed to write an efficient algorithm that correctly predicted the permutations of the seven-dimensional curve, which was subsequently verified by my exhaustive search algorithm. Insights that led to a simpler algorithm began to come when I drew a graph of the subregions whose permutations differ by only one inversion. After drawing the graph for the five-dimensional curve on a geometric drawing of the five-dimensional hypercube, it became clear that inversions in the permutations of subregions $S(r)$ correspond to inversions of adjacent digits with different values in the binary representations of the ranks~$r$. With an analysis of the necessary base cases, this led to the simple definition given in Section~\ref{sec:standard}, which specifies the permutation for any given rank.
\end{inset}

\subsection{Implementation of comparison operators}
\label{sec:implementhilbert}

In this section we discuss how to implement comparison operators for each of the two sets of curves presented above. We will first explain how to do this for points in the unit hypercube, then explain how to extend the operators to points with arbitrary non-negative coordinates, and finally how to extend the operators to points with arbitrary real coordinates. The main purpose of this section is to clarify how such comparison operators can be implemented, and to enable us to discover more properties of the curves by analysing the algorithms. For real implementations, one would have to implement the \Extract function, and one may consider optimizations such as look-up tables for permutations. You may find that many steps of the algorithm can be implemented and quite easily and efficiently by bit operations such as bit shifting and exclusive or. Readers who are not interested in implementations may want to skip to Section~\ref{sec:proofhilbertcontinuous} right away.

\LinesNotNumbered
\RestyleAlgo{plain}

\paragraph{Modifications for all curves}
For points in the unit hypercube we can simply fill in the details of Algorithm~\ref{alg:RGCcurvesoperator}. To start with, we substitute ``1'' for ``$b - 1$''.

Second, we implement the reflections. First, consider the reflections of regions with even rank $r$. Since $r$ is even, we either have $r = 0$, or there is some $i$ such that $\pred{r}_{i-1} \pred{r}_{i}...\pred{r}_{d-1} = 0\reps{1}{d-i}$ and $r_{i-1} r_{i}...r_{d-1} = 1\reps{0}{d-i}$, while the first $i-1$ digits of $\pred{r}$ and $r$ are equal. By Lemma~\ref{lem:XORconstruction}, $i$ is now the largest $i$ such that $c_{i}(r) = 1$. Furthermore, by Lemma~\ref{lem:changingdigit}, we have $c_{i-1}(r) \neq c_{i-1}(\pred{r})$, and thus, by definition, $\inv{a}_0(r) = i-1$.

If $r$ is odd, we either have $r = \reps 1d$, or there is some $i$ such that $r_{i-1} r_{i}...r_{d-1} = 0\reps{1}{d-i}$ and $\succ{r}_{i-1} \succ{r}_{i}...\succ{r}_{d-1} = 1\reps{0}{d-i}$, and we find that $i$ is the largest $i$ such that $c_{i}(r) = 1$, $c_{i-1}(r) \neq c_{i-1}(\succ{r})$, and $\inv{a}_0(r) = i-1$. Finally, observe that $r$ is odd whenever the number of ones in $c(r)$ is odd, since $c(0) = 0$ and, as $r$ increases to $\reps 1d$, the value of $r$ and the number of ones in $c(r)$ both change by one in every step.

This leads to the following simple method to determine the reflections $m(r)$ based on $c(r)$, for $0 < r < \reps1d$.
Go through the digits $c_i(r)$ in order of increasing $i$: if $c_i(r) = 0$, set $m_i(r) = 0$; if $c_i(r) = 1$, set $m_i(r) = 1$ and set $\inv{a}_0(r) = i-1 \pmod d$. This procedure computes $m_i(r)$ correctly except when $c_i(r) \neq c_i(\pred{r})$ or $i < d-1$. Therefore we correct $m(r)$ by changing $m_{d-1}(r)$ and, when the number of ones in $c(r)$ is even, changing $m_i(r)$ for $i = \inv{a}_0(r)$. For $r = 0$, the exact same procedure also works correctly if $\inv{a}_0(r)$ is initialized to $d-1$, as $\inv{a}_0(r)$ will remain unchanged, $m_i(r)$ will be set to zero correctly for all $i$, and in the end, the two corrections (of $m_{d-1}(r)$ and of $m_i(r)$ for $i = \inv{a}_0(r) = d-1$) cancel each other's effect. For $r = \reps 1d$, we have $c(r) = 1\reps0{d-1}$, and the procedure will also set $\inv{a}_0(r)$ and $m_r$ correctly.

Therefore we implement the reflections correctly if we modify Algorithm~\ref{alg:RGCcurvesoperator} as follows. After Line~\ref{alg:rgccurves:initforward}, add:

\medskip
\begin{algorithm}[H]
$\iMinusOne \leftarrow d-1$; $\AltPermutation[0] \leftarrow \Permutation[\iMinusOne]$\;
\end{algorithm}

\medskip\noindent
Replace Line~\ref{alg:rgccurves:odddigit} by:

\medskip
\begin{algorithm}[H]
\If{$\pFirstDigit = 1$}{
$\Forward \leftarrow \Not \Forward$;
$\Reflected[\Permutation[i]] \leftarrow \Not \Reflected[\Permutation[i]]$\;
$\AltPermutation[0] \leftarrow \Permutation[\iMinusOne]$\;
}
$\iMinusOne \leftarrow i$\;
\end{algorithm}

\medskip\noindent
Remove Line~\ref{alg:rgccurves:reflect}. Instead, add the following before Line~\ref{alg:rgccurves:update} (outside the \textbf{for} loop):

\begin{algorithm}[H]
$\Reflected[\Permutation[d-1]] \leftarrow \Not \Reflected[\Permutation[d-1]]$\;
\If{$\Forward$}{
  \tcp{$\pFirstDigit$ was one an even number of times}
  $\Reflected[\AltPermutation[0]] \leftarrow \Not \Reflected[\AltPermutation[0]]$\;
}
\end{algorithm}

\medskip\noindent
Even better, we move Line~\ref{alg:rgccurves:initforward} to just before Line~\ref{alg:rgccurves:digitloop}, and `correct' $\Reflected[\AltPermutation[0]]$ unconditionally. The only effect of this is that we will enter the next iteration of the \textbf{repeat}...\textbf{until} loop with wrong values for $\Forward$ and $\Reflected[\Permutation[0]]$, but the reader may verify that this does not change the return values and $\Rank[i]$ as computed in the first iteration of the large \textbf{for} loop. In contrast, the value of $\pFirstDigit$ will change, and thus the values of $\Forward$ and $\Reflected[\Permutation[0]]$ will be correct again by the end of the first iteration of the \textbf{for} loop. Thus we get Algorithm~\ref{alg:binarycurvesoperator}.

\LinesNumbered
\RestyleAlgo{ruled}
\begin{algorithm}
\lFor{$i \leftarrow 0$ \KwTo $d-1$}{$\Reflected[i] \leftarrow \False$}\;
\lFor{$i \leftarrow 0$ \KwTo $d-1$}{$\Permutation[i] \leftarrow i$\tcp*[r]{maintains which entry in each row of $T$ is non-zero}}
initialize an array $\AltPermutation[0..d-1]$\tcp*[r]{for a new value of \Permutation under construction}\label{alg:binarycurves:altpermutation}
$\Forward \leftarrow \True$\;
\Repeat{all remaining digit strings $p[0],...,p[d-1]$ and $q[0],...,q[d-1]$ are empty}{
  $\iMinusOne \leftarrow d-1$; $\AltPermutation[0] \leftarrow \Permutation[\iMinusOne]$\;\label{alg:binarycurves:initainv0}
  \For{$i \leftarrow 0$ \KwTo $d-1$}{\label{alg:binarycurves:dimloop}
    $\pFirstDigit \leftarrow \Extract(p[\Permutation[i]])$; $\qFirstDigit \leftarrow \Extract(q[\Permutation[i]])$\;
    \If{$\Reflected[\Permutation[i]]$}{
      $\pFirstDigit \leftarrow (1 - \pFirstDigit)$; $\qFirstDigit \leftarrow (1 - \qFirstDigit)$\label{alg:binarycurves:reflect}
    }
    \lIf{$\pFirstDigit < \qFirstDigit$}{\Return \Forward}\ \lElseIf{$\pFirstDigit > \qFirstDigit$}{\Return \Not \Forward}\;
    \lIf{\Forward}{$\Rank[i] \leftarrow \pFirstDigit$}\ \lElse{$\Rank[i] \leftarrow (1 - \pFirstDigit)$}\;\label{alg:binarycurves:updaterank}
    \If{$\pFirstDigit = 1$}{
      $\Forward \leftarrow \Not \Forward$;
      $\Reflected[\Permutation[i]] \leftarrow \Not \Reflected[\Permutation[i]]$\;
      $\AltPermutation[0] \leftarrow \Permutation[\iMinusOne]$\;\label{alg:binarycurves:updateainv0}
    }
    $\iMinusOne \leftarrow i$\;\label{alg:binarycurves:updateiminus1}\label{alg:binarycurves:dimloopend}
  }
  \lFor{$i \leftarrow 1$ \KwTo $d-1$}{$\AltPermutation[i] \leftarrow \Permutation[\inv{a}_i(\Rank)]$}\;\label{alg:binarycurves:rotate}
  $\Reflected[\Permutation[d-1]] \leftarrow \Not \Reflected[\Permutation[d-1]]$;
  $\Reflected[\AltPermutation[0]] \leftarrow \Not \Reflected[\AltPermutation[0]]$\;
  swap $\Permutation$ and $\AltPermutation$\;
}
\Return\False\tcp*[r]{$p$ and $q$ are equal}
\caption{Framework for implementation of Butz-Moore and harmonious Hilbert curves\label{alg:binarycurvesoperator}}
\end{algorithm}

\paragraph{Butz-Moore curves}
The permutations of Butz-Moore curves are always rotations (in the permutation sense of the word), therefore they can be maintained by only keeping track of $\Permutation[0]$, without storing $\Permutation[1..d-1]$. Thus we can simplify Algorithm~\ref{alg:binarycurvesoperator} substantially 
and get Algorithm~\ref{alg:butzmooreoperator}: here we maintain $\Permutation[0]$ in a variable called $\Rotation$. We define the value of $-1 \bmod d$ as $d-1$.

\begin{algorithm}
$\Forward \leftarrow \True$; $\Rotation \leftarrow 0$;
\lFor{$i \leftarrow 0$ \KwTo $d-1$}{$\Reflected[i] \leftarrow \False$}\;
\Repeat{all remaining digit strings $p[0],...,p[d-1]$ and $q[0],...,q[d-1]$ are empty}{
  $i \leftarrow \Rotation$; $\iMinusOne \leftarrow (i - 1) \bmod d$;
  $\NewRotation \leftarrow \iMinusOne$\;
  \Repeat{$i = \Rotation$}{
    $\pFirstDigit \leftarrow \Extract(p[i])$; $\qFirstDigit \leftarrow \Extract(q[i])$\;
    \If{$\Reflected[i]$}{
      $\pFirstDigit \leftarrow (1 - \pFirstDigit)$; $\qFirstDigit \leftarrow (1 - \qFirstDigit)$\label{alg:binarycurves:reflect}
    }
    \lIf{$\pFirstDigit < \qFirstDigit$}{\Return \Forward};
    \lElseIf{$\pFirstDigit > \qFirstDigit$}{\Return \Not \Forward}\;
    \If{$\pFirstDigit = 1$}{
      $\Forward \leftarrow \Not \Forward$;
      $\Reflected[i] \leftarrow \Not \Reflected[i]$;
      $\NewRotation \leftarrow \iMinusOne$\;
    }
    $\iMinusOne \leftarrow i$; $i \leftarrow (i + 1) \bmod d$\;
  }
  $\Reflected[\iMinusOne] \leftarrow \Not \Reflected[\iMinusOne]$\;
  $\Reflected[\NewRotation] \leftarrow \Not \Reflected[\NewRotation]$\;
  $\Rotation \leftarrow \NewRotation$\;
}
\Return\False
\caption{Implementation of the Butz-Moore curves\label{alg:butzmooreoperator}}
\end{algorithm}

\paragraph{Harmonious Hilbert curves}
To implement the harmonious Hilbert curves, remove Lines \ref{alg:binarycurves:initainv0}, \ref{alg:binarycurves:updateainv0} and~\ref{alg:binarycurves:updateiminus1}, and replace Line~\ref{alg:binarycurves:rotate} by Algorithm~\ref{alg:harmoniousrotation}.
The first \textbf{for} loop of Algorithm~\ref{alg:harmoniousrotation} can be avoided by integrating it into the loop on Line~\ref{alg:binarycurves:dimloop}--\ref{alg:binarycurves:dimloopend} of Algorithm~\ref{alg:binarycurvesoperator}.

\begin{algorithm}
\SetKwData{Next}{indexForNext}
\SetKwData{NumberOfOnes}{numberOfOnes}
\SetKw{KwDown}{down}
$\NumberOfOnes \leftarrow 0$; \lFor{$i \leftarrow 0$ \KwTo $d-1$}{$\NumberOfOnes \leftarrow \NumberOfOnes + \Rank[i]$}\;
\If{$\Rank[d-1] = 0$}{
$\Next[1] \leftarrow 0$; $\Next[0] \leftarrow \NumberOfOnes$\;
}
\Else($\Rank[d-1] = 1$){
$\Next[0] \leftarrow 0$; $\Next[1] \leftarrow d - \NumberOfOnes$\;
}
\For{$i \leftarrow d-1$ \KwDown \KwTo $0$}{
  $\AltPermutation[\Next[\Rank[i]]] \leftarrow \Permutation[i]$; increment $\Next[\Rank[i]]$\;
}
\caption{Algorithm to compute the inverse permutation of region \textit{rank} of a harmonious Hilbert curve, and to apply it to the permutation \textit{axis} to obtain \textit{altAxis}.\label{alg:harmoniousrotation}}
\end{algorithm}

\paragraph{Points with non-negative coordinates}
To enable our implementation to compare points outside the unit hypercube, we consider the unit hypercube as the first subregion of a larger cube. Again, as in Section~\ref{sec:implementternary}, for the harmonious Hilbert curves it is easiest to consider the unit hypercube as the first subregion of the first subregion of a larger cube, because the first subregion of the first subregion is neither rotated nor reflected. Thus we get the following algorithm:
\LinesNotNumbered
\RestyleAlgo{plain}

\medskip
\begin{algorithm}[H]
\lWhile{any coordinate of $p$ or $q$ is at least 1}{divide all coordinates of $p$ and $q$ by 4}\;
run Algorithm~\ref{alg:binarycurvesoperator} for harmonious Hilbert curves on $p$ and $q$\;
\end{algorithm}

\paragraph{Points with negative coordinates}
One could try to extend the implementation to negative coordinates in a way similar to what is explained in Section~\ref{sec:implementternary}. Recall that the basic idea is to move the origin of the coordinate system for the original, hypercube-filling curves, to somewhere properly inside the unit hypercubes. By `zooming out', one can then obtain curves that fill the full $d$-dimensional space, rather than covering only points with non-negative coordinates. However, the resulting curves would only be interdimensionally consistent, if the original hypercube-filling curves retain interdimensional consistency with the redefined origin.
Since I have not found any location for the origin that meets this requirement, I omit the details for extensions to negative coordinates.

\subsection{Proof of vertex-continuity}
\label{sec:proofhilbertcontinuous}

In this section we will prove that the Butz-Moore and the harmonious Hilbert curves are indeed curves, that is, they are vertex-continuous. We will do so in two steps: first we prove that standard Hilbert curves are vertex-continuous, and then we proof that the Butz-Moore and the harmonious Hilbert curves are standard Hilbert curves.

\begin{lemma}
Standard Hilbert curves constitute mono-Hilbert curves.
\end{lemma}
\begin{proof}
All we have to prove is that the curves are vertex-continuous.

We first establish the location of the entrance and exit gates of the curves.
Recall that the curves start in the region $S(0)$ with location $\reps 0d$, and by definition $m_i(0) = 0$ for all $0 \leq i < d$. As a result, in the recursion within $S(0)$, the curve will always start in the region closest to the origin. Thus the entrance gate of the curve is at the origin. Let $r$ be $2^d - 1 = \reps 1d$. The curves end in the region $S(r)$ with location $c(r) = 1\reps 0{d-1}$, that is, the region touching the point $(1,0,...,0)$. Note that $c(\pred{r}) = 1\reps 0{d-2}1$, and therefore, $m(r) = 1\reps 0{d-2}1$. The permutation $a(r)$ is the reversal permutation. The coordinates of the last subregion within $S(r)$ are obtained by applying the permutation specified by $a(r)$ and the reflections specified by $m(r)$ to the coordinates $c(r)$: we find that the last subregion visited within $S(r)$ is again the region closest to the the point $(1,0,...,0)$. Thus, in recursion, the curve will always end in the region touching the point $(1,0,...,0)$, which is therefore the exit gate of the curve.

I claim that for $0 < r < 2^d$, the connecting gate between subregion $S(\pred{r})$ and $S(r)$ is at the point $x(r) = (x_0(r),...,x_{d-1}(r))$, where $x_i(r) = \frac12(c_i(\pred{r}) + c_i(r))$ for $0 \leq i < d-1$, and $x_{d-1}(r) = \frac12$. For ease of notation, we let $x(0) = (0,0,..,0)$ and $x(2^d) = (1,0,...,0)$ be the entrance and the exit gate, respectively, of the complete curve. (Figure~\ref{fig:3DHarmonious} shows an example in three dimensions). To prove that my claim is correct, we should establish that the curve within each region $S(r)$ starts in $x(r)$ and ends in $x(\succ{r})$. I leave it to the reader to verify this for the special cases of the first subregion ($r = 0$) and the last subregion ($r = 2^d-1$). We will now see how to verify that the curve within $S(r)$ starts in $x(r)$ and ends in $x(\succ{r})$ for $0 < r < 2^d-1$.

Observe that the location of $S(r)$ differs from the location of one of its neighbours in the ordering in coordinate $c_{d-1}$, and it differs from the location of the other neighbour in coordinate $c_j$ where $j = \inv{a}_0(r)$. Thus, $j = \inv{a}_0(r)$ is the unique $j < d-1$ such that $c_j(\pred{r}) \neq c_j(\succ{r})$. The entrance gate of subregion $S(r)$ is at $\tau(r)((0,0,...,0))$, which is the point $p = (p_0,...,p_{d-1})$ such that $p_i = \frac12(c_i(r) + m_i(r))$ for $0 \leq i < d$; the exit gate is at $\tau(r)((1,0,...,0))$, which is the point $q = (q_0,...,q_{d-1})$ such that $q_j = \frac12(c_j(r) + 1 - m_j(r))$ for $j = \inv{a}_0(r)$, and $q_i = \frac12(c_i(r) + m_i(r)) = p_i$ for all $i \in \{0, ..., d-1\} \setminus \{\inv{a}_0(r)\}$. We have to show (i) $p_i = x_i(r)$ and (ii) $q_i = x_i(\succ{r})$ for all $i \in \{0,...,d-1\}$.
For $i < d-1$, we can rewrite (i) as: $c_i(r) + m_i(r) = c_i(\pred{r}) + c_i(r)$, and therefore: $m_i(r) = c_i(\pred{r})$. This corresponds directly to the definition of $m_i(r)$. For $i < d-1$ and $i \neq j$, we can rewrite (ii) as: $c_i(r) + m_i(r) = c_i(r) + c_i(\succ{r}) = c_i(r) + c_i(\pred{r})$; again, this corresponds directly to the definition of $m_i(r)$. For $j = \inv{a}_0(r)$, we rewrite (ii) as: $c_j(r) + 1 - m_j(r) = c_j(r) + c_j(\succ{r}) = c_j(r) + 1 - c_j(\pred{r})$, and therefore: $m_j(r) = c_j(\pred{r})$.
Again, this corresponds directly to the definition of $m_i(r)$. Finally, for $i = d-1$, we rewrite both (i) and (ii) as: $c_i(r) + m_i(r) = 1$. This, too, corresponds directly to the definition of $m_i(r)$, QED.
\end{proof}

\begin{lemma}
The Butz-Moore curves and the harmonious Hilbert curves defined above are standard Hilbert curves.
\end{lemma}
\begin{proof}
What we have to prove is that the permutations $a$ in the definition of the Butz-Moore curves and the harmonious Hilbert curves satisfy the definition of $\inv{a}_0$ in standard Hilbert curves.

For the Butz-Moore curves this is trivial, as their permutations are defined in terms of a given $\inv{a}_0$.

For the harmonious Hilbert curves the proof is more involved. First, when $r = 0$ or $r = \reps 1d$, all bits of $r$ are equal. Thus, by the definition of harmonious Hilbert curves, $a_{d-1}(r) = (d - 1) - (d - 1) = 0$, which is indeed consistent with the definition of standard Hilbert curves. Second, consider the case when $0 < r < 2^d - 1$ and $r$ is odd, and let $i$ have the lowest value such that $r_i \neq \succ{r}_i$. Then $i < d-1$, $r_i = 0$, and $r_{i+1},...,r_{d-1} = 1$. By the definition of harmonious Hilbert curves, we will have $a_i(r) = 0$, and by the definition of standard Hilbert curves, we will have $c_i(r) \neq c_i(\succ{r})$, and thus, indeed, $a_i(r) = 0$. Third, consider the case when $0 < r < 2^d - 1$ and $r$ is even, and let $i$ have the lowest value such that $r_i \neq \pred{r}_i$. Then $i < d-1$, $r_i = 1$, and $r_{i+1},...,r_{d-1} = 0$. By the definition of harmonious Hilbert curves, we will have $a_i(r) = 0$, and by the definition of standard Hilbert curves, we will have $c_i(r) \neq c_i(\pred{r})$, and thus, indeed, $a_i(r) = 0$.
\end{proof}

\begin{corollary}
The Butz-Moore curves and the harmonious Hilbert curves defined above are mono-Hilbert curves.
\end{corollary}

\subsection{Proof of interdimensional consistency}
\label{sec:proofhilbert}
In this section we will prove that the harmonious Hilbert curves are interdimensionally consistent. Let $f_d$ denote the $d$-dimensional harmonious Hilbert curve. As observed in Section~\ref{sec:visibleorders}, to prove that the harmonious Hilbert curves are interdimensionally consistent, it suffices to show that each curve $f_d$ (for $d > 1$) shows $f_{d-1}$ on each front face. In fact, we will prove something stronger, namely that each curve $f_d$ shows $f_{d-1}$ on each front face, and each back face $F^1_i$ with $0 \leq i < d-1$ shows $f_{d-1}$ mirrored in dimension $i$ (note that this dimension corresponds to dimension $i+1$ in the $d$-dimensional space). This leaves $F^1_{d-1}$ as the only face for which we do not prove anything about the visible order.

Below, we write $f$ for $f_d$, and by $c$, $a$, $m$, $o$, $M$, and $\tau$ we denote the location, permutation, reflection, translation, transformation matrix and transformation functions that define $f$.
We write $f'$ for $f_{d-1}$, and by $c'$, $a'$, $m'$, $o'$, $M'$ and $\tau'$ we denote the location, permutation and reflection, translation, transformation matrix and transformation functions that define $f'$.
The proof that each curve $f$ shows $f'$ (sometimes mirrored, as specified above) on all faces except $F^1_{d-1}$ goes by induction on increasing level of refinement. More precisely, for a level of refinement $\depth$, we consider the recursive balanced subdivision of the $(d-1)$-dimensional unit hypercube $U'$ into $2^{\depth(d-1)}$ regions, and prove the following induction hypothesis:\begin{itemize}
\item[(i)] the visible order of $f$ on any face $F' = F^0_i$ with $0 \leq i \leq d-1$ visits these regions in the same order as $f'$;
\item[(ii)] the visible order of $f$ on any face $F' = F^1_i$ with $0 \leq i < d-1$ visits these regions in the same order as $f'$ mirrored in dimension $i$.
\end{itemize}
The proof follows the same general approach as the proof of consistency for the sets of mono-Wunderlich curves discussed in Section~\ref{sec:ternary}.

\paragraph{Base case $\depth = 1$}
As a base case, consider refinement level $\ell = 1$, and a face $F' = F^k_i$. Recall that $f$ is defined based on subdivision of a $d$-dimensional unit hypercube---which we will denote $U$---into $2^d$ regions $S(0),...,S(\reps 1d)$. This induces a subdivision of $F'$ into $2^{d-1}$ regions $F'(0),...,F'(\reps 1{d-1})$, each of which is a $(d-1)$-dimensional face of one of the regions $S(0),...,S(\reps 1d)$. Let $F'(0),...,F'(\reps 1{d-1})$ be indexed in the order in which the $d$-dimensional regions that contain them appear in $S(0),...,S(\reps 1d)$. Thus, there is a monotonously increasing function $\sigma: \{0,...,\reps 1{d-1}\} \rightarrow \{0,...,\reps 1d\}$ such that $F'(r')$ is a face of $S(\sigma(r'))$. Note that for $r' \in \{0,...,\reps 1{d-1}\}$, the locations of the regions $F'(r')$ within $U'$ are given by $\takeout{i}c(\sigma(r')) = \takeout{i}\BRGC^d(\sigma(r'))$. Furthermore, the co-domain of $\sigma$ consists of exactly those values $r$ such that $c_i(r) = \BRGC^d_i(r) = k$. Thus, the locations of the regions $F'(r')$ within $U'$, in order of increasing~$r'$, form the sequence $\reduction{i}{k}\BRGC^d$, which, by Lemma~\ref{lem:binarydigitremoval}, is exactly the sequence $\BRGC^{d-1}$ if $k = 0$, and the sequence $\mirror{i}\BRGC^{d-1}$ if $k = 1$ and $i < d-1$. Hence, the visible order of $f$ on any face $F' = F^k_i$ visits the regions of refinement level $\depth = 1$ in the same order as $f'$, or $f'$ mirrored in dimension $i$, respectively, thus establishing the base case for our induction.

\paragraph{Induction}
For the inductive step, we have to be more careful than with the mono-Wunderlich curves. Before, all we had to do is proving that the transformations that map the curve $f$ through $U$ to the section of $f$ through $S(\sigma(r'))$, also transform the order $f'$ shown on a face $F^*(r')$ into the section of $f'$ that should be visible on $F'(r')$. More concretely, we showed that we had to prove that we have $\takeout{i}m(\sigma(r')) = m'(r')$ (Equation~\ref{eq:matchreflections}) and $\takeout{i}a(\sigma(r')) = a'(r')$ (Equation~\ref{eq:matchpermutations}). However, with the harmonious Hilbert curves, not all faces show $f'$: some faces show a mirrored version, and one face ($F^1_{d-1}$) does not show $f'$ in any way. Denote the face of $U$ that maps to $F'(r')$ by $F^*(r')$. We will use $r$ as a shorthand for $\sigma(r')$. Recall that we now have:\[F^*(r') = F^{|k - m_i(r)|}_{a_i(r)}.\]
Therefore, we have to show:\begin{itemize}
\item that we never have $k \neq m_i(r)$ while $a_i(r) = d-1$ (never use the `broken' face);
\item that the reflections $m(r)$ undo any mirroring of the order on $F^*(r')$ and apply any necessary mirroring on $F'$;
\item that the permutations match: $\takeout{i}a(r) = a'(r')$ (this is no different from the situation with mono-Wunderlich curves in Section~\ref{sec:proofternary}).
\end{itemize}
Recall that, by the definition of $F'$ as $F^k_i$, we have $c_i(r) = k$.

We will first establish that we have $r' = \takeout{i}r$.
If $k = 0$, then, by definition of $r = \sigma(r')$, we have $c(r) = \putin{i}{0}c'(r')$, and hence, $\BRGC^d_i(r) = 0$ and $\takeout{i}\BRGC^d(r) = \BRGC^{d-1}(r')$. By Lemma~\ref{lem:XORconstruction} we now have $r_i = r_{i-1}$, and thus, by Lemma~\ref{lem:binarydigitremoval}, part~(i), case~(a), we have $r' = \takeout{i}r$.
Similarly, if $k = 1$ and $i < d-1$, then, by definition of $r = \sigma(r')$, we have $c(r) = \putin{i}{1}\mirror{i}c'(r')$, and hence, $\BRGC^d_i(r) = 1$ and $\takeout{i}\BRGC^d(r) = \mirror{i}\BRGC^{d-1}(r')$. By Lemma~\ref{lem:XORconstruction} we now have $r_i \neq r_{i-1}$, and thus, by Lemma~\ref{lem:binarydigitremoval}, part~(i), case~(c), we have $r' = \takeout{i}r$.

\paragraph{Never the broken face}
For the sake of contradiction, suppose $k \neq m_i(r)$ and $a_i(r) = d-1$. We distinguish three cases: $r = 0$; $r > 0$ and $i < d-1$; and $i = d-1$.

First, consider the case of $r = 0$. Then $k = c_i(r)$ and $m_i(r) = c_i(r)$ by definition, so we cannot have $k \neq m_i(r)$.

Second, consider the case of $r > 0$ and $i < d-1$. By the definition of the permutations $a(r)$ in harmonious Hilbert curves, it follows from $a_i(r) = d-1$ that $r_i = r_{d-1}$. However, by the definition of the reflections $m_i(r)$ in standard Hilbert curves, we also have $c_i(\pred{r}) = m_i(r)$. Since $c_i(r) = k \neq m_i(r) = c_i(\pred{r})$, this implies (by Lemma~\ref{lem:changingdigit}) that $i$ is the smallest index such that $r_i \neq \pred{r}_i$ and therefore $r_i = 1$ and $r_{d-1} = 0$. This contradicts the observation that $r_i = r_{d-1}$, so we cannot have $k \neq m_i(r)$ and $a_i(r) = d-1$.

Third, consider the case of $i = d-1$. In this case, our induction hypothesis only concerns the case $k = 0$, and thus we have $c_{d-1}(r) = 0$, $a_{d-1}(r) = d-1$ and $m_{d-1}(r) = 1$. By the definition of the permutations $a(r)$ in harmonious Hilbert curves, $a_{d-1}(r) = d-1$ if and only if $r_i \neq r_{d-1}$ for all $i < d-1$. But then, by Lemma~\ref{lem:XORconstruction}, $c_{d-1}(r) = |r_{d-1} - r_{d-2}| = 1$, which contradicts our observation that $c_{d-1}(r) = 0$. So in this case we cannot have $k \neq m_i(r)$ and $a_i(r) = d-1$ either.

\paragraph{Reflections}
We need to prove that the reflections $m(r)$ undo any mirroring of the order on $F^*(r')$ and apply any necessary mirroring on $F'$.
Let $j$ be $a_i(r)$ and let $g$ be $\inv{a}_{j+1}(r)$, that is, $g$ is defined such that $a_g(r) = a_i(r) + 1$. We define $\mathring{m}(r)$ as the reflections of $S(r)$ with application (or equivalently, cancellation) of the mirroring of the order of $f'$ as shown on $F^*(r')$. Note that the order of $f'$ as shown on $F^*(r')$ is mirrored only if $m_i(r) \neq k$, and if so, it is mirrored in coordinate $a_i(r)$, which is coordinate $a_i(r) + 1$ in $d$-dimensional space on $F^*(r')$, and this coordinate is mapped to coordinate $g$ on $F'$. Therefore, we have $\mathring{m}(r) = m(r)$ if $m_i(r) = k$, and $\mathring{m}(r) = \mirror{g}m(r)$ if $m_i(r) \neq k$. Furthermore, we define $\mathring{m}'(r')$ as the reflections of $F'(r)$ with the necessary mirroring of $F'$, that is, $\mathring{m}'(r') = m'(r')$ if $k = 0$, and $\mathring{m}'(r') = \mirror{i}m'(r')$ if $k = 1$. What we need to prove is now: $\takeout{i}\mathring{m}(r) = \mathring{m}'(r')$.

We define $\mathring{c}$ as follows: if $m_i(r) = k$, then $\mathring{c} = c$; if $m_i(r) \neq k$, then $\mathring{c} = \mirror{g}c$.
Thus, $\mathring{m}(r)$ is expressed in terms of $\mathring{c}(r)$ and $\mathring{c}(\pred{r})$ just as $m(r)$ is expressed in terms of $c(r)$ and $c(\pred{r})$ in the definition of standard Hilbert curves.
Similarly, we define $\mathring{c}'$ as follows: if $k = 0$, then $\mathring{c}' = c'$; if $k = 1$, then $\mathring{c}' = \mirror{i}c'$. Thus, $\mathring{m}'(r')$ is expressed in terms of $\mathring{c}'(r')$ and $\mathring{c}'(\pred{r}')$ just as $m(r)$ is expressed in terms of $c(r)$ and $c(\pred{r})$ in the definition of standard Hilbert curves.

Note that $\putin{i}{k}\mathring{c}'$ gives the locations of the subregions of $U'$ that touch $F'$, in the order in which these subregions are visited by $f$; thus we have, for any $q' \in \{0,...,\reps 1{d-1}\}$:\begin{equation}\label{eq:takeoutifromc}
\mathring{c}'(q') = \takeout{i}c(\sigma(q')).
\end{equation}

We are now ready to prove $\takeout{i}\mathring{m}(r) = \mathring{m}'(r')$. We distinguish four cases:\begin{enumerate}
\item[(1)] $m_i(r) = k$ and $r = 0$;
\item[(2)] $m_i(r) = k$ and $r > 0$ and $i < d-1$;
\item[(3)] $m_i(r) = k$ and $r > 0$ and $i = d-1$;
\item[(4)] $m_i(r) \neq k$ and $r' = 0$;
\item[(5)] $m_i(r) \neq k$ and $r' > 0$.
\end{enumerate}
In case (1), we have $\mathring{m}(r) = m(r)$ (because $m_i(r) = k$), we have $c_i(r) = 0$ (because $r = 0$), and thus, $k = 0$ and $\mathring{m}'(r') = m'(r')$. Since $r' = \takeout{i}r = 0$ we have $\mathring{m}(r) = m(r) = c(r) = 0$ and $\mathring{m}'(r') = m'(r') = c'(r') = 0$, and therefore $\takeout{i}\mathring{m}(r) = \mathring{m}'(r')$, QED.

In case (2), we have $\mathring{m}(r) = m(r)$ (because $m_i(r) = k$).
We have $c_i(\pred{r}) = m_i(r) = k$. Thus, the subregion preceding $S(r)$ in the $d$-dimensional order also contributes the region preceding $F'(r')$ on $F'$, of which the location is given by $\mathring{c}'(\pred{r}')$ (note that if $k = 1$, then $f'$ as visible on $F'$ is mirrored in dimension $i$). Therefore we have $\mathring{c}'(\pred{r}') = \takeout{i}c(\pred{r})$ and $\mathring{c}'(r') = \takeout{i}c(r)$. Therefore, for $z < i$ we have: $\mathring{m}'_z(r') = \mathring{c}'_z(\pred{r}') = c_z(\pred{r}) = m_z(r)$; for $i \leq z < d-2$ we have: $\mathring{m}'_z(r') = \mathring{c}'_z(\pred{r}') = c_{z+1}(\pred{r}) = m_{z+1}(r)$; and finally we have $\mathring{m}'_{d-2}(r') = 1 - \mathring{c}'_{d-2}(r') = 1 - c_{d-1}(r) = m_{d-1}(r)$. Thus, $\mathring{m}'(r') = \takeout{i}m(r) = \takeout{i}\mathring{m}(r)$, QED.

In case (3), by the definition of the permutations $m$ in a standard Hilbert curve and the definition of the case $m_i(r) = k$, we have $c_{d-1}(r) = 1 - m_{d-1}(r) = 1 - k$, but this contradicts $c_i(r) = k$, so this cannot occur.

In case (4), we have $m_i(r) \neq k$, and therefore $\mathring{m}(r) = \mirror{g}m(r)$, where $g$ is such that $a_g(r) = a_i(r) + 1$.
Since $r' = 0$, we have that $r$ equals the smallest $q$ such that $c_i(q) = k$. If $k = 0$, this implies $r = 0$, but then, by the definition of reflections in a standard Hilbert curve, we would have $m_i(r) = c_i(r) = 0 = k$, contradicting the definition of case (4). Therefore, $k = 1$, and the induction hypothesis only concerns cases with $i < d-1$. By Lemma~\ref{lem:XORconstruction}, we now have $r = \reps{0}{i}10\reps{0}{d-i-2}$ and $\pred{r} = \reps{0}{i}01\reps{1}{d-i-2}$. Applying the definitions of the permutations $a$ and the reflections in $m$ in standard and harmonious Hilbert curves,
we find $a_i(r) = 0$, $g = \inv{a}_1(r) = d-1$, $c(r) = \reps{0}{i}11\reps{0}{d-i-2}$, $c(\pred{r}) = \reps{0}{i}01\reps{0}{d-i-2}$;
if $i < d-2$ we have $m(r) = \reps{0}{i}01\reps{0}{d-i-3}1$; if $i = d-2$ we have $m(r) = \reps{0}{i}00$; therefore we have $\mathring{m}(r) = \mirror{d-1}m(r) = \reps{0}{i}01\reps{0}{d-i-2}$ in all cases, and $\takeout{i}\mathring{m}(r) = \reps{0}{i}1\reps{0}{d-i-2} = \mirror{i}\reps{0}{d-1} = \mathring{m}'(0)$, QED.

In case (5), we have $m_i(r) \neq k$, and therefore $\mathring{m}(r) = \mirror{g}m(r)$, where $g$ is such that $a_g(r) = a_i(r) + 1$.
Since $r' > 0$, we also have $r > 0$, and both $\pred{r}$ and $\pred{r}'$ exist. We have $c_i(\pred{r}) = m_i(r) \neq k$. Thus, the subregion preceding $S(r)$ in the $d$-dimensional order does not contain the region preceding $F'(r')$ on $F'$. Let $p'$ be $\pred{r}'$ and let $p$ be $\sigma(p')$. Note that we have $c_i(p) = k$, $c_i(\succ{p}),...,c_i(\pred{r}) = 1 - k$, $c_i(r) = k$, and:\begin{equation}\label{eq:nodiffexcepti}
\takeout{i}c(\pred{r}) = \takeout{i}c(r).
\end{equation}
By Lemma~\ref{lem:changingdigit}, as the rank increases from $p$ to $r$, digit $i$ must be the most significant digit that changes from 0 to 1 between $p$ and $\succ{p}$, and the next time it changes from 0 to 1 is between $\pred{r}$ and $r$. Therefore, between $\succ{p}$ and $\pred{r}$, there is a unique more significant digit $h$ that changes from 0 to 1, and it does so exactly once. Thus we get the following values for the digits of $p$, $\succ{p}$, $\pred{r}$ and $r$:

\begin{centering}
\addvspace{.5\baselineskip}
\leavevmode\begin{tabular}{l||ccc|c|c}
$x$ & $x_0,...,x_{h-1}$ & $x_h$ & $x_{h+1},...,x_{i-1}$ & $x_i$ & $x_{i+1},...,x_{d-1}$ \\
\hline
$p$        & anything    & 0 & 1 & 0 & 1 \\
$\succ{p}$ & same as $p$ & 0 & 1 & 1 & 0 \\
$\pred{r}$ & same as $p$ & 1 & 0 & 0 & 1 \\
$r$        & same as $p$ & 1 & 0 & 1 & 0 \\
\end{tabular}

\addvspace{.5\baselineskip}
\end{centering}

\noindent Note that we may have $h = i-1$ or $i = d-1$, so digits $x_{h+1},...,x_{i-1}$ and digits $x_{i+1},...,x_{d-1}$ might not exist.
When rank digit $h$ changes as the rank increases from $p$ to $r$, the corresponding digit in the location changes as well. Thus, $h$ must have the unique value such that $\mathring{c}'_h(p') \neq \mathring{c}'_h(r')$, and we get:\begin{equation}\label{eq:mirrorcprime}
\mathring{c}'(p') = \mirror{h}\mathring{c}'(r')
\end{equation}
Since $c_i(p) = c_i(r)$, we have that $h$ also has the unique value such that $c_h(p) \neq c_h(r)$, so:\begin{equation}\label{eq:mirrorc}
c(p) = \mirror{h}c(r)
\end{equation}
Applying the definition of the permutations $a$ in harmonious Hilbert curves to the numbers in the table, we find that we have $a_h(r) = a_i(r) + 1$. Thus, since $m_i(r) = 1 - k$, we have $\mathring{c} = \mirror{h}c$.
Note that $p' = \takeout{i}p$ (just as $r' = \takeout{i}r$). Using Equations \ref{eq:takeoutifromc}, \ref{eq:nodiffexcepti}, \ref{eq:mirrorcprime} and~\ref{eq:mirrorc} and $h < i$, we get:\begin{equation}
\mathring{c}'(\pred{r}') = \mathring{c}'(p') = \mirror{h}\mathring{c}'(r') = \takeout{i}\mirror{h}c(r) = \takeout{i}\mirror{h}c(\pred{r}) = \takeout{i}\mathring{c}(\pred{r}).
\end{equation}
Thus, for $z < i$ we have: $\mathring{m}'_z(r') = \mathring{c}'_z(\pred{r}') = \mathring{c}_z(\pred{r}) = \mathring{m}_z(r)$; for $i \leq z < d-2$ we have $\mathring{m}'_z(r') = \mathring{c}'_z(\pred{r}') = \mathring{c}_{z+1}(\pred{r}) = \mathring{m}_{z+1}(r)$; and finally we have $\mathring{m}'_{d-2}(r') = 1 - \mathring{c}'_{d-2}(r') = 1 - c_{d-1}(r) = 1 - \mathring{c}_{d-1}(r) = \mathring{m}_{d-1}(r)$. Thus, $\mathring{m}'(r') = \takeout{i}\mathring{m}(r)$, QED.

\paragraph{Permutations}
Recall that we have $r' = \takeout{i}r$. Therefore, if $i < d-1$, we have $r'_{d-2} = r_{d-1}$. If $i = d-1$, the induction hypothesis only concerns the case of $k = 0$, that is, $c_{d-1}(r) = 0$, and therefore, by Lemma~\ref{lem:XORconstruction}, we have $r_{d-1} = r_{d-2} = r'_{d-2}$. Thus, we always have $r'_{d-2} = r_{d-1}$.

Recall that the respective inverses of $a'(r')$ and $a(r)$ (where $r = \sigma(r')$) can be constructed from the identity permutation by moving all indices $h$ such that $r'_h = r'_{d-2}$, or $r_h = r_{d-1} = r'_{d-2}$, respectively, to the front, reversing the order of those indices, and reversing the order of the remaining indices. Inserting a digit in $r'$ at position $i$ to obtain $r = \sigma(r')$, does not affect which of the other digits are moved to the front, since this only depends on the value of $r'_{d-2}$. Thus, inserting a digit in $r'$ at position $i$ to obtain $r$, moving the index $i$ in $\inv{a}(r)$, and subsequently removing index $i$ from $\inv{a}(r)$ (which is what $\takeout{i}a(r)$ does), does not affect the order of the remaining indices relative to each other. Therefore we have $a'(r') = \takeout{i}a(r)$, QED.

\subsection{Orientation of the Butz-Moore and harmonious Hilbert curves}

By Lemma~\ref{lem:neutral}, a $d$-dimensional mono-Hilbert curve has neutral orientation if and only if there is a number~$\depth$, such that each of the $d!$~possible permutations of $d$~numbers can be constructed as the composition of $\depth$~permutations from~$a$. Trivially, each 1-dimensional curve has neutral orientation. Below we analyse whether the $d$-dimensional curves for $d \geq 2$ have neutral orientation.

\paragraph{Butz-Moore curves} In the Butz-Moore curves, all permutations are rotations. Only if $d = 2$, this suffices to create all possible permutations by composition. Therefore, only the two-dimensional Butz-Moore curve, that is, only the original Hilbert curve, has neutral orientation; the higher-dimensional Butz-Moore curves have not.

\paragraph{Harmonious Hilbert curves} In the harmonious Hilbert curves, $a(\reps{0}{d-2}10)$ is the permutation that swaps the first two elements and then reverses the whole sequence, and $a(\reps0{d-i}\reps{1}i)$ is the permutation that rotates the sequence $i$ steps to the right and then reverses the whole sequence. This allows us to construct any swap of two adjacent elements at positions $j-1$ and $j$ by composing the four permutations $a(\reps0{j+1}\reps{1}{d-j-1})$, $a(\reps{0}{d-2}10)$, $a(\reps0{d-j-1}\reps{1}{j+1})$ and $a(\reps0d)$. Since any permutation can be created from any other by at most $d(d-1)/2$ swaps of adjacent elements, we get that any permutation can be created from any other by composing at most $2d(d-1)$ permutations from $a$. However, to satisfy the condition in Lemma~\ref{lem:neutral}, there should be a depth $\depth$ such that each permutation can be created from the identity permutation by composing \emph{exactly} $\depth$ permutations. In the harmonious Hilbert curves, this can indeed be done for $\depth = 2d(d-1)$, since any permutation that is produced from $k < d(d-1)/2$ swaps can be produced by a sequence of $4k$ permutations producing the swaps, followed by $2d(d-1) - 4k$ reversal permutations $a(\reps0d)$. It follows that the harmonious Hilbert curves have neutral orientation.

Note that the analysis given above is by no means tight: in reality much less than $2d(d-1)/2$ permutations from $a$ may suffice to produce all possible permutations. For example, in four dimensions three permutations suffice, rather than $2d(d-1) = 24$.

\section{\Halffix- and diagonal-extradimensional space-filling curves}
\label{sec:composition}

\subsection{Definition of \halffix- and diagonal-extradimensionality}
Recall that a $d$-dimensional space-filling curve $f$ is \emph{extradimensional to} a $d'$-dimensional space-filling curve $f'$ with co-domain $U'$, if the following holds for any valid ordering $\inv{f'}$ of $f'$ and \emph{any} choice of $\mu$: there is a valid ordering $\inv{f}$ of $f$ such that for any pair of points $a,b\in U'$ we have $\inv{f'}(a) < \inv{f'}(b)$ if and only if $\inv{f}(\lift\mu(a)) < \inv{f}(\lift\mu(b))$~(Definition~\ref{def:extradimensional}). Here $\lift\mu$ is a function that defines the correspondence between points in the co-domains of $f$ and $f'$: the sequence of coordinates of $\lift\mu(p)$ is simply the sequence of coordinates of $p$, with zeros inserted in certain places according to $\mu$. In this section we will consider two variants of extradimensionality: \halffix-extradimensionality and diagonal-extradimensionality.

In the case of \halffix-extradimensionality, $d = 2d'$ and $\mu$ is restricted to the identity function $\identity$, defined by $\identity(i) = i$. Thus, the coordinates of $\lift\mu(p)$ are simply the coordinates of $p$ with $d'$ zeros added at the end:

\begin{definition}\label{def:halfextradimensional}
A $2d'$-dimensional space-filling curve $f$ is \emph{\halffix-extradimensional to} a $d'$-dimensional space-filling curve $f'$ with co-domain $U'$, if the following holds for any valid ordering $\inv{f'}$ of $f'$: there is a valid ordering $\inv{f}$ of $f$ such that for any pair of points $a,b\in U'$ we have $\inv{f'}(a) < \inv{f'}(b)$ if and only if $\inv{f}(\lift\identity(a)) < \inv{f}(\lift\identity(b))$.
\end{definition}

Given a $d'$-dimensional point $p = (p_0,...,p_{d'-1})$, we define $\diaglift(p)$ as the $2d'$-dimensional point $(q_0,...,q_{2d'-1})$ such that $q_i = p_{(i \bmod d')}$. In other words, $\diaglift(p)$ is obtained from $p$ by concatenating two copies of the coordinate sequence of $p$. For a set of points $P$, let $\diaglift(P)$ be the set $\bigcup_{p \in P} \diaglift(p)$.

\begin{definition}\label{def:diagextradimensional}
A $2d'$-dimensional space-filling curve $f$ is \emph{diagonal-extradimensional to} a $d'$-dimensional space-filling curve $f'$ with co-domain $U'$, if the following holds for any valid ordering $\inv{f'}$ of $f'$: there is a valid ordering $\inv{f}$ of $f$ such that for any pair of points $a,b\in U'$ we have $\inv{f'}(a) < \inv{f'}(b)$ if and only if $\inv{f}(\diaglift(a)) < \inv{f}(\diaglift(b))$.
\end{definition}

In previous work, we found that certain applications call for space-filling curves that are \halffix- and/or diagonal-extradimensional to a given space-filling curve with good locality-preserving properties~\cite{Haverkort4D}. For details, see Inset~\ref{box:rtrees}.

\begin{inset}
\caption{Building R-trees with extradimensional space-filling curves.}
\footnotesize
An R-tree is a balanced tree structure in which the leaves store bounding boxes of $d$-dimensional objects. Each object is stored exactly once, and typically each leaf contains the bounding boxes of a fair number of objects: enough to fill a page on disk. As the bounding box of an object, the smallest axis-aligned $d$-dimensional box that contains the object is used. Each leaf or other node of the tree also has a bounding box of its own, which is the smallest axis-aligned $d$-dimensional box that encloses all objects stored in the subtree rooted at that node. In addition, each non-leaf node stores the bounding boxes of its children. This structure can be used to answer several types of queries efficiently~\cite{proefschrift}, provided the object bounding boxes are distributed over the leaves in such a way that leaves have small bounding boxes.

One heuristic way to achieve this, is by the use of space-filling curves~\cite{Kamel}. To determine the ordering of two-dimensional axis-parallel rectangles (object bounding boxes) in the tree, we represent the rectangles as four-dimensional points, which are subsequently sorted by the order of their positions along a four-dimensional space-filling curve. Two ways to represent rectangles $[x_{\mathrm{min}},x_{\mathrm{max}}] \times [y_{\mathrm{min}},y_{\mathrm{max}}]$ as four-dimensional points were considered. The \emph{xy-representation} would use the point:\[ \left(x_{\mathrm{min}},y_{\mathrm{min}},x_{\mathrm{max}},y_{\mathrm{max}}\right),\]
whereas the \emph{cd-representation} (centre/dimensions) would use the point:\[
\left(\frac{x_{\mathrm{min}} + x_{\mathrm{max}}}2, \frac{y_{\mathrm{min}} + y_{\mathrm{max}}}2, x_{\mathrm{max}} - x_{\mathrm{min}}, y_{\mathrm{max}} - y_{\mathrm{min}}\right).\]

Since in practice, many rectangles are small, xy-representation leads to many points where the first and the third coordinate are almost identical, and the second and the fourth coordinate are almost identical. The other solution, cd-representation, leads to many points whose last two coordinates are almost zero. In previous work we argued and showed evidence that it is important that such four-dimensional points are ordered as much as possible as if the represented rectangles were simply ordered according to the positions of their centre points along a good two-dimensional space-filling curve~\cite{Haverkort4D}. This can be achieved by using four-dimensional curves that are diagonal- or \halffix-extradimensional to Hilbert's two-dimensional curve.

The results in this paper could be used to try this approach to building R-trees in higher dimensions (for example, using six-dimensional curves for three-dimensional boxes), or to experiment with 3-regular curves instead of 2-regular curves for this application.
\label{box:rtrees}
\end{inset}

Any $2d'$-dimensional Peano curve, coil curve, half-coil curve, Meurthe curve, or harmonious Hilbert curve (see Sections \ref{sec:ternary} and~\ref{sec:hilbert}) is \halffix-extradimensional, but generally not diagonal-extradimensional, to the $d'$-dimensional Peano curve, coil curve, half-coil curve, Meurthe curve, or harmonious Hilbert curve, respectively. In this section we will show how, given any $d'$-dimensional $\base$-regular order-preserving vertex-continuous mono-curve $f'$ with $\base \in \{2,3\}$ and $f'(0) = (0,0,...,0)$, one can construct a $2d'$-dimensional $\base$-regular vertex-continuous mono-curve $f$ that is both \halffix-extradimensional and diagonal-extra-dimensional to $f'$. This may provide a useful alternative to the curves from the previous sections if diagonal-extradimensionality is desired, or if one wants to base the $2d'$-dimensional curve on a $d'$-dimensional curve that is different from those presented in the previous sections.

\subsection{Monotone space-filling curves} 
As a building block in our construction of \halffix- and diagonal-extradimensional curves, we will use two-dimensional curves that have certain special properties. We will denote these curves by $h_2$ and $h_3$, where $h_2$ is the two-dimensional Hilbert curve (Figure~\ref{fig:monotone}, left), and $h_3$ is the {\cyrrm {Ya}}-curve (Figure~\ref{fig:monotone}, right). Lemma \ref{lem:HilbertIsEdgeMonotone} states that $h_2$ and $h_3$ have the properties which we will need to prove the correctness of our construction in Section~\ref{sec:halffixconstruction}.

\begin{figure}
\centering
\hbox to \hsize{\hfill
\includegraphics[width=\hsize]{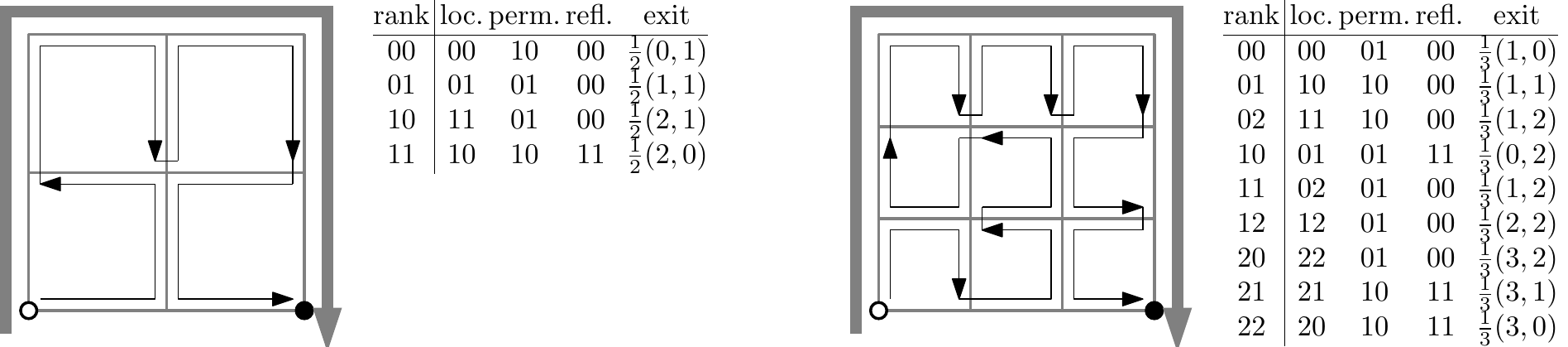}\hfill
}
\caption{Hilbert's curve (left) and the {\cyrrm {Ya}}-curve (right), defined graphically and by means of a table.}
\label{fig:monotone}
\end{figure}

\begin{definition}\label{def:monotone}
An ordering $\inv{f}$ of a space-filling curve $f$ is \emph{monotone} from $a$ to $b$, if $\inv{f}(p)$ increases monotonically as a point $p$ moves in a straight line from $a$ to $b$.
\end{definition}

\begin{lemma}
\label{lem:HilbertIsEdgeMonotone}
\label{lem:HilbertIsDiagMonotone}
The curves $h_2$ and $h_3$ have valid orderings $\inv{h_2}$ and $\inv{h_3}$ that are monotone from $(0,0)$ to $(1,0)$ and from $(0,0)$ to $(1,1)$.
\end{lemma}
\begin{proof}
We will prove the lemma for the Hilbert curve $h_2$; the proof for the {\cyrrm {Ya}}-curve $h_3$ is very similar and is left as an exercise for the interested reader. In fact, we will prove a stronger claim for $h_2$, namely that $h_2$ has a valid ordering $\inv{h_2}$ that is monotone from $(0,0)$ to $(0,1)$, from $(0,1)$ to $(1,1)$, from $(1,1)$ to $(1,0)$, from $(0,0)$ to $(1,0)$, and from $(0,0)$ to $(1,1)$.

We will prove the lemma by induction on increasing level of refinement of the subdivision of the unit square into subsquares.
More precisely, for a level of refinement $\ell$, we consider the recursive balanced subdivision of the unit square into $4^\ell$ square regions, and prove the following induction hypothesis: there is a valid ordering $\inv{h_2}$ that:\begin{enumerate}
\item[(1)] puts the regions touching the left edge of the unit square in order from bottom $(0,0)$ to top $(0,1)$;
\item[(2)] puts the regions touching the top edge of the unit square in order from left $(0,1)$ to right $(1,1)$;
\item[(3)] puts the regions touching the right edge of the unit square in order from top $(1,1)$ to bottom $(1,0)$;
\item[(4)] puts the regions touching the bottom edge of the unit square in order from left $(0,0)$ to right $(1,0)$;
\item[(5)] puts the regions on the ascending diagonal of the unit square in order from bottom $(0,0)$ to top $(1,1)$.
\end{enumerate}
As a base case, we take $\ell = 1$. For this base case, the induction hypothesis can easily be verified by inspection of Figure~\ref{fig:monotone}.

We will now prove that if the induction hypothesis holds for refinement level $\ell-1$, then it also holds for refinement level $\ell$. Note that, as far as the left edge, the bottom edge and the diagonal are concerned, the induction hypothesis is invariant under the transformation applied in the first (lower left) quadrant of $h_2$, and, as far as the right and the bottom edge are concerned, the induction hypothesis is invariant under the transformation applied in the last (lower right) quadrant of $h_2$. Thus \inv{h_2} places:\begin{itemize}
\item the regions touching the lower half of the left edge of the unit square before the regions touching the upper half of the left edge of the unit square (by the induction hypothesis for level~1);
\item the regions touching the lower half of the left edge of the unit square in order from bottom to top (by the induction hypothesis for level $\ell-1$);
\item the regions touching the upper half of the left edge of the unit square in order from bottom to top (by the induction hypothesis for level $\ell-1$).
\end{itemize}
As a result, $\inv{h_2}$ puts all level-$\ell$ regions touching the left edge of the unit square in order from bottom to top. In a similar way, one can verify the order induced by $\inv{h_2}$ on the regions touching the top, right, and bottom edges and on the ascending diagonal of the unit square.
\end{proof}

\subsection{Constructing \halffix- and diagonal-extradimensional curves}
\label{sec:halffixconstruction}

To be able to present and analyse our construction, we need some notation. When $f: [0,1] \rightarrow [0,1]^d$ is a $d$-dimensional space-filling curve, let $f^{[i]}(t)$ be coordinate $i$ of $f(t)$, where coordinates are numbered from $0$ to $d-1$. In other words, $f(t) = (f^{[0]}(t),f^{[1]}(t),...,f^{[d-1]}(t))$. For two $d'$-dimensional points $p$ and $q$, we let $p|q$ be the $2d'$-dimensional point obtained by concatenating the coordinate sequences of $p$ and $q$, that is, $p|q = (p_0,...,p_{d'-1},q_0,...,q_{d'-1})$. For a $2d'$-dimensional point $p$, we let $\lefthalf{p}$ denote the $d'$-dimensional point obtained by taking the first $d'$ coordinates of $p$, that is, $\lefthalf{p} = (p_0,...,p_{d'-1})$. Similary, $\righthalf{p}$ is the point obtained by taking the last $d'$ coordinates of $p$, that is, $\righthalf{p} = (p_{d'},...,p_{2d'-1})$. We use 0 not only to denote the number zero, but also to denote the origin of the coordinate system, in any number of dimensions.

Now, given a $d'$-dimensional $\base$-regular space-filling curve $f'$ with $\base \in \{2,3\}$, we can define a $2d'$-dimensional space-filling curve $f$ by:\[\begin{array}{ll}
f(t) = f'(h_b^{[0]}(t)) | f'(h_b^{[1]}(t))
\end{array}\]
To put it differently: for $k \in \{0,1\}$ and $i \in \{0,...,d'-1\}$, we define $f^{[kd'+i]}(t)$ as $f'^{[i]}(h_b^{[k]}(t))$.
With a slight abuse of terminology and notation, we can think of $f$ as the composition of $f'$ and $h_\base$, denoted $f' \circ h_\base$. Below we will prove the following:

\begin{theorem}\label{th:composition}
If $f'$ is an order-preserving vertex-continuous mono-curve and $f'(0) = 0$, then $f = f' \circ h_\base$ is a $\base$-regular vertex-continuous space-filling curve, and $f$ is \halffix-extradimensional and diagonal-extradimensional to $f'$. If, in addition, $f'$ is symmetric, then $f$ is an order-preserving mono-curve.
\end{theorem}

In Section~\ref{sec:implementhalffix} we will see how to implement a comparison operator for $f$. We will now prove Theorem~\ref{th:composition} in four parts (Lemmas~\ref{lem:compositecurve}--\ref{lem:2regDiagonallyExtradimensional}).

\begin{lemma}\label{lem:compositecurve}
If $f'$ is a order-preserving $\base$-regular vertex-continuous mono-curve,
then $f = f' \circ h_\base$ is a $\base$-regular vertex-continuous space-filling curve.
\end{lemma}
\begin{proof}
For each point $p \in [0,1]^{2d'}$ there are values $a, b \in [0,1]$ such that $f'(a) = \lefthalf{p}$ and $f'(b) = \righthalf{p}$ (because $f'$ is a surjective function from $[0,1]$ to $[0,1]^{d'}$), and there is a value $t$ such that $h_2(t) = (a,b)$ (because $h_2$ is a surjective function from $[0,1]$ to $[0,1]^2$). Therefore $f$ is a surjective function from $[0,1]$ to $[0,1]^{2d'}$. Furthermore, since $f'$ and $h_b$ are continuous, so is $f$.

It remains to prove that $f$ is $\base$-regular. Consider level $d'\depth$ of refinement of $h_b$. At this level of refinement, $h_b$ maps each interval $[r/\base^{2d'\depth},(r+1)/\base^{2d'\depth}]$, for $r \in \{0,...,\base^{2d'\depth}-1\}$, to a square of width $1/\base^{d'\depth}$. The horizontal range of the square, that is, the possible values of $h_b^{[0]}$ that are attained in this square, correspond exactly to the pre-image of one of the $\base^{d'\depth}$ level-$\ell$ regions of $f'$, thus defining the region occupied in $2d'$-dimensional space with respect to the first $d'$ dimensions. Similarly, the vertical range of the square corresponds exactly to the pre-image of a region of $f'$, thus defining the region occupied in $2d'$-dimensional space with respect to the last $d'$ dimensions. Thus, each interval $[r/\base^{2d'\depth},(r+1)/\base^{2d'\depth}]$ is mapped exactly to the Cartesian product of two $d'$-dimensional hypercube regions of width $1/\base^\depth$, which is a $2d'$-dimensional hypercube region of width $1/\base^\depth$. This implies the hierarchical structure of a $b$-regular curve.
\end{proof}

\begin{lemma}\label{lem:compositesymmetric}
If $f'$ is a symmetric $\base$-regular mono-curve,
then $f = f' \circ h_\base$ is an order-preserving $\base$-regular mono-curve.
\end{lemma}
\begin{proof}
Extending our notation from Section~\ref{sec:numerical}, we denote the subregions, locations, permutations, and reflections of a curve $f$ by $S^f$, $c^f$, $a^f$ and $m^f$, respectively. Let $h_{b\cdot \ell}$ be the level $\ell$ of refinement of $h_b$, at which level regions are numbered from 0 to $b^{2\ell}-1$, and the region coordinates $c^{h_{b\cdot \ell}}_i$ are between 0 and $\base^\ell - 1$ (inclusive). The transformation of $f$ in region $r \in \{0,...,\base^{2d'}-1\}$ is the result of the successive effects of (i) the permutation $a^{h_{b\cdot d'}}(r)$, (ii) the reflections $m^{h_{b\cdot d'}}(r)$, (iii) the permutations $a^{f'}(c^{h_{b\cdot d'}}_k(r))$ for $k \in \{0,1\}$, and (iv) the reflections $m^{f'}(c^{h_{b\cdot d'}}_k(r))$. Note that there are no reversals, since $f'$ is symmetric, and therefore, order-preserving.

(i) The permutation $a^{h_{b\cdot d'}}(r)$ is either 01 or 10: the latter results in a rotation of $f$ that consists in swapping coordinates $0,...,d'-1$ with coordinates $d',...,2d'-1$.

(ii) Because $f'$ is symmetric, there is a transformation $\rho$ (consisting of rotation and/or reflections) such that $f'(1-h_b^{[k]}(t)) = \rho(f'(h_b^{[k]}(t)))$. A reflection $m^{h_{b\cdot d'}}_0(r)$ results in replacing the coordinates $\lefthalf{f}(t) = f'(h_b^{[0]}(t))$ by $f'(1-h_b^{[0]}(t)) = \rho(f'(h_b^{[0]}(t)) = \rho(\lefthalf{f}(t))$ for $i \in \{0,...,d-1\}$, which results in a transformation of $f$ that consists in applying $\rho$ to the first $d'$ coordinates while leaving the last $d'$ coordinates unaffected; thus, this transformation of $f$ consists of rotations and reflections in $2d'$-dimensional space. Similarly, a reflection $m^{h_{b\cdot d'}}_1(r)$ results in a transformation of $f$ that consists in applying $\rho$ to the last $d'$ coordinates while leaving the first $d'$ coordinates unaffected.

(iii,iv) Similar to $\rho$, the permutations $a^{f'}(c^{h_{b\cdot d'}}_k(r))$ and the reflections $m^{f'}(c^{h_{b\cdot d'}}_k(r))$ result in transformations of $f$ that consist of rotations and reflections.

Thus, the transformation of $f$ in region $r \in \{0,...,\base^{2d'}-1\}$ consists of rotations and reflections, and $f$ is an order-preserving mono-curve.
\end{proof}

Note that if $f'$ is not symmetric, then $f$ is not necessarily a mono-curve. A reflection of $h_b$ in coordinate $k \in \{0,1\}$ results in replacing $f^{[kd'+i]}(h_b^{[k]}(t))$ by $f^{[kd'+i]}(1-h_b^{[k]}(t))$ for $i \in \{0,...,d'-1\}$. In general, this transformation cannot be expressed as a simple permutation, reflection, or reversal of the $d$-dimensional curve $f$. For example, if $\base = 3$ and $f' = h_3$ (the {\cyrrm {Ya}}-curve of Figure~\ref{fig:monotone}), then $f$ is not a mono-curve. The same might be the case if $f'$ is a Meurthe curve (see Section~\ref{sec:ternary}).

However, if $f'$ is symmetric (for example, a Butz-Moore Hilbert curve or a harmonious Hilbert curve from Section~\ref{sec:hilbert}, or a Peano curve, a coil curve, or a half-coil curve from Section~\ref{sec:ternary}), then $f$ will be an order-preserving mono-curve. For example, if $f'$ is a standard Hilbert curve, then (using the notation from the proof of Lemma~\ref{lem:compositesymmetric}) a reflection $m^{h_{b\cdot d'}}_k(r)$ leads to a reflection in axis $\inv{a}^{f'}_0(c^{h_{b\cdot d'}}_k(r))$. Thus we get:\[\begin{array}{lll}
c^f_{kd'+i}(r) & = c^{f'}_i(c^{h_{2\cdot d'}}_k(r)) \\
a^f_{kd'+i}(r) & = a^{f'}_i(c^{h_{2\cdot d'}}_k(r)) + d' \cdot a^{h_{2\cdot d'}}_k(r) & \\
m^f_{kd'+i}(r) & = m^{f'}_i(c^{h_{2\cdot d'}}_k(r)) & \hbox{if $m^{h_{2\cdot d'}}_k(r) = 0$ or $a^{f'}_i(c^{h_{2\cdot d'}}_k(r)) \neq 0$} \\
m^f_{kd'+i}(r) & = 1 - m^{f'}_i(c^{h_{2\cdot d'}}_k(r)) & \hbox{if $m^{h_{2\cdot d'}}_k(r) = 1$ and $a^{f'}_i(c^{h_{2\cdot d'}}_k(r)) = 0$}
\end{array}\]
Table~\ref{tab:hilberthilbertcomposition} shows the resulting curve definition if $f'$ is the two-dimensional Hilbert curve, and thus, the resulting curve $f$ is $h_2 \circ h_2$. Note that $f$ is not a standard Hilbert curve as defined in Section~\ref{sec:standard}, since the locations of the subregions do not conform to the binary reflected Gray code.

\begin{table}
\caption{The four-dimensional curve $h_2 \circ h_2$, created from the composition of two-dimensional Hilbert curves. This curve is \halffix- and diagonal-extradimensional to the two-dimensional Hilbert curve.}
\label{tab:hilberthilbertcomposition}
\addvspace{.5\baselineskip}\centering\footnotesize
\begin{tabular}{r|cccc||r|cccc}
rank & loc. & perm. & refl. & exit & rank & loc. & perm. & refl. & exit \\
\hline
0000 & 0000 & 1032 & 0000 & $\frac12$(0,1,0,0) & 1000 & 1111 & 2301 & 0000 & $\frac12(1,1,2,1)$ \\
0001 & 0100 & 2310 & 0000 & $\frac12$(0,1,0,1) & 1001 & 1110 & 0132 & 0011 & $\frac12(2,1,2,1)$ \\
0010 & 0101 & 2301 & 0000 & $\frac12$(0,1,1,1) & 1010 & 1010 & 1032 & 1111 & $\frac12(2,0,2,1)$ \\
0011 & 0001 & 1023 & 0110 & $\frac12$(0,0,1,1) & 1011 & 1011 & 3201 & 1010 & $\frac12(2,0,1,1)$ \\
0100 & 0011 & 3201 & 0000 & $\frac12$(0,0,2,1) & 1100 & 1001 & 1023 & 1010 & $\frac12(2,1,1,1)$ \\
0101 & 0010 & 1032 & 0011 & $\frac12$(0,1,2,1) & 1101 & 1101 & 2301 & 1010 & $\frac12(2,1,0,1)$ \\
0110 & 0110 & 0132 & 0011 & $\frac12$(1,1,2,1) & 1110 & 1100 & 2310 & 1001 & $\frac12(2,1,0,0)$ \\
0111 & 0111 & 2301 & 1010 & $\frac12$(1,1,1,1) & 1111 & 1000 & 1032 & 1100 & $\frac12(2,0,0,0)$ \\
\end{tabular}
\end{table}

\begin{lemma}
\label{lem:2regPrefixExtradimensional}
The curve $f$ is \halffix-extradimensional to $f'$.
\end{lemma}
\begin{proof}
Let $\inv{f'}$ be any valid ordering of $f$. By Definition~\ref{def:halfextradimensional}, we have to prove that there is a valid ordering $\inv{f}$ of $f$ such that for any pair of points $a, b \in U'$ we have $\inv{f'}(a) < \inv{f'}(b)$ if and only if $\inv{f}(a|0) < \inv{f}(b|0)$.

Let $\inv{h_b}$ be a valid ordering of $h_b$ which is monotone from $(0,0)$ to $(1,0)$ (such an ordering exists by Lemma~\ref{lem:HilbertIsEdgeMonotone}). We now define $\inv{f}$ by $\inv{f}(p) = \inv{h_b}(\inv{f'}(\lefthalf{p}),\inv{f'}(\righthalf{p}))$.
Recall that $f'$ is a $b$-regular curve with the entrance gate in the origin, and therefore $\inv{f'}(0)  = 0$. Thus, we have
$\inv{f}(a|0) = \inv{h_b}((\inv{f'}(a),\inv{f'}(0))) = \inv{h_b}((\inv{f'}(a),0))$ and
$\inv{f}(b|0) = \inv{h_b}((\inv{f'}(b),\inv{f'}(0))) = \inv{h_b}((\inv{f'}(b),0))$.
By the monotonicity of $\inv{h_b}$ from $(0,0)$ to $(1,0)$, it follows that
$\inv{f}(a|0) < \inv{f}(b|0)$ if and only if $\inv{f'}(a) < \inv{f'}(b)$, QED.
\end{proof}

\begin{lemma}
\label{lem:2regDiagonallyExtradimensional}
The curve $f$ is diagonal-extradimensional to $f'$.
\end{lemma}
\begin{proof}
Let $\inv{f'}$ be any valid ordering of $f$. By Definition~\ref{def:diagextradimensional}, we have to prove that there is a valid ordering $\inv{f}$ of $f$ such that for any pair of points $a, b \in U'$ we have $\inv{f'}(a) < \inv{f'}(b)$ if and only if $\inv{f}(a|a) < \inv{f}(b|b)$.

Let $\inv{h_b}$ be a valid ordering of $h_b$ which is monotone from $(0,0)$ to $(1,1)$ (such an ordering exists by Lemma~\ref{lem:HilbertIsDiagMonotone}). We now define $\inv{f}$ by $\inv{f}(p) = \inv{h_b}(\inv{f'}(\lefthalf{p}),\inv{f'}(\righthalf{p}))$. Thus, we have
$\inv{f}(a|a) = \inv{h_b}((\inv{f'}(a),\inv{f'}(a)))$ and
$\inv{f}(b|b) = \inv{h_b}((\inv{f'}(b),\inv{f'}(b)))$.
By the monotonicity of $\inv{h_b}$ from $(0,0)$ to $(1,1)$, it follows that
$\inv{f}(a|a) < \inv{f}(b|b)$ if and only if $\inv{f'}(a) < \inv{f'}(b)$, QED.
\end{proof}

\subsection{Implementing a \halffix- and diagonal-extradimensional curve}
\label{sec:implementhalffix}

In Section~\ref{sec:halffixconstruction} we saw how to one can construct a \halffix- and diagonal-extradimensional curve $f$, given a $\base$-regular order-preserving vertex-continuous mono-curve $f'$ that satisfies certain conditions. We also saw that if $f'$ is symmetric, then $f$ is an order-preserving mono-curve, which could be implemented using Algorithm~\ref{alg:generaloperator}.

We will now see another method to implement a comparison operator for $f$, which can also be applied if $f'$ is not symmetric and which may also be more efficient. We illustrate this method by implementing \halffix- and diagonal-extradimensional curves based on the Butz-Moore generalization of Hilbert curves. Our goal is to compare points according to a valid ordering $\inv{f}$ of $f$. Considering the construction of $f$, we can implement this by computing $\inv{f}(p)$ as $\inv{h_2}(\inv{f'}(\lefthalf{p}),\inv{f'}(\righthalf{p}))$. To do so, we rewrite Algorithm~\ref{alg:binarycurvesoperator} as a comparator object for Butz-Moore curves that is initialized with the coordinates of two $d'$-dimensional points $p'$ and~$q'$, and provides a method to compute $\inv{f'}(p')$ and $\inv{f'}(q')$ one digit a time. A comparison operator for $\inv{f}$ is then implemented by initializing two comparator objects, one with $\lefthalf{p}$ and~$\lefthalf{q}$, and one with $\righthalf{p}$ and~$\righthalf{q}$, and then simply computing $\inv{h_2}$, extracting digits from the comparator objects rather than directly from the coordinates of $p$ and $q$.

\SetKwData{Localp}{pLocal}
\SetKwData{Localq}{qLocal}
\SetKw{Empty}{empty}
\LinesNumbered
\RestyleAlgo{ruled}
The details are as follows. Each object has member variables \Forward, \Rotation, \Reflected, $i$, \hbox{\iMinusOne}, \NewRotation, \Localp and \Localq. To initialize an object, we run Algorithm~\ref{alg:initobject}, calling it with two $d'$-dimensional points $p$ and $q$; to extract a rank digit (one for $p$ and one for $q$), we run Algorithm~\ref{alg:extractdigit}. The extraction method works correctly up to and including the first time the rank digit differs between $p$ and~$q$ (digits extracted afterwards may be wrong). The algorithm for the $2d'$-dimensional curve is given as Algorithm~\ref{alg:diagcomparator}.

\begin{algorithm}
$\Localp \leftarrow p$; $\Localq \leftarrow q$;
$\Forward \leftarrow \True$; $\Rotation \leftarrow 0$;
\lFor{$i \leftarrow 0$ \KwTo $d'-1$}{$\Reflected[i] \leftarrow \False$}\;
$i \leftarrow \Rotation$; $\iMinusOne \leftarrow d'-1$;
$\NewRotation \leftarrow \iMinusOne$\;
\caption{Initialization of a comparator object for a Butz-Moore curve\label{alg:initobject}}
\end{algorithm}

\begin{algorithm}
\SetKwData{pRank}{pRank}
\SetKwData{qRank}{qRank}
$\pFirstDigit \leftarrow \Extract(\Localp[i])$; $\qFirstDigit \leftarrow \Extract(\Localq[i])$\;
\If{$\Reflected[i]$}{
  $\pFirstDigit \leftarrow (1 - \pFirstDigit)$; $\qFirstDigit \leftarrow (1 - \qFirstDigit)$
}
\If{\Forward}{
  $\pRank \leftarrow \pFirstDigit$; $\qRank \leftarrow \qFirstDigit$
}
\Else{
  $\pRank \leftarrow (1 - \pFirstDigit)$; $\qRank \leftarrow (1 - \qFirstDigit)$
}
\If{$\pFirstDigit = 1$}{
  $\Forward \leftarrow \Not \Forward$;
  $\Reflected[i] \leftarrow \Not \Reflected[i]$;
  $\NewRotation \leftarrow \iMinusOne$\;
}
$\iMinusOne \leftarrow i$; $i \leftarrow (i + 1) \bmod d'$\;
\If{$i = \Rotation$}{
  $\Reflected[\iMinusOne] \leftarrow \Not \Reflected[\iMinusOne]$\;
  $\Reflected[\NewRotation] \leftarrow \Not \Reflected[\NewRotation]$\;
  $\Rotation \leftarrow \NewRotation$\;
  $i \leftarrow \Rotation$; $\iMinusOne \leftarrow (i - 1) \bmod d'$;
  $\NewRotation \leftarrow \iMinusOne$\;
}
\Return $(\pRank, \qRank)$
\caption{Extracting a single rank digit from a comparator object for a Butz-Moore curve\label{alg:extractdigit}}
\end{algorithm}

\begin{algorithm}
\SetKwData{Compare}{compare}
initialize a comparator object $\Compare[0]$ with points $\lefthalf{p}$ and $\lefthalf{q}$\;
initialize a comparator object $\Compare[1]$ with points $\righthalf{p}$ and $\righthalf{q}$\;
$\Forward \leftarrow \True$; $\Rotation \leftarrow 0$;
\lFor{$i \leftarrow 0$ \KwTo $1$}{$\Reflected[i] \leftarrow \False$}\;
\Repeat{all remaining digit strings $p[0],...,p[2d'-1]$ and $q[0],...,q[2d'-1]$ are empty}{
  $i \leftarrow \Rotation$; $\iMinusOne \leftarrow (i - 1) \bmod 2$;
  $\NewRotation \leftarrow \iMinusOne$\;
  \Repeat{$i = \Rotation$}{
    extract $(\pFirstDigit,\qFirstDigit)$ from $\Compare[i]$\;
    \If{$\Reflected[i]$}{
      $\pFirstDigit \leftarrow (1 - \pFirstDigit)$; $\qFirstDigit \leftarrow (1 - \qFirstDigit)$\label{alg:binarycurves:reflect}
    }
    \lIf{$\pFirstDigit < \qFirstDigit$}{\Return \Forward};
    \lElseIf{$\pFirstDigit > \qFirstDigit$}{\Return \Not \Forward}\;
    \If{$\pFirstDigit = 1$}{
      $\Forward \leftarrow \Not \Forward$;
      $\Reflected[i] \leftarrow \Not \Reflected[i]$;
      $\NewRotation \leftarrow \iMinusOne$\;
    }
    $\iMinusOne \leftarrow i$; $i \leftarrow (i + 1) \bmod 2$\;
  }
  $\Reflected[\iMinusOne] \leftarrow \Not \Reflected[\iMinusOne]$\;
  $\Reflected[\NewRotation] \leftarrow \Not \Reflected[\NewRotation]$\;
  $\Rotation \leftarrow \NewRotation$\;
}
\Return\False
\caption{Algorithm to compare the position of two points $p$ and $q$ along a $2d'$-dimensional curve that is \halffix- and diagonal-extradimensional to a $d'$-dimensional Butz-Moore curve.\label{alg:diagcomparator}}
\end{algorithm}

\section{Remaining questions for theory and practice}
\label{sec:discussion}

\paragraph{Building R-trees}
As explained in Inset~\ref{box:rtrees}, the curves in this paper could be used to build R-trees in higher dimensions (for example, using six-dimensional curves for three-dimensional boxes), or to experiment with 3-regular curves instead of 2-regular curves for this application. Some first, rough experiments, indicate that it may depend a lot on the distribution of the data whether the new curves offer any advantages over previously known curves---but it seems unlikely that the new curves have serious disadvantages in this application. Further experiments would need to be done.

\paragraph{Fast implementations}
To get most out of the space-filling curves presented in this paper, one may want to design highly-optimized implementations of comparison operators. As compared to the 3-regular curves, the 2-regular curves have the advantage that many calculations can be implemented by means of bit shifting and/or exclusive-or operations on many bits at once. A possible advantage of interdimensionally consistent curves, is that a highly-optimized (possibly hard-wired) implementation of a $d$-dimensional curve, for some fixed value of $d$, can also be used as an implementation of each $d'$-dimensional curve with $d' \leq d$. If the cost of working with more dimensions than necessary is mitigated by parallel circuitry, then, being able to use a single hard-wired $d$-dimensional implementation for any $d' \leq d$, may reduce overhead.

\paragraph{Proving that the harmonious Hilbert curves are interdimensionally consistent}
At this point, I would like to pose the main result of this paper as an open question. In Section~\ref{sec:proofhilbert}, I proved that the harmonious Hilbert curves consitute an interdimensionally consistent set. However, I believe the proof is hard to understand and to verify. In particular, I do not find the current proof more convincing than several earlier attempts at a proof, which I only dismissed as erroneous because they led to the conclusion that the harmonious Hilbert curves could not be interdimensionally consistent. Indeed, these earlier attempts had to be wrong, or my exhaustive search algorithm to find the curves and the algorithms I used to test the implementations all had to be wrong. I implemented comparison operators for all curves from Section \ref{sec:ternary} and~\ref{sec:hilbert} for $d \in \{2,3,4,5\}$ and verified interdimensional consistency with respect to the ordering of subregions down to at least the third level of refinement. Although now, everything seems to fall into place and I am convinced that my implementations and Section~\ref{sec:proofhilbert} do not contain any mistakes that cannot be repaired, it would still be good if we could find a simpler proof.

\paragraph{Clearing up interdimensional consistency under tie-breaking}
Tie-breaking is a subtle issue which was touched on in Sections \ref{sec:intro} and \ref{sec:definitions}, but not resolved completely. Space-filling curves are surjective, but not bijective---in fact, they cannot be bijective. Often this issue is dealt with by considering the level-$\ell$ approximation of the space-filling curve---say, the Hilbert curve $h_2$, for concreteness---which is the polygonal curve whose vertices are the lower left corners of the $4^\ell$ successive squares at the $\ell$-th level of refinement. Clearly, for any $\ell$, the level-$\ell$ approximation visits each point of the unit square at most once, and in the limit for $\ell \rightarrow \infty$, the approximation visits all points of the unit square. Unfortunately, this does not imply that each point is visited only once: this approach offers no way around the fact that, for example, $h_2(1/6) = h_2(1/2)$. This can be seen by writing out $h_2(1/6)$ and $h_2(1/2)$ in binary notation:
$h_2(1/110) = h_2(0.00101010...) = (0.0111...,0.0111...) = (0.1,0.1)$, and $h_2(1/10) = h_2(0.1) = (0.1,0.1)$. This makes the inverse of $h_2$ ill-defined.

In the definitions and proofs in this paper, we dealt with the ambiguity as follows. We allow any `inverse' $\inv{h_2}$ such that for all points $p$ in the unit hypercube, we have $h_2(\inv{h_2}(p)) = p$; and we prove that no matter which inverse $\inv{h_2}$ we choose, we can choose at least one valid inverse for each higher-dimensional harmonious Hilbert curve so that we get an interdimensionally consistent set.

However, our implementations of comparison operators are essentially based on level-$\ell$ approximations. The idea is that this implements the tie-breaking method that always assigns each point $p$ to the quadrant that lies to its upper right. In fact, this approach would not work correctly if we would accept coordinates that end with an infinite sequence $111...$: it would put the points $(0.0111,0.0111)$, $(0,1)$ and $(0.1,0.1)$ in this order, even though the first and the last point are the same. Our algorithms do work correctly as long as the coordinates of each point $p$ are represented as binary numbers with a \emph{finite} number of digits. Recall that our algorithms do not accept the number~1 for a coordinate.

The above observations imply that the proofs of interdimensional consistency for the curves as such, are not necessarily valid for the curves with the specific choice of inverses as implemented in our algorithms. The fact that our implementations of mono-Wunderlich curves are interdimensionally consistent could still be established easily by inspecting the algorithms (see Section~\ref{sec:implementternary}). For harmonious Hilbert curves the situation is less clear. The algorithm for harmonious Hilbert curves is much harder to analyse. The proof that harmonious Hilbert curves are interdimensionally consistent relies on the proper handling of points with coordinates equal to 1, but the implementation does not accept the number 1 and has not been verified to work correctly when given the alternative representation 0.111... of the same number. Although everything seems to indicate that this is ultimately just a matter of presentation, it is one that would be good to resolve.

\paragraph{Interdimensionally consistent sets of mono-Hilbert curves with interior origins?}
In Section~\ref{sec:implementhilbert} I had to give up on extending the harmonious Hilbert curves to filling the full $d$-dimensional space (rather than only the space of points with non-negative coordinates) while retaining interdimensional consistency. The reason was that I could not find a way to move the origin into the interior of the $d$-dimensional unit hypercube, while maintaining that the $d$-dimensional harmonious Hilbert curve shows the $(d-1)$-dimensional curve on the intersection of the unit hypercube with any axis-parallel $(d-1)$-dimensional plane through the origin. I conjecture that this is, in fact, impossible.
One could also try to find out if an interdimensionally consistent set of 2-regular vertex-continuous space-filling curves---not necessarily mono-curves---could exist.

\paragraph{How unique are sets of interdimensionally consistent curves?}
The fact that the three-, four-, five-, six- and seven-dimensional harmonious Hilbert curves emerged as the unique results of an automated search, leads me to conjecture that the harmonious Hilbert curves actually constitute the \emph{only} interdimensionally consistent set of mono-Hilbert curves. The catch is that, for $d$-dimensional curves with $d \geq 4$, the search was restricted to symmetric curves in which the subregions are traversed in the order of binary reflected Gray codes. Although this seems to be the most natural choice, and although this is indeed the only solution for $d = 2$ and $d = 3$, I cannot be sure that no alternative solutions for $d \geq 4$ were ruled out by these restrictions.

While I believe that the mono-Hilbert curves may contain only one interdimensionally consistent set, we already saw that the mono-Wunderlich curves contain at least four interdimensionally consistent sets. There must be many more, if only because the subregion in the centre can be rotated freely without affecting the curves shown on the faces of the unit hypercube. One could try to find a full characterization of all interdimensionally consistent sets of mono-Wunderlich curves---possibly with the added condition that the back faces also show the $(d-1)$-dimensional curves.

\subsection*{Acknowledgements}
I thank \emph{Freek van Walderveen} for finding the first efficient algorithm to generate the permutations of the $d$-dimensional harmonious Hilbert curve, thus providing a strong indication that the work presented in this paper would indeed be possible. The idea of composing space-filling curves to produce \halffix- and diagonal-extradimensional curves came to me while analysing the implications of a conjecture by \emph{Arie Bos}.

\bibliographystyle{abbrv}

\end{document}